\pdfoutput=1
\documentclass[pdflatex,sn-mathphys-num]{sn-jnl}

\usepackage{graphicx}%
\usepackage{multirow}%
\usepackage{amsmath,amssymb,amsfonts}%
\numberwithin{equation}{section}
\usepackage{amsthm}%
\usepackage{mathrsfs}%
\usepackage[title]{appendix}%
\usepackage{xcolor}%
\usepackage{textcomp}%
\usepackage{manyfoot}%
\usepackage{booktabs}%
\usepackage{algorithm}%
\usepackage{algorithmicx}%
\usepackage{algpseudocode}%
\usepackage{listings}%
\usepackage{bm}
\usepackage{newtxtext}
\usepackage{natbib}
\usepackage{setspace}
\usepackage{ulem} 
\usepackage{accents} 
\theoremstyle{thmstyleone}%
\newtheorem{theorem}{Theorem}
\theoremstyle{thmstyletwo}%
\newtheorem{remark}{Remark}%

\theoremstyle{thmstylethree}%

\begin{document}

\title[Article Title]{The intrinsic decomposition of vorticity dynamics on an arbitrarily moving and deforming boundary}


\author*[1]{\fnm{Tao} \sur{Chen}}\email{chentao2023@njust.edu.cn}

\affil*[1]{\orgdiv{School of Physics}, \orgname{Nanjing University of Science and Technology}, \orgaddress{ \city{Nanjing}, \postcode{210094}, \country{China}}}


\abstract{Boundary vorticity dynamics provides a rigorous theoretical foundation for understanding vorticity creation at boundaries, vorticity-boundary interactions, as well as the rational design of effective boundary flow control strategies. It cornerstone is the boundary vorticity flux (BVF), first introduced by Lighthill in 1963, which quantities the local rate of vorticity production at a boundary, and thereby serves as a mathematical measure of distributed vorticity source strength. Owing to the non-uniqueness of the vorticity current tensor, two alternative definitions of BVF naturally exist, leading to the Lighthill-Panton-Wu and Lyman-Huggins interpretations of vorticity dynamics. Recent advances in vorticity and vortex kinematics have revealed two elementary modes embedded in vorticity: the rigid-rotation mode and the spin mode, extending beyond the classical single-vorticity paradigm in fluid mechanics. For each interpretation, by adopting a differential-geometric approach, we develop a general theory of the intrinsic decomposition of BVF for compressible Newtonian fluid interacting with an arbitrarily moving and deforming boundary surface. The analyses are further extended to the decomposition of boundary enstrophy dynamics, centered on the boundary enstrophy flux (BEF). Beyond the existing literature, the new theory explicitly identifies a complete set of boundary sources for the rigid-rotation and spin modes, as well as for various enstrophy constituents, arising from the interplay among external force, surface geometry and kinematics, and both longitudinal and transverse physical processes on a deformable boundary. It is noteworthy that introducing a conjugate curvature tensor pair consistently yields compact mathematical representations for all source terms, manifesting as bilinear (or quadratic-form-type) couplings between fundamental vortcity modes and the surface curvature tensors, irrespective of the complexity or generality of the boundary kinematics.}




\maketitle
\tableofcontents
\section{Introduction}\label{sec1}
Vorticity, defined as the curl of the velocity field $(\bm{\omega}\equiv\bm{\nabla}\times\bm{u})$, is a fundamental derived quantity in fluid mechanics~\citep{Helmholtz1858,Truesdell1954,Batchelor1967}. Representing twice the local
angular velocity of a fluid volume element, it has been extensively used to investigate coherent vortex structures, which are regarded as the sinews and muscles of complex fluid motions~\citep{Kuchemann1965report,Saffman1992Vortex}. In contrast to the primitive variables (velocity and pressure), vorticity-based approaches usually provide a more intuitive physical interpretation and deeper insight into the generation, interaction, and evolution of complex vortical structures within a flow field~\citep{Wu2006vorticity,Zhou2024}. At the fundamental level, \textit{boundary vorticity dynamics} is a crucial subject of fluid mechanics that investigates the physical mechanisms underlying vorticity creation at solid boundaries or fluid-fluid interfaces, as well as the ensuing interactions between vorticity and boundaries~\citep{WuWu1996,Wu2006vorticity}. The theory of boundary vorticity dynamics proves valuable in complex flow diagnosis, configuration design, and the development of adaptive boundary flow control methods~\citep{Wu2018review}.

It has long been recognized that the most important source of vorticity is the surface of a solid boundary~\citep{Lamb1932hydrodynamics,Batchelor1967}. However, the concept of vorticity creation at a boundary had not crystallized into a formal theory until the seminal work of~\citet{Lighthill1963}, who laid the first cornerstone of boundary vorticity dynamics. For a two-dimensional (2D) incompressible viscous flow interacting with a stationary flat rigid wall,~\citet{Lighthill1963} introduced the \textit{boundary vorticity flux} (BVF), ${\sigma}\equiv\bm{n}\bm{\cdot}\bm{\nabla}{\omega}=\nu\partial_{n}{\omega}$, to quantitatively characterize the vorticity source strength on a no-slip wall. Physically, it describes the local vorticity creation rate per unit area on a boundary. Here, $\nu$ is the kinematic viscosity, $\omega$ is the vorticity component perpendicular to the 2D plane, and $\bm{n}$ is the unit normal vector on the boundary oriented into the fluid. By applying the Navier-Stokes (NS) equations at the wall and invoking the no-slip velocity
adherence condition, he derived an exact on-wall relation revealing boundary vorticity creation by
the surface pressure gradient at a rate measured by the BVF. He demonstrated that the vorticity-based considerations illuminate the detailed development of the boundary layer just as clearly as do momentum-based analyses. This in fact embedded the whole boundary-layer theory within the broader context of vorticity dynamics. Despite the fundamental nature of this insight, the importance of vorticity creation and Lighthill's pioneering work did not receive due recognition until the late 1970s~\citep{WuWu1996}. Lighthill himself later reaffirmed the central idea in a survey paper on waves and hydrodynamic loading~\citep{lighthill1979waves}, as well as in his textbook on theoretical fluid mechanics~\citep{lighthill1986informal}.

The definition of Lighthill's BVF was subsequently generalized by~\citet{Panton1984} to three-dimensional (3D) scenarios in his monograph on incompressible flows, yielding the expression $\bm{\sigma}_{\bm\omega}^{(1)}\equiv\nu\bm{n}\bm{\cdot}\bm{\nabla\omega}=\nu\partial_{n}\bm{\omega}$ for a flat wall. Concurrently, the wall tangential acceleration was first identified by~\citet{Morton1984} as a constituent of the vorticity source. Immediately afterward, the physics of forced, unsteady, separated flows was explained in a unified way, based on the BVF $\bm{\sigma}_{\bm\omega}^{(1)}$~\citep{reynolds1985review}.
The general expression of $\bm{\sigma}_{\bm\omega}^{(1)}$ was later derived by~\citet{Wu1986} for 3D incompressible viscous flow over an arbitrarily curved stationary wall, revealing an additional 3D viscous contribution associated with the interplay between vorticity and surface curvature tensor. This 3D effect had been briefly noted by~\citet{Lighthill1963} as a small term, though without further elaboration; exceptions, however, arise near sharp edges or spiral points of the skin-friction lines~\citep{Wu1986}. The full expression of $\bm{\sigma}_{\bm\omega}^{(1)}$ was eventually derived by~\citet{Wu1993} for 3D compressible viscous flow interacting with an arbitrarily moving and deforming
wall, and later generalized to an interface separating two immiscible fluids within a sharp-interface framework~\citep{Wu1995,WuWu1996}. For these historical reasons, we refer to $\bm{\sigma}_{\bm\omega}^{(1)}$ as the \textit{Lighthill-Panton-Wu BVF}. A similar extension applicable to any continuous media was independently made by~\citet{Hornung1989,Hornung1990}, based on the Stokes-Helmholtz decomposition of the divergence of the total stress tensor. Based on~\citet{Wu1993}, a new representation of $\bm{\sigma}_{\bm\omega}^{(1)}$ was proposed by~\citet{Xie2013688}, which partitions the contributions from the tangential and normal vorticity components on a boundary. This formulation is directly associated with the intrinsic generalized Stokes formula of the second kind~\citep{Xie2013SCPMA}.
By introducing the nonholonomic basis theory, the kinematic expression of $\bm{\sigma}_{\bm\omega}^{(1)}$ was derived in an orthogonal
coordinate system aligned with the surface principal directions~\citep{Xie2020theory}. Representative applications of $\bm{\sigma}_{\bm\omega}^{(1)}$-based theory include duct flows with deformable boundaries~\citep{Chen2016Vorticity,ChenYu2017}, wake vortex suppression by traveling wave on a flexible surface~\citep{Xin2013Vorticity,Wu2007Suppression}, turbulent drag reduction by traveling wave of flexible wall~\citep{Zhao2004Turbulent}, dynamic stall delay induced by airfoil pitching-up
motion~\citep{Zou2017Causal}, and geometric optimization in a model pump-turbine~\citep{Lu2020BoundaryVorticity}, and so forth.

An alternative definition of BVF, $\bm{\sigma}_{\bm\omega}^{(2)}\equiv-\nu\bm{n}\times(\bm{\nabla}\times\bm{\omega})$, was proposed by~\citet{Lyman1990vorticity} (referred to as the \textit{Lyman vorticity flux}), which directly incorporates
the purely 3D viscous contribution into the definition. While the integrals of $\bm{\sigma}_{\bm\omega}^{(1)}$ and $\bm{\sigma}_{\bm\omega}^{(2)}$ over any closed surface are identical, the two definitions differ in their description of the local vorticity creation rates on a boundary. These two BVFs were shown to arise from distinct definitions
of the viscous vorticity current tensor in the vorticity evolution equation~\citep{Terrington2023LH,chen2024physics}. Consequently,
they offer two different yet equally valid dynamical interpretations of vorticity dynamics, which are referred to as the Lighthill-Panton-Wu and Lyman-Huggins interpretations~\citep{Terrington2023LH}. The non-uniqueness of BVF was also discussed by~\citet{Kolar2003} in terms of general and effective vorticity fluxes, and the differences between the two formulations were further clarified by~\citet{Chen2021features}. Several conceptual advantages of $\bm{\sigma}_{\bm\omega}^{(2)}$ have been emphasized in the context of vortex reconnection~\citep{Terrington2021Lymanflux} and vortex ring attachment to a free surface~\citep{Terrington2022JFMfree}. Most recently, the Lighthill-Panton-Wu BVF $\bm{\sigma}_{\bm\omega}^{(1)}$ was proved to be the uniquely correct choice for use as the Neumann boundary condition in both adjoint and stochastic representations of Lagrangian vorticity dynamics~\citep{Wang2022Origin,Du2025Back}, whereas the Lyman flux $\bm{\sigma}_{\bm\omega}^{(2)}$ was shown to be generally related to the Huggins vorticity current tensor~\citep{huggins1970exact,huggins1971dynamical,huggins1980vortex} and the Josephson-Anderson relation~\citep{Eyink2008Turbulent,Eyink2021,DuZaki2025}. Due to these historical reasons, we refer to $\bm{\sigma}_{\bm\omega}^{(2)}$ as the \textit{Lyman-Huggins BVF}.

Enstrophy, defined as $\Omega\equiv\omega^2/2$, quantifies the rotational intensity fo a flow without accounting for the direction of vorticity. Being directly associated with the BVF, the boundary enstrophy flux (BEF) $F_{\Omega}\equiv\bm{\omega}\bm{\cdot}\bm{\sigma}_{\bm\omega}$ measures the rate of enstrophy creation on a boundary~\citep{Wu1993,WuWu1996}. The BEF indicates the boundary sources and sinks of an enstrophy field and is notably linked to complex topological features such as isolated critical points and sepration/attachment lines in a skin-friction field~\citep{Chen2021turbulence,liu2023reconstruction}. The evolution equation for the BEF, derived by~\citet{Chen2021features}, was recently applied to analyze sheet/cloud cavitation over a hydrofoil~\citep{WangHaobo2026}. From multiple perspectives, an explicit intrinsic coupling relation has been derived between the BEF and the elementary surface physical quantities (which include skin friction, surface vorticity, surface pressure, and surface curvature tensor)~\citep{LiuTS2016,Chen2019BEF}, which indicates that the BEF can be generated by the viscous coupling between skin friction and surface pressure gradient, as well as quadratic vorticity-curvature interaction. In fact, skin friction and surface pressure have been regarded as the footprints of near-wall coherent structures~\citep{wu1988zfw,bewley2004footprint,liu2018skinfriction,Chen2021turbulence}. The application of this relation also leads to a new aerodynamic lift formula~\citep{liu2017aiaatn}. Based on the variational formulation of this equation~\citep{LiuShen2008}, global skin friction and surface pressure fields with high resolution can be mutually inferred using the conventional measurement techniques in experiments: the global luminescent oil-film (GLOF) method for skin friction~\citep{LiuTS2019PAS} and pressure-sensitive paint (PSP) for surface pressure~\citep{Liu2021PSP}. A similar approach extracts skin friction from surface temperature measurements via temperature-sensitive paint (TSP)~\citep{miozzi2016global,Miozzi2019Skin}. Recent theoretical studies have revealed that the BEF generates the boundary dilatation flux unidirectionally~\citep{Mao2022unified,ChenTao2022TAML}, by which wall sound sources were identified in flow past a cylinder as an application~\citep{ZhaoYH2025}. Corresponding to the two BVFs, there exist two types of BEF which we refer to as \textit{the Lighthill-Panton-Wu and Lyman-Huggins BEFs}~\citep{Terrington2023LH}, denoted by $F_{\Omega}^{(1)}$ and $F_{\Omega}^{(2)}$. New explorations on BVF and BEF include interfacial vorticity dynamics in the Navier-Stokes-Korteweg system with a diffuse interface~\citep{Chen2022NSK,Chen2024IJMF}, the generalized interface circulation dynamics within a sharp-interface framework~\citep{Terrington2021Lymanflux,Terrington2022JFM}, as well as the
Lie derivatives of fundamental surface quantities with respect to the celerity of near-wall
coherent structures~\citep{Chen2023Lie}.

Moreover, from the perspective of vortex force theory in aerodynamics, analyzing the Lamb vector $\bm{L}\equiv\bm{\omega}\times\bm{u}$ and its spatial divergence $\vartheta_{L}\equiv\bm{\nabla}\bm{\cdot}\bm{L}$ (i.e., the Lamb dilatation) could provide valuable insights into the origins of aerodynamic
forces exerted on an immersed body~\citep{wang2013pf,Wu2018review,Chen2023PhysicaD}. The Lamb dilatation appears as an acoustic source in Lighthill's wave equation~\citep{Lighthill1952}, and in vortex-sound formulations~\citep{Powell1964,Howe1975}. One one hand, the near-wall Taylor-series expansion solution for $\bm{L}$ was derived to elucidate the initial formation of the Lamb vector in the viscous sublayer, where the coefficients at successive orders are expressed using the elementary surface quantities including the BEF~\citep{Yang2007Steady,Chen2022AIP}. This formal expansion has been shown to demonstrate a high degree of agreement with the simulation data of supersonic flat-plate flow~\citep{Xue2025}. One the other hand, analogous to the BEF,~\citet{Chen2023PhysicaD} introduced the concept of the boundary Lamb dilatation flux (BLDF) and discussed its relationship with near-wall vortical structures based on its evolution equation. Furthermore, the acoustic source term for the improved Phillips equation for longitudinal field is identified as $-2Q$~\citep{Phillips1960,Mao2022AMS}, where $Q$ denotes the second principal invariant of the velocity gradient tensor (VGT)~\citep{Hunt1988Eddies,Jeong1995}. The boundary $Q$ flux and its decomposition were proposed for generic compressible viscous flow interacting with a stationary wall~\citep{Chen2024Qflux}. The results indicate that the coupling between skin friction and surface pressure gradient does not generate the boundary $Q$ flux due to the equal contributions from boundary enstrophy and squared strain-rate fluxes, whereas the quadratic vorticity-curvature interaction term does contribute. These findings provide a rational foundation for modeling VGT invariant dynamics in wall-bounded flows. For a recent review of the theory and applications of boundary vorticity dynamics, the reader is referred to~\cite{Mao2023Solid}.

The research on boundary vorticity dynamics discussed above is all based on a single measure of vorticity within the classical framework of fluid mechanics. Recent theoretical studies, however, have revealed two elementary modes embedded within the vorticity: the rigid-rotation mode (swirling) and the spin mode (shear). Relevant investigations on distinguishing vorticity modes are rather limited and can be grouped into two categories. The first category belongs to the characteristic algebraic description, which primarily applies the theory of Schur form and normal-nilpotent decomposition to the VGT~\citep{Li2010Theoretical,Liu2018Rortex,Mao2023vortex}. This approach has yielded two invariant vorticity decompositions within the characteristic framework: one is the Liutex-shear vorticity decomposition with a positive spin mode~\citep{Liu2018Rortex}, the other features a negative spin mode~\citep{Chen2026General}. However, the eigenvalue-based theories impose limitations on further theoretical analysis of vorticity dynamics. The second category pertains to material and field descriptions, focusing on the rotational kinematics of directed material elements with more degrees of freedom beyond the coverage of vorticity measure~\citep{Chen2025Kinematic,Chen2026General}. Subsequently, a unifying principle proposed by~\citet{Chen2026operator} links these distinct vorticity modes to two communicative vorticity/vortex operators within a novel operator-form vorticity decomposition framework. These studies provide useful tools for characterizing the formation, interaction and mutual transformation of axial vortices and shear layers in complex flows. Moving beyond pure vorticity kinematics, the present study aims to develop a general theory for the intrinsic decomposition of vorticity dynamics on an arbitrarily moving and deforming boundary. For the first time, the boundary sources of vorticty modes and their associated enstrophy constituents are explicitly expressed in terms of surface physical and geometric quantities under both Lighthill-Panton-Wu and Lyman-Huggins interpretations. To the best of the authors' knowledge, vorticity dynamics theory on a deformable boundary remains very limited to date, although numerous numerical simulations and experimental studies exist.

The remainder of this paper is organized as follows. Section~\ref{surface_geometry} describes the surface geometry and defines the associated differential gradient operators. Section~\ref{VGT_vorticity} introduces decompositions of vorticity and enstrophy, and elucidate mathematical structures of surface velocity gradients. In Section~\ref{xxxx1}, we propose several general theorems pertinent to fluid kinematics on a deformable boundary. Sections~\ref{def_expre_BVFs} and~\ref{kkk1} introduce two alternative definitions of BVF and BEF, followed by derivations of their complete expressions. In~Sections~\ref{BVF_LPW} and~\ref{BVF_LH}, we establish intrinsic decompositions of boundary vorticity dynamics under Lighthill-Panton-Wu and Lyman-Huggins interpretations on an arbitrarily moving and deforming boundary surface. Finally, several supplementary theorems are proposed in~Section~\ref{Several supplementary theorems}. Concluding remarks are made in Section~\ref{conclusion}.

\section{Surface geometry and related differential operators}\label{surface_geometry}
\subsection{Rational description of base surface and its neighborhood}
For any fixed instant, let $\bm{\Sigma}\subset\mathbb{R}^3$ be a sufficiently regular smooth surface, locally parametrized by the curvilinear coordinates 
 $\bm{x}=(x^1,x^2)\in\mathscr{D}_{\bm{x}}$ via the bijective $\mathcal{C}^\infty$-mapping~\citep{Aris1962vectors,Xie2013SCPMA,docarmo2016diffgeo}
\begin{equation}
	\bm{x}=(x^1,x^2)\in\mathscr{D}_{\bm{x}}\subset\mathbb{R}^{2}\mapsto\bm{\Sigma}(\bm{x})\in\mathbb{R}^{3},
\end{equation}
where $\mathscr{D}_{\bm{x}}\subset\mathbb{R}^2$  is a simply connected open domain. At any regular point $\mathcal{P}=\bm{\Sigma}(\bm{x})\in\bm{\Sigma}(\mathscr{D}_{\bm x})$, the local covariant basis vectors are defined as
\begin{equation}
	\bm{g}_{\alpha}(\bm{x})\equiv\frac{\partial\bm{\Sigma}}{\partial x^\alpha}(\bm{x})~~~(\alpha=1,2),
\end{equation}
which represent the tangent vectors of the coordinate curves and span the tangent space $T_{P}{\Sigma}\equiv{\rm Span}\left\{\bm{g}_{1}(\bm{x}),\bm{g}_{2}(\bm{x})\right\}$.
The surface metric tensor (i.e., the first fundamental form of $\bm{\Sigma}$) is defined as $\bm{G}\equiv{g}_{\alpha\beta}\bm{g}^{\alpha}\bm{g}^{\beta}~(\alpha,\beta=1,2)$, where the inner products $g_{\alpha\beta}\equiv\bm{g}_{\alpha}\bm{\cdot}\bm{g}_{\beta}$ define its covariant components. The contravariant components of $\bm{G}$ are uniquely determined by inverse metric tensor relation $g^{\alpha\beta}g_{\beta\gamma}=\delta_{\gamma}^{\alpha}$, with $\delta_{\gamma}^{\alpha}$ being the Kronecker delta. The contravariant basis vectors then follow as $\bm{g}^{\alpha}\equiv g^{\alpha\beta}\bm{g}_{\beta}$. Since the differential area element satisfies ${g}\equiv\det\left[g_{\alpha\beta}\right]=\lVert\bm{g}_{1}\times\bm{g}_{2}\rVert>0$, the unit normal to $\bm{\Sigma}$ is
\begin{equation}
	\bm{n}\equiv\frac{1}{\sqrt{g}}\star(\bm{g}_{1}\wedge\bm{g}_{2})=\frac{\bm{g}_{1}\times\bm{g}_{2}}{\lVert\bm{g}_{1}\times\bm{g}_{2}\rVert},
\end{equation}
where $\wedge$ and $\star$ denote the wedge product and Hodge star operators, respectively~\citep{chern2000lectures}. For notational convenience, we set $\bm{n}=\bm{g}^{3}=\bm{g}_{3}$ to represent the third dimension. The surface curvature tensor (i.e., the second fundamental form of $\bm{\Sigma}$) is defined as $\bm{K}\equiv{b}_{\alpha\beta}\bm{g}^{\alpha}\bm{g}^{\beta}$, where the convariant components are given by
\begin{eqnarray}
	b_{\alpha\beta}\equiv\frac{\partial\bm{g}_{\alpha}}{\partial x^\beta}\bm{\cdot}\bm{n}=\frac{\partial^2\bm{\Sigma}}{\partial x^\alpha\partial x^\beta}\bm{\cdot}\bm{n}
	=\frac{\partial\bm{g}_{\beta}}{\partial x^\alpha}\bm{\cdot}\bm{n}=b_{\beta\alpha}.
\end{eqnarray}
The contravariant and mixed components are obtained via the metric tensor components: $b^{\alpha\beta}=g^{\alpha\gamma}g^{\beta\lambda}b_{\gamma\lambda}$ and $b^{\alpha}_{\beta}=g^{\alpha\gamma}b_{\beta\gamma}$. The scalar curvature is $K\equiv{\rm tr}(\bm{K})=b_{\alpha}^{\alpha}=2H_{ M}$, where $H_{M}\equiv(\lambda_{1}+\lambda_{2})/2$ is the mean curvature, and $\lambda_{1,2}$ are the two principal curvatures. The Gaussian curvature is $K_{G}\equiv\det{(\bm{K})}=\det[b_{\alpha\beta}]/\det[{g_{\alpha\beta}}]=\lambda_{1}\lambda_{2}$, which, by Gauss's Theorema Egregium, depends only on the surface metric~\citep{chern2000lectures}. 
\begin{figure}[htbp]
	\centering 
	\includegraphics[width=1.0\columnwidth,trim={3cm 5.4cm 3cm 5.9cm},clip]{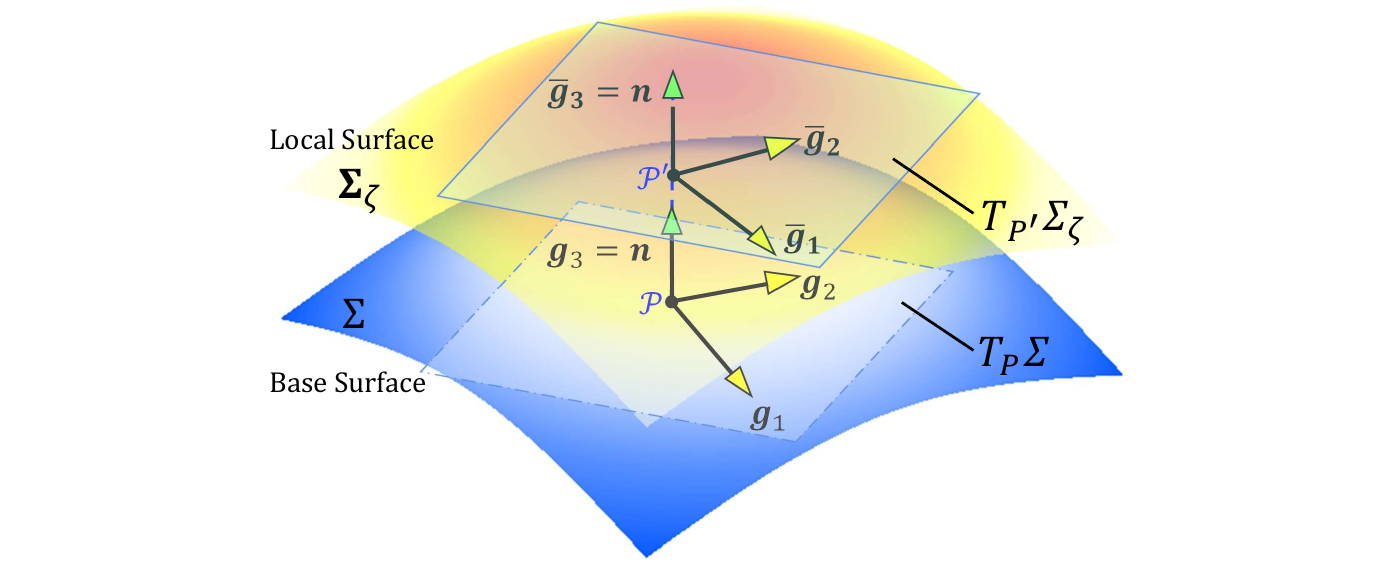}
	\caption{Schematic of the base surface $\bm{\Sigma}(\bm{x})$ and the local surface $\bm{\Sigma}_{\zeta}(\bm{x})=\bm{\Sigma}(\bm{x})+\zeta\bm{n}(\bm{x})$, where $P\in\bm{\Sigma}$ and $P^{\prime}\in\bm{\Sigma}_{\zeta}$; $\bm{n}$ denotes the unit normal vector. $T_{P}\Sigma\equiv{\rm Span}\left\{\bm{g}_{1},\bm{g}_{2}\right\}$ and $T_{P^\prime}\Sigma_{\zeta}\equiv{\rm Span}\left\{\bar{\bm{g}}_{1},\bar{\bm{g}}_{2}\right\}$ are the tangent planes at $P$ and $P^\prime$, respectively.} 
	\label{surface_extension}
\end{figure}

By performing the normal extension into a small proximity $\mathscr{U}\equiv\mathscr{D}_{\bm{x}}\times(-\delta,\delta)$ of the base surface $\bm{\Sigma}$, a 3D semi-orthogonal coordinate system can be constructed as~\citep{HuangKZ2003,docarmo2016diffgeo,Xie2020theory}
\begin{eqnarray}
	(\bm{x},x^3)=(x^1,x^2,\zeta)\in\mathscr{U}\subset\mathbb{R}^3\mapsto\bm{X}(\bm{x},\zeta)=\bm{\Sigma}(\bm{x})+\zeta\bm{n}(\bm{x})\equiv\bm{\Sigma}_{\zeta}(\bm{x})\in\mathbb{R}^{3},
\end{eqnarray}
where $\delta\in\mathbb{R}^{+}$; the base surface corresponds to $\zeta=0$, i.e., $\bm{X}(\bm{x},0)=\bm{\Sigma}(\bm{x})=\bm{\Sigma}_{0}(\bm{x})$. For each fixed $\zeta$, the surface $\bm{\Sigma}_{\zeta}(\bm{x})$  is referred to as the local surface at level $\zeta$. The local covariant basis vectors associated with the coordinates $(\bm{x},\zeta)$ are given by
\begin{subequations}
	\begin{equation}\label{eq26a}
		\bar{\bm{g}}_{\alpha}(\bm{x},\zeta)\equiv\frac{\partial\bm{\Sigma}_{\zeta}(\bm{x})}{\partial x^\alpha}=\left(\delta_{\alpha}^{\beta}-\zeta b_{\alpha}^{\beta}\right)\bm{g}_{\beta}(\bm{x})~~~(\alpha=1,2),
	\end{equation}
	\begin{equation}\label{eq26b}
		\bar{\bm{g}}_{3}(\bm{x},\zeta)\equiv\frac{\partial\bm{\Sigma}_{\zeta}(\bm{x})}{\partial\zeta}=\bm{n}(\bm{x})=\bm{g}_{3}(\bm{x}).
	\end{equation}
\end{subequations}
Using the full spatial gradient operator $\bm{\nabla}\equiv\bm{\nabla}_{\bm X}$ in $\mathbb{R}^{3}$, the corresponding contravariant basis vectors take the form as
\begin{subequations}
\begin{equation}
	\bar{\bm{g}}^{\alpha}(\bm{x},\zeta)\equiv\bm{i}_{\lambda}\frac{\partial x^\alpha}{\partial X^\lambda}=\bm{\nabla}x^{\alpha}~~(\alpha=1,2),
\end{equation}
\begin{equation}
	\bar{\bm{g}}^{3}(\bm{x},\zeta)\equiv\bm{i}_{\lambda}\frac{\partial x^3}{\partial X^\lambda}=\bm{\nabla}x^3=\bm{\nabla}\zeta=\bm{n},
\end{equation}	
\end{subequations}
where $\bm{i}_{\lambda}~(\lambda=1,2,3)$ denote the orthonormal basis vectors of the Cartesian coordinate system $\bm{X}=(X^1,X^2,X^3)$.
Consequently, the Jacobian determinant of the coordinate map is evaluated as
\begin{equation}
	J\equiv\det\left[\bar{\bm{g}}_{1},\bar{\bm{g}}_{2},\bar{\bm{g}}_{3}\right]=\left(1-K\zeta+K_{G}\zeta^2\right)\sqrt{g},
\end{equation}
Provided that $\delta>0$ is chosen sufficiently small so that $J>0$ throught $\mathscr{U}$, the mapping constitutes a smooth diffeomorphism. In other words, we restrict our attention to a family of local surfaces contained in a sufficiently thin neighborhood of the base surface.
The metric and curvature tensors of $\bm{\Sigma}_{\zeta}$ are then expressed as $\overline{\bm{G}}=\bar{g}_{\alpha\beta}\bar{\bm{g}}^{\alpha}\bar{\bm{g}}^{\beta}$ and $\overline{\bm{K}}=\bar{b}_{\alpha\beta}\bar{\bm{g}}^{\alpha}\bar{\bm{g}}^{\beta}$, with the components
\begin{subequations}
	\begin{equation}
		\bar{g}_{\alpha\beta}=g_{\alpha\beta}-2b_{\alpha\beta}\zeta+b_{\alpha}^{\gamma}b_{\gamma\beta}\zeta^2,
	\end{equation}
	\begin{equation}
		\bar{b}_{\alpha\beta}=b_{\alpha\beta}-b_{\alpha}^{\gamma}b_{\gamma\beta}\zeta.
	\end{equation}
\end{subequations}

According to the Cayley-Hamilton theorem, the following identity must hold~\citep{HuangKZ2003}:
\begin{eqnarray}\label{third}
\bm{K}^2\equiv\bm{K}\bm{\cdot}\bm{K}=b_{\alpha}^{\gamma}b_{\gamma\beta}\bm{g}^{\alpha}\bm{g}^{\beta}=K\bm{K}-K_G\bm{G},
\end{eqnarray}
where $\bm{K}^2$ is referred to as the third fundamental tensor. It is noted that~\eqref{third} exhibits a structure analogous to that of a rank-two algebraic equation for $\bm{K}$. Inspired by this resemblance and drawing upon the formal parallelism with Vieta's formulas,~\citet{Yin2008new} introduced the conjugate fundamental tensor $\hat{\bm{K}}$, which, together with $\bm{K}$, satisfies the following elegant pair of relations:
\begin{equation}\label{713_ab}
	\bm{K}+\hat{\bm{K}}=K\bm{G},~\hat{\bm{K}}\bm{\cdot}\bm{K}=K_{G}\bm{G}.
\end{equation}
In contrast to $\bm{K}$, the tensor $\hat{\bm{K}}$ appears only rarely in the standard differential-geometric literature. However, it plays an essential role in the construction of a pair of fundamental differential operators on curved surfaces, as well as in the description of shape-driving forces on biomembrane surfaces~\citep{Yin2008second,Yin2011AMM}. When $K_{G}\equiv\det(\bm{K})\neq0$ holds  at the point under consideration, the conjugate tensor admits the explicit representation
\begin{eqnarray}
	\hat{\bm{K}}=K\bm{G}-\bm{K}=K_{G}\bm{K}^{-1}.
\end{eqnarray}

\subsection{Surface differential operators}
\subsubsection{Surface gradient operator}
Consider a $k$-th order smooth tensor field $\bm{\Phi}$ defined in the small neighborhood $\mathscr{U}$ of the base surface $\bm{\Sigma}$.
On each local surface $\bm{\Sigma}_{\zeta}$, the surface gradient operator is defined as~\citep{WuWu1996,Xie2013SCPMA}
\begin{eqnarray}\label{sdo1}
	\overline{\bm{\nabla}}_{\pi}\equiv\bar{\bm{g}}^{\alpha}\frac{\partial}{\partial x^\alpha}=\bm{\nabla}-\bar{\bm{g}}^{3}\frac{\partial}{\partial x^3}=\bm{\nabla}-\bm{n}\frac{\partial}{\partial n},~\frac{\partial}{\partial n}\equiv\frac{\partial}{\partial x^3}=\frac{\partial}{\partial \zeta},
\end{eqnarray}
where the subscripts $\pi$ and $n$ denote the tangential and normal directions, respectively.
Restricting the operator to the base surface $\bm{\Sigma}$, we obtain
\begin{eqnarray}\label{sdo2}
	\bm{\nabla}_{\pi}\bm{\Phi}_{\Sigma}\equiv{\bm{g}}^{\alpha}\frac{\partial\bm{\Phi}_{\Sigma}}{\partial x^\alpha}=\bm{G}\bm{\cdot}\left[\bm{\nabla}\bm{\Phi}\right]_{\Sigma}=\left[\bm{\nabla}\bm{\Phi}\right]_{\Sigma}-\bm{n}\left[\frac{\partial\bm{\Phi}}{\partial n}\right]_{\Sigma}.
\end{eqnarray}
Here, $\bm{\Phi}_{\Sigma}\equiv\bm{\Phi}\lvert_{\Sigma}=\bm{\Phi}(\bm{x},0)$ denotes the restriction of $\bm{\Phi}$ on $\bm{\Sigma}$. For notational brevity, the subscript ``$\Sigma$'' will be omitted in the subsequent discussion whenever no ambiguity arises. It follows that the surface curvature tensor $\bm{K}$ and the scalar curvature $K$ can be expressed intrinsically in terms of the surface gradient operator and the surface normal on $\bm{\Sigma}$:
\begin{equation}
	\bm{K}=-\bm{\nabla}_{\pi}\bm{n},~~K\equiv{\rm tr}(\bm{K})=-\bm{\nabla}_{\pi}\bm{\cdot}\bm{n}.
\end{equation}

The Laplace-Beltrami operator on $\bm{\Sigma}$ is defined as ${\nabla}_{\pi}^2\bm{\Phi}_{\Sigma}\equiv\bm{\nabla}_{\pi}\bm{\cdot}\bm{\nabla}_{\pi}\bm{\Phi}_{\Sigma}$.
By applying~\eqref{sdo2} twice, the full Laplacian of $\bm{\Phi}$ in $\mathbb{R}^3$ can be evaluated on $\bm{\Sigma}$ as
\begin{eqnarray}\label{sdo4}
\left[\nabla^2\bm{\Phi}\right]_{\Sigma}&\equiv&\left[\bm{\nabla}\bm{\cdot}\bm{\nabla}\bm{\Phi}\right]_{\Sigma}=\bm{\nabla}_{\pi}\bm{\cdot}\left[\bm{\nabla\Phi}\right]_{\Sigma}+\bm{n}\bm{\cdot}\left[\frac{\partial\bm{\nabla\Phi}}{\partial n}\right]_{\Sigma}\nonumber\\
&=&\bm{\nabla}_{\pi}\bm{\cdot}\left(\bm{\nabla}_{\pi}\bm{\Phi}_{\Sigma}+\bm{n}\left[\frac{\partial\bm{\Phi}}{\partial n}\right]_{\Sigma}\right)+\left[\frac{\partial}{\partial n}\left(\bm{n}\bm{\cdot}\bm{\nabla\Phi}\right)\right]_{\Sigma}\nonumber\\
&=&{\nabla}_{\pi}^2\bm{\Phi}_{\Sigma}+\left(\bm{\nabla}_{\pi}\bm{\cdot}\bm{n}\right)\left[\frac{\partial\bm{\Phi}}{\partial n}\right]_{\Sigma}+\bm{n}\bm{\cdot}\bm{\nabla}_{\pi}\left[\frac{\partial\bm{\Phi}}{\partial n}\right]_{\Sigma}+\left[\frac{\partial^2\bm{\Phi}}{\partial n^2}\right]_{\Sigma}\nonumber\\
&=&{\nabla}_{\pi}^2\bm{\Phi}_{\Sigma}-K\left[\frac{\partial\bm{\Phi}}{\partial n}\right]_{\Sigma}+\left[\frac{\partial^2\bm{\Phi}}{\partial n^2}\right]_{\Sigma},
\end{eqnarray}
where $\nabla^2\equiv\bm{\nabla}\bm{\cdot}\bm{\nabla}$ denotes the standard Laplacian in the ambient Euclidean space $\mathbb{R}^3$.
\begin{theorem}\label{curl_1}
On the base surface $\bm{\Sigma}$, the surface curl of the unit normal $\bm{n}$ necessarily vanishes:
\begin{equation}\label{curl_n}
	\bm{\nabla}_{\pi}\times\bm{n}=\bm{0}.
\end{equation}
\end{theorem}
\begin{proof}[Proof of Theorem~{\upshape\ref{curl_1}}]	
Since the curvature tensor $\bm{K}$ is symmetric, it follows that
\begin{eqnarray}
	\bm{\nabla}_{\pi}\times\bm{n}&=&\bm{g}^{\alpha}\times\frac{\partial\bm{n}}{\partial x^\alpha}=2\star\mathscr{A}\left(\bm{g}^{\alpha}\frac{\partial\bm{n}}{\partial x^\alpha}\right)\nonumber\\
	&=&-2\star\mathscr{A}\left(\bm{\nabla}_{\pi}\bm{n}\right)=2\star\mathscr{A}(\bm{K})=\bm{0},
\end{eqnarray}
where $\mathscr{A}$ represents the antisymmetrization operator.
This completes the proof.
\end{proof}
\begin{theorem}\label{div_1}
On the base surface $\bm{\Sigma}$, the surface divergence of the metric tensor $\bm{G}$ is
\begin{equation}\label{div_G}
	\bm{\nabla}_{\pi}\bm{\cdot}\bm{G}=K\bm{n}.
\end{equation}
\end{theorem}
\begin{proof}[Proof of Theorem~{\upshape\ref{div_1}}] Note that the surface metric tensor can be expressed as $\bm{G}=\bm{g}_{\beta}\otimes\bm{g}^{\beta}$. Then, by employing the Gauss-Weingarten-Codazzi (GWC) formulas~\citep{HuangKZ2003}, we obtain
		\begin{eqnarray}
		\bm{\nabla}_{\pi}\bm{\cdot}\bm{G}&=&\bm{g}^{\alpha}\bm{\cdot}\frac{\partial}{\partial x^\alpha}\left(\bm{g}_{\beta}\bm{g}^{\beta}\right)\nonumber\\
		&=&\bm{g}^{\alpha}\bm{\cdot}\left(\frac{\partial\bm{g}_{\beta}}{\partial x^\alpha}\bm{g}^{\beta}+\bm{g}_{\beta}\frac{\partial\bm{g}^{\beta}}{\partial x^\alpha}\right)\nonumber\\
		&=&\bm{g}^{\alpha}\bm{\cdot}\left(\Gamma_{\alpha\beta}^{\lambda}\bm{g}_{\lambda}\bm{g}^{\beta}+b_{\alpha\beta}\bm{n}\bm{g}^{\beta}-\bm{g}_{\beta}\Gamma_{\alpha\lambda}^{\beta}\bm{g}^{\lambda}+\bm{g}_{\beta}b_{\alpha}^{\beta}\bm{n}\right)\nonumber\\
		&=&\Gamma_{\alpha\beta}^{\alpha}\bm{g}^{\beta}-\Gamma_{\alpha\lambda}^{\alpha}\bm{g}^{\lambda}+b_{\alpha}^{\alpha}\bm{n}\nonumber\\
		&=&K\bm{n}.
	\end{eqnarray}
	This completes the proof.
\end{proof}

\subsubsection{Levi-Civita gradient operator}
On the 2D Riemannian submanifold $(\bm{\Sigma},\bm{G})$, one can introduce the Levi-Civita connection $\nabla_{\frac{\partial}{\partial x^\alpha}}$ (also referred to as the induced connection)~\citep{Dubrovin1992modern,chern2000lectures}. For the tangent and cotangent vector fields $\left\{\frac{\partial}{\partial x^\alpha},dx^\beta\right\}$, the following two fundamental properties then hold:
\begin{eqnarray}
	\nabla_{\frac{\partial}{\partial x^\alpha}}\frac{\partial}{\partial x^\beta}=\Gamma_{\alpha\beta}^{\gamma}\frac{\partial}{\partial x^\gamma},~~	\nabla_{\frac{\partial}{\partial x^\alpha}}dx^\beta=-\Gamma_{\alpha\gamma}^{\beta}dx^\gamma,
\end{eqnarray}
where $\Gamma_{\alpha\beta,\gamma}$ and $\Gamma_{\alpha\beta}^{\gamma}\equiv g^{\gamma\lambda}\Gamma_{\alpha\beta,\lambda}$ are
the Christoffel connection coefficients of the first and second kinds, respectively. The former is evaluated from the partial derivatives of the covariant components of $\bm{G}$:
\begin{equation}
		\Gamma_{\alpha\beta,\gamma}=\frac{1}{2}\left(\frac{\partial g_{\alpha\gamma}}{\partial x^{\beta}}+\frac{\partial g_{\beta\gamma}}{\partial x^{\alpha}}-\frac{\partial g_{\alpha\beta}}{\partial x^{\gamma}}\right).
\end{equation}
By virtue
of the Weingarten mapping for the unit normal $\bm{n}=\frac{\partial}{\partial x^3}$, one observes that the Levi-Civita connection differs from the ambinent connection in $\mathscr{U}\subset\mathbb{R}^3$
by the normal connection due to the surface curvature tensor $\bm{K}$. Treating the coordinate basis $\left\{\frac{\partial}{\partial x^\alpha},dx^\beta\right\}$ as equivalently to the natural frame $\left\{\bm{g}_{\alpha},\bm{g}^{\beta}\right\}$ will lead to the well-known GWC formulas~\citep{HuangKZ2003,Xie2020theory} which account for the frame movement along the surface in classical differential geometry. Based on the surface Levi-Civita connection and the contravariant basis vector field, the Levi-Civita gradient operator is defined as~\citep{Xie2013688,Shiqian2017}
\begin{equation}
	\bm{\nabla}_{C}\equiv\bm{g}^{\alpha}\nabla_{\frac{\partial}{\partial x^\alpha}}.
\end{equation}
Specifically, for a tensor field $\bm{\Phi}\in\mathscr{T}^{2}(\mathbb{R}^3)$, the action of the Levi-Civita gradient operator on $\bm{\Sigma}$ is defined componentwise as
\begin{subequations}
	\begin{eqnarray}\label{LC_def}
		\bm{\nabla}_{C}\diamondsuit\bm{\Phi}&\equiv&\left(\bm{g}^{\alpha}{\nabla}_{\frac{\partial}{\partial x^{\alpha}}}\right)\diamondsuit\left(\Phi^{\beta}_{\cdot \gamma}\bm{g}_{\beta}\bm{g}^{\gamma}+\Phi^{\beta}_{\cdot 3}\bm{g}_{\beta}\bm{n}+\Phi^{3}_{\cdot \gamma}\bm{n}\bm{g}^{\gamma}+\Phi^{3}_{\cdot3}\bm{n}\bm{n}\right)\nonumber\\
		&=&\bm{g}^{\alpha}\diamondsuit{\nabla}_{\frac{\partial}{\partial x^{\alpha}}}\left(\Phi^{\beta}_{\cdot \gamma}\bm{g}_{\beta}\bm{g}^{\gamma}+\Phi^{\beta}_{\cdot 3}\bm{g}_{\beta}\bm{n}+\Phi^{3}_{\cdot \gamma}\bm{n}\bm{g}^{\gamma}+\Phi^{3}_{\cdot3}\bm{n}\bm{n}\right)\nonumber\\
		&=&(\nabla_{\alpha}\Phi^{\beta}_{\cdot \gamma})\bm{g}^{\alpha}\diamondsuit\bm{g}_{\beta}\bm{g}^{\gamma}
		+(\nabla_{\alpha}\Phi^{\beta}_{\cdot 3})\bm{g}^{\alpha}\diamondsuit\bm{g}_{\beta}\bm{n}
		+(\nabla_{\alpha}\Phi^{3}_{\cdot \gamma})\bm{g}^{\alpha}\diamondsuit\bm{n}\bm{g}^{\gamma}\nonumber\\
		& &+(\nabla_{\alpha}\Phi^{3}_{\cdot3})\bm{g}^{\alpha}\diamondsuit\bm{n}\bm{n},
	\end{eqnarray}
	with the precise form of the connection terms depending on the tensor component type.
	Here, the symbol $\diamondsuit$ represents an arbitrary admissible tensor-algebraic operator (such as the inner product $\bm{\cdot}$, the cross product $\times$, or the tensor product $\otimes$). The operator $\nabla_{\alpha}$  represents covariant differentiation with respect to the coordinate $x^{\alpha}$, and it acts exclusively on indices associated with the tangent bundle of $\bm{\Sigma}$. The explicit expressions for each covariant derivative are given as follows:
	\begin{equation}
		\nabla_{\alpha}\Phi^{\beta}_{\cdot \gamma}\equiv\frac{\partial\Phi^{\beta}_{\cdot \gamma}}{\partial x^{\alpha}}+\Gamma_{\alpha\lambda}^{\beta}\Phi^{\lambda}_{\cdot \gamma}-\Gamma_{\alpha\gamma}^{\lambda}\Phi^{\beta}_{\cdot \lambda},~\nabla_{\alpha}\Phi^{3}_{\cdot \gamma}\equiv\frac{\partial\Phi^{3}_{\cdot \gamma}}{\partial x^{\alpha}}-\Gamma_{\alpha\gamma}^{\lambda}\Phi^{3}_{\cdot \lambda};
	\end{equation}
	\begin{equation}
		\nabla_{\alpha}\Phi^{\beta}_{\cdot 3}\equiv\frac{\partial\Phi^{\beta}_{\cdot 3}}{\partial x^{\alpha}}+\Gamma_{\alpha\lambda}^{\beta}\Phi^{\lambda}_{\cdot 3},~	\nabla_{\alpha}\Phi^{3}_{\cdot3}=\frac{\partial\Phi^{3}_{\cdot3}}{\partial x^\alpha}.
	\end{equation}
\end{subequations}

\section{Velocity gradient tensor and vorticity decompositions}\label{VGT_vorticity}
\subsection{Velocity gradient tensor, dilatation and vorticity}
The symmetric-antisymmetric decomposition (SAD) of the velocity gradient tensor (VGT) $\bm{A}\equiv\bm{\nabla u}$ is expressed as~\citep{Truesdell1954,Batchelor1967}
\begin{subequations}\label{SAD}
	\begin{equation}\label{SAD1}
		\bm{A}=\bm{D}+\bm{\varOmega},
	\end{equation}
	\begin{equation}\label{SAD2}
		\bm{D}\equiv\mathscr{S}[\bm{A}]=\frac{1}{2}\left(\bm{A}+\bm{A}^{\rm T}\right),~~\bm{\varOmega}\equiv\mathscr{A}[\bm{A}]=\frac{1}{2}\left(\bm{A}-\bm{A}^{\rm T}\right),
	\end{equation}
\end{subequations}
where $\bm{D}$ is the strain-rate tensor, and $\bm{\varOmega}$ is the rotation-rate tensor. The operators $\mathscr{S}$ and $\mathscr{A}$ denote, respectively, the
symmetric and antisymmetric parts of a second-order tensor. It is noted that $\bm{A}$ is defined as the left gradient of the velocity field $\bm{u}$, while its transpose, denoted by ``T'', corresponds to the right gradient tensor, $\bm{A}^{\rm T}=\bm{\nabla}\bm{u}^{\rm T}\equiv\bm{u}\bm{\nabla}$. Since ${\rm tr}(\bm{\varOmega})=0$ (where ${\rm tr}$ denotes the trace operator), the dilatation, a scalar quantity that characterizes fluid compressibility, is given by
\begin{equation}
	\vartheta\equiv\bm{\nabla}\bm{\cdot}\bm{u}={\rm tr}(\bm{A})={\rm tr}(\bm{D}).
\end{equation}
The evaluation of the dual vector of the rotation-rate tensor yields the vorticity vector,
\begin{equation}
	\bm{\omega}\equiv2\star\bm{\varOmega}=\star(\bm{\nabla}\wedge\bm{u})=\bm{\nabla}\times\bm{u}.
\end{equation}
The interaction between the longitudinal and transverse processes in complex flows can be well described under the $(\vartheta,\bm{\omega})$-framework~\citep{Wu2006vorticity}.

\subsection{Surface-element-based kinematic vorticity decomposition}
Consider a material surface element $\delta\bm{\Sigma}\equiv\delta\Sigma\bm{n}$, oriented by the unit normal $\bm{n}$ and with the differential surface area $\delta\Sigma=\lVert\delta\bm{\Sigma}\rVert$. Using the surface rate of deformation tensor~\citep{Dishington1965} 
\begin{equation}\label{dishington}
	\bm{B}\equiv\vartheta\bm{I}-\bm{A}^{\rm T},
\end{equation}
the material rate of change of $\delta\bm{\Sigma}$ is expressed as~\citep{Truesdell1954,Batchelor1967,WuJZ2005JFM,Xie2013SCPMA}
\begin{equation}\label{SS1}
	\frac{1}{\delta\Sigma}\frac{D\delta\bm{\Sigma}}{Dt}=\bm{n}\bm{\cdot}\bm{B},
\end{equation}
where $D/Dt$ represents the material derivative, and $\bm{I}$ is the metric tensor for $\mathbb{R}^3$. The right-hand side of~\eqref{SS1} can be equivalently written as $-(\bm{n}\times\bm{\nabla})\times\bm{u}=-\left(\bm{n}\times\bm{\nabla}_{\pi}\right)\times\bm{u}$, indicating that this evolution rate is intrinsically determined by the instantaneous surface geometry and tangential velocity distribution, independent of the external flow field in the ambient space. 
Applying the Leibniz rule of partial differentiation to the left-hand side of~\eqref{SS1} then yields
\begin{equation}
	\frac{1}{\delta\Sigma}\frac{D\delta\bm{\Sigma}}{Dt}=\bm{n}\left(	\frac{1}{\delta\Sigma}\frac{D\delta{\Sigma}}{Dt}\right)+\frac{D\bm{n}}{Dt}.
\end{equation}
Therefore, an orthogonal decomposition of~\eqref{SS1} readily follows~\citep{WuJZ2005JFM,Xie2013SCPMA}
\begin{subequations}
	\begin{equation}\label{aa1}
		\frac{1}{\delta\Sigma}\frac{D\delta{\Sigma}}{Dt}=\bm{\nabla}_{\pi}\bm{\cdot}\bm{u}=\bm{\nabla}_{\pi}\bm{\cdot}\bm{u}_{\pi}-Ku_{n}\equiv\vartheta_{\pi},
	\end{equation}
	\begin{equation}\label{aa2}
		\frac{D\bm{n}}{Dt}=\bm{W}_{\rm eff}\times\bm{n}=-\bm{\nabla}_{\pi}\bm{u}\bm{\cdot}\bm{n}=-(\bm{K}\bm{\cdot}\bm{u}_{\pi}+\bm{\nabla}_{\pi}u_{n}).
	\end{equation}
\end{subequations}
In~\eqref{aa1} and~\eqref{aa2}, we have employed the orthogonal decomposition of the velocity field $\bm{u}=\bm{u}_{\pi}+u_n\bm{n}$, where $\bm{u}_{\pi}\equiv\bm{G}\bm{\cdot}\bm{u}=\bm{n}\times(\bm{u}\times\bm{n})$ is the surface tangential velocity, and $\bm{u}_{n}=u_{n}\bm{n}$ is the surface-normal velocity. In a unified manner, the dilatation is decomposed as
\begin{subequations}
	\begin{equation}
		\vartheta=\vartheta_{\pi}+\vartheta_{n},
	\end{equation}
	\begin{equation}
	\vartheta_{\pi}=\bm{\nabla}_{\pi}\bm{\cdot}\bm{u},~\vartheta_{n}=\frac{\partial u_{n}}{\partial n}=\bm{n}\bm{\cdot}\bm{A}\bm{\cdot}\bm{n}=\bm{n}\bm{\cdot}\bm{D}\bm{\cdot}\bm{n}.
	\end{equation}
\end{subequations}
In~\eqref{aa1}, the relative rate of change of the surface area element $\delta\Sigma$ is determined by the surface velocity divergence $\bm{\nabla}_{\pi}\bm{\cdot}\bm{u}={\rm tr}(\bm{\nabla}_{\pi}\bm{u})$, which is decomposed into the tangential velocity divergence $(\bm{\nabla}_{\pi}\bm{\cdot}\bm{u}_{\pi})$ and a term ($-Ku_n=-2H_{M}u_{n}$) arising from the coupling between the mean curvature and the surface-normal velocity.
In~\eqref{aa2}, $\bm{W}_{\rm eff}\equiv-\bm{n}\times(\bm{A}\bm{\cdot}\bm{n})=-\bm{n}\times(\bm{\nabla}_{\pi}\bm{u}\bm{\cdot}\bm{n})$ is the effective angular velocity of the unit normal $\bm{n}$, which is caused by the tangential variation of $u_{n}$, and two regular circular motions along the two principal directions of the surface. The total angular velocity of $\bm{n}$ is given by $\bm{W}=\bm{W}_{\rm eff}+(\omega_{n}/2)\bm{n}=-\bm{W}_{\rm D}+\bm{\omega}/2$, where $\bm{W}_{\rm D}\equiv\bm{n}\times(\bm{n}\bm{\cdot}\bm{D})$ denotes the specific angular velocity attributable to the strain-rate tensor~\citep{Chen2026General}. Although this contribution vanishes for a rigid-body motion, $\bm{W}_{\rm D}$ plays an essential role in general fluid flows where deformation is present.

Based on the unit normal $\bm{n}$ and the SAD~\eqref{SAD1}, a direction-dependent vorticity decomposition (DVD) on $\bm{\Sigma}$ was recently proposed  as~\citep{Chen2026General}
\begin{subequations}\label{decomp_vorticity}
		\begin{equation}\label{decomp_vorticity1}
		\bm{\omega}=\bm{R}+\bm{S}+\bm{\omega}_{n},
	\end{equation}
		\begin{equation}\label{decomp_vorticity3}
	\bm{\omega}_{\pi}=\bm{G}\bm{\cdot}\bm{\omega}=\bm{R}+\bm{S},~\bm{\omega}_{n}\equiv\omega_{n}\bm{n}=[(\bm{n}\times\bm{\nabla}_{\pi})\bm{\cdot}\bm{u}_{\pi}]\bm{n},
	\end{equation}
	\begin{equation}\label{decomp_vorticity2}
		\bm{R}\equiv2\bm{W}_{\rm eff}=-2\bm{n}\times\left(\bm{A}\bm{\cdot}\bm{n}\right),~~
		\bm{S}\equiv2\bm{W}_{\rm D}=2\bm{n}\times\left(\bm{D}\bm{\cdot}\bm{n}\right).
	\end{equation}
\end{subequations}
Here, $\bm{R}$ is termed the rigid-rotation mode, which equals twice the effective angular velocity of the unit normal $\bm{n}$. The spin mode $\bm{s}$ equals twice the specific angular velocity $\bm{W}_{D}$, which is identical to the relative vorticity $\bm{\omega}_{r}\equiv\bm{\omega}-2\bm{W}$. According to the generalized Caswell formula~\citep{Caswell1967,WuJZ2005JFM}, the relative vorticity $\bm{\omega}_{r}$ determines the surface shear stress for a compressible viscous Newtonian fluid: $\bm{\tau}=\mu\bm{\omega}_{r}\times\bm{n}=\mu\bm{S}\times\bm{n}$, where $\mu$ is the dynamic viscosity. In this sense, $\bm{S}$ is also referred to as the shear mode in viscous flows. For convenience, we introduce three tangent vectors $(\bm{\xi},\bm{\xi}_{R},\bm{\xi}_{S})\in T\Sigma$, which are orthogonal to $(\bm{\omega},\bm{R},\bm{s})$, respectively, and are defined as follows:
\begin{subequations}\label{xi_decomp}
	\begin{equation}\label{xi_decomp1}
		\bm{\xi}\equiv\bm{\omega}\times\bm{n}=\bm{\xi}_{R}+\bm{\xi}_{S},
	\end{equation}
	\begin{equation}\label{xi_decomp2}
	\bm{\xi}_{R}\equiv\bm{R}\times\bm{n}=2\bm{W}_{\rm eff}\times\bm{n}=-2\left(\bm{K}\bm{\cdot}\bm{u}_{\pi}+\bm{\nabla}_{\pi}u_{n}\right),
	\end{equation}
		\begin{equation}\label{xi_decomp3}
\bm{\xi}_{S}\equiv\bm{S}\times\bm{n}=\bm{\omega}_{r}\times\bm{n}=\bm{\tau}/\mu.
	\end{equation}
\end{subequations}

\subsection{Decomposition of enstrophy}
Enstrophy is defined as half the squared magnitude of the vorticity vector,
\begin{eqnarray}
\Omega\equiv\frac{1}{2}\omega^2=\frac{1}{2}\bm{\omega}\bm{\cdot}\bm{\omega},
\end{eqnarray}
 which measures the local intensity of rotational motion without regard to the orientation of the vorticity vector. Based on the vorticity decomposition in \eqref{decomp_vorticity}, the enstrophy admits the following decomposition:
\begin{subequations}
	\begin{equation}
		\Omega=\Omega_{\pi}+\Omega_{n},
	\end{equation}
	\begin{equation}
		\Omega_{\pi}\equiv\frac{1}{2}\omega_{\pi}^2=\Omega_{R}+\Omega_{S}+\Omega_{RS}	,~\Omega_{n}\equiv\frac{1}{2}\omega_{n}^2.
	\end{equation}
	\begin{equation}
		\Omega_{R}\equiv\frac{1}{2}{R}^2=\frac{1}{2}\bm{R}\bm{\cdot}\bm{R},~	\Omega_{S}\equiv\frac{1}{2}{S}^2=\frac{1}{2}\bm{S}\bm{\cdot}\bm{S},~\Omega_{RS}\equiv\bm{R}\bm{\cdot}\bm{S}.
	\end{equation}
\end{subequations}
Here, $\Omega_{R}$ and $\Omega_{S}$ denote the enstrophy contributions from the rigid-rotation and spin modes, respectively. The coupling term $\Omega_{RS}$ is directly associated with the Klein-Kaden-Betz (KKB) and anti-KKB configurations of the vorticity modes in a flow field. It is well established that the only known mechanism for the rapid formation of axial vortices in a fluid of small viscosity is the rolling up of thin shear layers \citep{Betz1950,Wu2006vorticity}. The KKB configuration, characterized by $\Omega_{RS}>0$, occurs when axial vortices are formed via the wrapping of
shear layers, accompanied by a conversion from spin mode to rigid-rotation mode, ultimately leading to vortex intensification~\citep{Klein1910,Kaden1931,Betz1950,Chen2026General}. The anti-KKB configuration with ${\Omega}_{RS}<0$~\citep{Chen2026General,Chen2026operator}, arises when the axial vortex possesses sufficient swirling strength to maintain its coherence, such that opposing spin effects cannot significantly disrupt the primary vortex structure. During
the formation of an axial vortex, the anti-KKB mechanism could be activated to suppress the unbounded growth of swirling strength in the inner core region, while the KKB mechanism may
continue to govern the roll-up of the outer shear layers.

\subsection{Surface and Levi-Civita gradients of velocity}
\begin{theorem}\label{SLC1} 
On the base surface $\bm{\Sigma}$, the following identities hold:
\begin{subequations}
	\begin{equation}\label{qq00}
		\bm{\nabla}_{\pi}\bm{u}=\bm{\nabla}_{C}\bm{u}_{\pi}-u_n\bm{K}-\frac{1}{2}\bm{\xi}_{R}\bm{n},
	\end{equation}
	\begin{equation}\label{qq0}
		\bm{\nabla}_{\pi}\bm{u}_{\pi}=\bm{\nabla}_{C}\bm{u}_{\pi}+(\bm{K}\bm{\cdot}\bm{u}_{\pi})\bm{n}.
	\end{equation}
\end{subequations}
Equivalently, we have
\begin{subequations}
	\begin{equation}\label{eq320a}
		\left(\bm{\nabla}_{\pi}\bm{u}\right)_{\pi}=\bm{\nabla}_{C}\bm{u}_{\pi}-u_{n}\bm{K},~
		\bm{\nabla}_{\pi}\bm{u}\bm{\cdot}\bm{n}=-\frac{1}{2}\bm{\xi}_{R}=-\bm{W}_{\rm eff}\times\bm{n},
	\end{equation}
	\begin{equation}\label{eq320c}
		\left(\bm{\nabla}_{\pi}\bm{u}_{\pi}\right)_{\pi}=\bm{\nabla}_{C}\bm{u}_{\pi},~\bm{\nabla}_{\pi}\bm{u}_{\pi}\bm{\cdot}\bm{n}=\bm{K}\bm{\cdot}\bm{u}_{\pi}.
	\end{equation}
\end{subequations}
\end{theorem}
\begin{proof}[Proof of Theorem~{\upshape\ref{SLC1}}]	
On the base surface $\bm{\Sigma}$, the Levi-Civita gradient of $\bm{u}$ is evaluated as
\begin{eqnarray}\label{qq1}
	\bm{\nabla}_{C}\bm{u}&=&\bm{g}^{\alpha}{\nabla}_{\frac{\partial}{\partial x^{\alpha}}}\left(u_{\beta}\bm{g}^{\beta}+u_{n}\bm{n}\right)\nonumber\\
	&=&(\nabla_{\alpha}u_{\beta})\bm{g}^{\alpha}\bm{g}^{\beta}+\left(\bm{g}^{\alpha}\frac{\partial u_n}{\partial x^\alpha}\right)\bm{n}\nonumber\\
	&=&\bm{\nabla}_{C}\bm{u}_{\pi}+(\bm{\nabla}_{\pi}u_{n})\bm{n},
\end{eqnarray}
where the covariant derivative of the covariant velocity component is expressed as
\begin{equation}\label{qq2}
	\nabla_{\alpha}u_{\beta}=\frac{\partial u_\beta}{\partial x^{\alpha}}-\Gamma_{\alpha\beta}^{\gamma}u_{\gamma}.
\end{equation}

The surface tangential velocity gradient is expanded as
\begin{equation}\label{qq3}
	\bm{\nabla}_{\pi}\bm{u}=\bm{\nabla}_{\pi}\bm{u}_{\pi}+\bm{\nabla}_{\pi}(u_n\bm{n}).
\end{equation}
The first term is evaluated as
\begin{equation}\label{qq4}
	\bm{\nabla}_{\pi}\bm{u}_{\pi}=(\nabla_{\alpha}u_{\beta})\bm{g}^{\alpha}\bm{g}^{\beta}+u_{\beta}b_{\alpha}^{\beta}\bm{g}^{\alpha}\bm{n}=\bm{\nabla}_{C}\bm{u}_{\pi}+(\bm{K}\bm{\cdot}\bm{u}_{\pi})\bm{n}.
\end{equation}
The second term is evaluated as
\begin{equation}\label{qq5}
	\bm{\nabla}_{\pi}(u_n\bm{n})=(\bm{\nabla}_{\pi}u_{n})\bm{n}+u_{n}\bm{\nabla}_{\pi}\bm{n}=(\bm{\nabla}_{\pi}u_{n})\bm{n}-u_{n}\bm{K}.
\end{equation}
Substituting~\eqref{qq4} and~\eqref{qq5} into~\eqref{qq3} yields
\begin{eqnarray}\label{qq6}
	\bm{\nabla}_{\pi}\bm{u}=\bm{\nabla}_{C}\bm{u}_{\pi}+\left(\bm{K}\bm{\cdot}\bm{u}_{\pi}+\bm{\nabla}_{\pi}u_{n}\right)\bm{n}-u_{n}\bm{K}.
\end{eqnarray}
By using~\eqref{xi_decomp2},~\eqref{qq6} becomes
\begin{eqnarray}\label{qq7}
	\bm{\nabla}_{\pi}\bm{u}=\bm{\nabla}_{C}\bm{u}_{\pi}-\frac{1}{2}\bm{\xi}_{R}\bm{n}-u_{n}\bm{K}.
\end{eqnarray}
Letting $u_{n}=0$ in~\eqref{qq6} gives~\eqref{qq0}.
This completes the proof.
\end{proof}

\begin{theorem}\label{SLC2}
On the base surface $\bm{\Sigma}$, the full velocity gradient tensor $\bm{A}\equiv\bm{\nabla u}$ is decomposed as
\begin{subequations}
	\begin{equation}\label{qq8}
		\bm{A}=\vartheta_{n}\bm{n n}+\bm{n}\left(\frac{1}{2}\bm{\xi}_{R}+\bm{\xi}_{S}\right)+\bm{\nabla}_{\pi}\bm{u}.
	\end{equation}
Equivalently, we have the following representation
	\begin{equation}\label{qq9}
		\bm{A}=\vartheta_{n}\bm{n n}+\bm{n}\bm{\xi}_{S}+\mathscr{A}\left(\bm{n}\bm{\xi}_{R}\right)+\bm{\nabla}_{C}\bm{u}_{\pi}-u_{n}\bm{K}.
	\end{equation}
	where
\begin{equation}
	\mathscr{A}(\bm{n}\bm{\xi}_{R})\equiv\frac{1}{2}\left(\bm{n}\bm{\xi}_{R}-\bm{\xi}_{R}\bm{n}\right)=\frac{1}{2}\bm{n}\wedge\bm{\xi}_{R},~\left(\bm{\nabla}_{\pi}\bm{u}\right)_{\pi}=\bm{\nabla}_{C}\bm{u}_{\pi}-u_{n}\bm{K}.
\end{equation}
\end{subequations}
\end{theorem}
\begin{proof}[Proof of Theorem~{\upshape\ref{SLC2}}] Using the identity between $\bm{\varOmega}$ and $\bm{\omega}$:
\begin{equation}
2\bm{n}\bm{\cdot}\bm{\varOmega}=\bm{\omega}\times\bm{n}=\bm{\xi},
\end{equation}
it follows that
\begin{eqnarray}\label{qq10}
	\bm{A}&=&\bm{\nabla}_{\pi}\bm{u}+\bm{n}\bm{n}\bm{\cdot}\bm{A}\nonumber\\
	&=&\bm{\nabla}_{\pi}\bm{u}+\bm{nn}\bm{\cdot}(\bm{A}-\bm{A}^{\rm T}+\bm{A}^{\rm T})\nonumber\\
	&=&\bm{\nabla}_{\pi}\bm{u}+\bm{n}\left(2\bm{n}\bm{\cdot}\bm{\varOmega}\right)+\bm{n}\left(\bm{A}\bm{\cdot}\bm{n}\right)\nonumber\\
	&=&\bm{\nabla}_{\pi}\bm{u}+\bm{n}\bm{\xi}+\bm{n}\left(\bm{\nabla}_{\pi}\bm{u}\bm{\cdot}\bm{n}+\vartheta_{n}\bm{n}\right)
\end{eqnarray}
Substituting~\eqref{qq00} into~\eqref{qq10} yields
\begin{eqnarray}\label{qq11}
	\bm{A}&=&\bm{\nabla}_{\pi}\bm{u}+\bm{n}\bm{\xi}+\bm{n}\left(-\frac{1}{2}\bm{\xi}_{R}+\vartheta_{n}\bm{n}\right)\nonumber\\
	&=&\bm{\nabla}_{\pi}\bm{u}+\bm{n}\bm{\xi}_{R}+\bm{n}\bm{\xi}_{S}+\bm{n}\left(-\frac{1}{2}\bm{\xi}_{R}+\vartheta_{n}\bm{n}\right)\nonumber\\
	&=&\vartheta_{n}\bm{nn}+\bm{n}\left(\frac{1}{2}\bm{\xi}_{R}+\bm{\xi}_{S}\right)+\bm{\nabla}_{\pi}\bm{u}.
\end{eqnarray}
Therefore,~\eqref{qq8} is proved. Combining~\eqref{qq00} and~\eqref{qq11} yields~\eqref{qq9}. This completes the proof.
\end{proof}
\begin{remark}
It follows directly from~\eqref{qq8} or~\eqref{qq11} that the strain-rate tensor on $\bm{\Sigma}$ is decomposed as
\begin{eqnarray}\label{eq3p24}
\bm{D}\equiv\mathscr{S}(\bm{A})&=&\vartheta_{n}\bm{nn}+\mathscr{S}\left(\bm{n}\left(\bm{\xi}-\frac{1}{2}\bm{\xi}_{R}\right)\right)+\mathscr{S}\left(\bm{\nabla}_{\pi}\bm{u}\right)\nonumber\\
&=&\vartheta_{n}\bm{nn}+\mathscr{S}\left(\bm{n}\left(\bm{\omega}_{\pi}-\bm{W}_{\rm eff}\right)\times\bm{n}\right)+\mathfrak{D}.
\end{eqnarray} 
Equation \eqref{eq3p24} was previously derived in the context of continuum mechanics on Riemannian manifolds \citep{Xie2013688}. The last term in \eqref{eq3p24} is identified as the surface strain-rate tensor~\citep{Xie2013SCPMA}: 
\begin{equation}
	\mathfrak{D}\equiv\mathscr{S}\left(\bm{\nabla}_{\pi}\bm{u}\right)=\frac{1}{2}(\bm{\nabla}_{\pi}\bm{u}+\bm{u}\bm{\nabla}_{\pi}).
\end{equation}
The remaining terms collectively constitute the relative strain-rate tensor on a deformable surface~\citep{Chen2016Vorticity,Xie2020theory}:
\begin{equation}
\bm{D}-\mathfrak{D}=\vartheta_{n}\bm{nn}+\mathscr{S}\left(\bm{n}\left(\bm{\omega}_{\pi}-\bm{W}_{\rm eff}\right)\times\bm{n}\right).
\end{equation}
\end{remark}
\begin{remark}
From~\eqref{qq00} and~\eqref{eq320a}, we have
\begin{equation}\label{re1}
\bm{\nabla}_{\pi}\bm{u}=\left(\bm{\nabla}_{\pi}\bm{u}\right)_{\pi}-\frac{1}{2}\bm{\xi}_{R}\bm{n}.
\end{equation}
Substituting~\eqref{re1} into the first equality in~\eqref{qq11} gives
\begin{eqnarray}\label{re3p25}
\bm{A}&=&\vartheta_{n}\bm{nn}+\bm{n}\bm{\xi}-\mathscr{S}(\bm{n}\bm{\xi}_{R})+\left(\bm{\nabla}_{\pi}\bm{u}\right)_{\pi}\nonumber\\
&=&\vartheta_{n}\bm{nn}+\bm{n}(\bm{\omega}\times\bm{n})
-\left[\bm{n}(\bm{W}_{\rm eff}\times\bm{n})+(\bm{W}_{\rm eff}\times\bm{n})\bm{n}\right]+\left(\bm{\nabla}_{\pi}\bm{u}\right)_{\pi}.
\end{eqnarray}
Equation~\eqref{re3p25} recovers the result in~\citet[Eq. (3.5)]{WuJZ2005JFM}.
\end{remark}
\begin{remark}
Evaluating the symmetric part of~\eqref{qq9} yields
\begin{eqnarray}\label{CW}
		\bm{D}\equiv\mathscr{S}(\bm{A})
		&=&\mathscr{S}(\vartheta_{n}\bm{n n})+\mathscr{S}(\bm{n}\bm{\xi}_{S})+\mathscr{S}\circ\mathscr{A}\left(\bm{n}\bm{\xi}_{R}\right)+\mathscr{S}\left(\left(\bm{\nabla}_{\pi}\bm{u}\right)_{\pi}\right)\nonumber\\
		&=&\vartheta_{n}\bm{nn}+\mathscr{S}(\bm{n}\bm{\xi}_{S})+\mathscr{S}\left(\left(\bm{\nabla}_{\pi}\bm{u}\right)_{\pi}\right),
\end{eqnarray}
where the last term can be further expanded as
\begin{equation}\label{CW1}
	\mathscr{S}\left(\left(\bm{\nabla}_{\pi}\bm{u}\right)_{\pi}\right)=\mathscr{S}\left(\bm{\nabla}_{C}\bm{u}_{\pi}\right)-u_{n}\bm{K}.
\end{equation}
Equation~\eqref{CW} is referred to as the generalized Caswell formula, or the Caswell-Wu formula~\citep{WuJZ2005JFM}. It should be noted that the invariant representation~\eqref{CW1} based on the Levi-Civita gradient operator, was not explicitly given in~\citet{WuJZ2005JFM}; in that reference, only the component form under a surface-orthogonal curvilinear coordinate system was reported. In the special case of flow past a stationary rigid no-slip wall, $\bm{\omega}=\bm{S}$ and $\bm{R}=\bm{0}$ hold on $\bm{\Sigma}$, so that the Caswell formula follows straightforwardly as~\citep{Caswell1967}
\begin{eqnarray}
	\bm{D}=\vartheta_{n}\bm{nn}+\mathscr{S}(\bm{n\xi})=\vartheta_{n}\bm{nn}+\frac{1}{2}\bm{n}(\bm{\omega}\times\bm{n})+\frac{1}{2}(\bm{\omega}\times\bm{n})\bm{n}.
\end{eqnarray}
Evaluating the anti-symmetric pat of~\eqref{qq9} yields
\begin{eqnarray}
	\bm{\Omega}\equiv\mathscr{A}(\bm{A})&=&\mathscr{A}(\vartheta_{n}\bm{n n})+\mathscr{A}(\bm{n}\bm{\xi}_{S})+\mathscr{A}\circ\mathscr{A}\left(\bm{n}\bm{\xi}_{R}\right)+\mathscr{A}(\left(\bm{\nabla}_{\pi}\bm{u}\right)_{\pi})\nonumber\\
	&=&\mathscr{A}(\bm{n}\bm{\xi}_{S})+\mathscr{A}(\bm{n}\bm{\xi}_{R})+\mathscr{A}(\left(\bm{\nabla}_{\pi}\bm{u}\right)_{\pi})\nonumber\\
	&=&\mathscr{A}(\bm{n}\bm{\xi})+\mathscr{A}(\left(\bm{\nabla}_{\pi}\bm{u}\right)_{\pi}),
\end{eqnarray}
where the last term can be expressed as
\begin{eqnarray}
	\mathscr{A}(\left(\bm{\nabla}_{\pi}\bm{u}\right)_{\pi})&=&\mathscr{A}(\bm{\nabla}_{C}\bm{u}_{\pi})=*\left(\frac{1}{2}\bm{\omega}_{n}\right)\nonumber\\
	&=&\frac{1}{2}(\nabla_{\alpha}u_{\beta})\bm{g}^{\alpha}\wedge\bm{g}^{\beta}=\frac{1}{2}\left(\nabla_{\alpha}u_{\beta}-\nabla_{\beta}u_{\alpha}\right)\bm{g}^{\alpha}\bm{g}^{\beta}.
\end{eqnarray}
\end{remark}

\section{Several theorems of fluid kinematics on a deformable boundary}\label{xxxx1}
We now establish several theorems concerning fluid kinematics on a deformable boundary $\bm{\Sigma}$.
\begin{theorem}\label{thm1}
	For any tangent vector field $\bm{a}\in T\Sigma$ on the base surface $\bm{\Sigma}$, it holds that
	\begin{equation}\label{id1}
		\bm{n}\bm{\times}\left(\bm{K}\bm{\cdot}\bm{c}\right)=K\bm{a}-\bm{K}\bm{\cdot}\bm{a},
	\end{equation}
where $\bm{a}$ and $\bm{c}$ constitute an orthogonal vector pair ($\bm{c}=\bm{a}\times\bm{n}$ and $\bm{a}=\bm{n}\times\bm{c}$). 	By the first equality in~\eqref{713_ab},~\eqref{id1} can be equivalently formulated as
\begin{equation}\label{id1_equivalent}
	\bm{n}\times\left(\bm{K}\bm{\cdot}\bm{c}\right)=\hat{\bm{K}}\bm{\cdot}(\bm{n}\times\bm{c})=\hat{\bm{K}}\bm{\cdot}\bm{a}.
\end{equation}
\end{theorem}
\begin{proof}[Proof of Theorem~{\upshape\ref{thm1}}]
Under the surface coordinate system $\bm{x}=(x^1,x^2)$, we have $\bm{a}={a}_{\alpha}\bm{g}^{\alpha}$ and $\bm{c}=a_{\alpha}\varepsilon^{\alpha3\beta}\bm{g}_{\beta}$, where $\bm{\mathcal{E}}\equiv\varepsilon^{3\alpha\beta}\bm{g}_{\alpha}\bm{g}_{\beta}$ is the surface Levi-Civita permutation tensor. The surface curvature tensor is expressed as $\bm{K}=b_{\gamma\lambda}\bm{g}^{\gamma}\bm{g}^{\lambda}$. It follows that
\begin{eqnarray}
	\bm{n}\bm{\times}\left(\bm{K}\bm{\cdot}\bm{c}\right)&=&	\bm{n}\bm{\times}\left({a}_{\alpha}b_{\gamma\beta}\varepsilon^{\alpha3\beta}\bm{g}^{\gamma}\right)\nonumber\\
	&=&{a}_{\alpha}b_{\gamma\beta}\varepsilon^{3\beta\alpha}\varepsilon^{3\gamma\lambda}\bm{g}_{\lambda}\nonumber\\
	&=&{a}_{\alpha}b_{\gamma\beta}\left(g^{\beta\gamma}g^{\alpha\lambda}-g^{\beta\lambda}g^{\alpha\gamma}\right)\bm{g}^{\lambda}\nonumber\\
	&=&{a}_{\alpha}b_{\gamma\beta}g^{\beta\gamma}g^{\alpha\lambda}\bm{g}_{\lambda}-{a}_{\alpha}b_{\gamma\beta}g^{\beta\lambda}g^{\alpha\gamma}\bm{g}_{\lambda}\nonumber\\
	&=&b_{\beta}^{\beta}\left(a_\alpha\bm{g}^{\alpha}\right)-{a}_{\alpha}b_{\beta}^{\alpha}\bm{g}^{\beta}\nonumber\\
	&=&K\bm{a}-\bm{K}\bm{\cdot}\bm{a}\nonumber\\
	&=&\hat{\bm{K}}\bm{\cdot}\bm{a}=\hat{\bm{K}}\bm{\cdot}(\bm{n}\times\bm{c}).
\end{eqnarray}
This completes the proof.
\end{proof}
\begin{remark}
	Taking the dot product of $\bm{a}$ with \eqref{id1_equivalent} yields transformation identities among quadratic forms:
	\begin{equation}\label{ff1}
		\bm{c}\bm{\cdot}\bm{K}\bm{\cdot}\bm{c}=\bm{a}\bm{\cdot}\hat{\bm{K}}\bm{\cdot}\bm{a},~~\bm{c}\bm{\cdot}\hat{\bm{K}}\bm{\cdot}\bm{c}=	\bm{a}\bm{\cdot}{\bm{K}}\bm{\cdot}\bm{a};
	\end{equation}
	Due to symmetry, it readily follows that
	\begin{equation}\label{ff2}
	\bm{c}\bm{\cdot}\bm{K}\bm{\cdot}\bm{c}+\bm{a}\bm{\cdot}\bm{K}\bm{\cdot}\bm{a}=Kc^2=Ka^2.
	\end{equation}
	Moreover, for orthogonal pairs $(\bm{c}_{i},\bm{a}_{i})~(i=1,2)$ where $\bm{c}_{i}=\bm{a}_{i}\times\bm{n}$, we have
	\begin{equation}\label{ff3}
	\bm{c}_{1}\bm{\cdot}\bm{K}\bm{\cdot}\bm{c}_{2}=\bm{a}_{1}\bm{\cdot}\hat{\bm{K}}\bm{\cdot}\bm{a}_{2},~~\bm{c}_{1}\bm{\cdot}\hat{\bm{K}}\bm{\cdot}\bm{c}_{2}=	\bm{a}_{1}\bm{\cdot}{\bm{K}}\bm{\cdot}\bm{a}_{2}.
\end{equation}
\end{remark}

\begin{theorem}\label{thm2}
On the base surface $\bm{\Sigma}$, the wall-normal derivative of the full velocity gradient tensor $\bm{A}\equiv\bm{\nabla u}$ can be expressed as
\begin{subequations}
\begin{equation}\label{hh1}
	\left[\frac{\partial\bm{A}}{\partial n}\right]_{\Sigma}=\bm{K}\bm{\cdot}\bm{\nabla}_{\pi}\bm{u}+\bm{\nabla}_{\pi}\left[\frac{\partial\bm{u}}{\partial n}\right]_{\Sigma}+\bm{n}\left[\frac{\partial^2\bm{u}}{\partial n^2}\right]_{\Sigma},
\end{equation}
where the first two terms on the right-hand side are, respectively,
\begin{eqnarray}\label{hhh1}
\bm{K}\bm{\cdot}\bm{\nabla}_{\pi}\bm{u}&=&\bm{K}\bm{\cdot}\bm{\nabla}_{C}\bm{u}_{\pi}-u_{n}\bm{K}^2-\frac{1}{2}(\bm{K}\bm{\cdot}\bm{\xi}_{R})\bm{n}\nonumber\\
&=&\bm{K}\bm{\cdot}\bm{\nabla}_{C}\bm{u}_{\pi}-u_{n}\left(K\bm{K}-K_{G}\bm{G}\right)\nonumber\\
& &+\left(K\bm{K}\bm{\cdot}\bm{u}_{\pi}-K_{G}\bm{u}_{\pi}+\bm{K}\bm{\cdot}\bm{\nabla}_{\pi}u_{n}\right)\bm{n},
\end{eqnarray}
\begin{equation}\label{hh2}
\bm{\nabla}_{\pi}\left[\frac{\partial\bm{u}}{\partial n}\right]_{\Sigma}=	\left(\bm{\nabla}_{\pi}\vartheta_{n}\right)\bm{n}-\vartheta_{n}\bm{K}+\bm{\nabla}_{\pi}\left(\frac{1}{2}\bm{\xi}_{R}+\bm{\xi}_{S}\right).
\end{equation}
\end{subequations}
The effects of mean curvature and Gaussian curvature are already reflected in~\eqref{hhh1}.
\end{theorem}
\begin{proof}[Proof of Theorem~{\upshape\ref{thm2}}] In the neighborhood $\mathscr{U}\subset\mathbb{R}^3$ of the base surface $\bm{\Sigma}$, $\bm{A}$ can be expanded as
	\begin{eqnarray}
		\bm{A}\equiv\bm{\nabla u}=\bar{\bm{g}}^{i}\frac{\partial\bm{u}}{\partial x^i}=\bar{\bm{g}}^{\alpha}\frac{\partial\bm{u}}{\partial x^\alpha}+\bm{n}\frac{\partial\bm{u}}{\partial\zeta}.
	\end{eqnarray}
Then, applying the Leibniz rule to evaluate the wall-normal derivative of $\bm{A}$ gives
	\begin{eqnarray}\label{ss1}
		\frac{\partial\bm{A}}{\partial\zeta}&=&\frac{\partial}{\partial\zeta}\left(\bar{\bm{g}}^{\alpha}\frac{\partial\bm{u}}{\partial x^\alpha}+\bm{n}\frac{\partial\bm{u}}{\partial\zeta}\right)\nonumber\\
		&=&\frac{\partial\bar{\bm{g}}^{\alpha}}{\partial\zeta}\frac{\partial\bm{u}}{\partial x^\alpha}+\bar{\bm{g}}^{\alpha}\frac{\partial}{\partial\zeta}\left(\frac{\partial\bm{u}}{\partial x^\alpha}\right)+\frac{\partial\bm{n}}{\partial\zeta}\frac{\partial\bm{u}}{\partial\zeta}+\bm{n}\frac{\partial^2\bm{u}}{\partial\zeta^2}\nonumber\\
		&=&\frac{\partial\bar{\bm{g}}^{\alpha}}{\partial\zeta}\frac{\partial\bm{u}}{\partial x^\alpha}+\bar{\bm{g}}^{\alpha}\frac{\partial}{\partial x^\alpha}\left(\frac{\partial\bm{u}}{\partial\zeta}\right)+\bm{n}\frac{\partial^2\bm{u}}{\partial\zeta^2}\nonumber\\
		&=&\frac{\partial\bar{\bm{g}}^{\alpha}}{\partial\zeta}\frac{\partial\bm{u}}{\partial x^\alpha}+\overline{\bm{\nabla}}_{\pi}\left(\frac{\partial\bm{u}}{\partial\zeta}\right)+\bm{n}\frac{\partial^2\bm{u}}{\partial\zeta^2}.
		\end{eqnarray}
		
		On the base surface $\bm{\Sigma}$,~\eqref{ss1} reduces to
	\begin{eqnarray}\label{ss2}
	\left[\frac{\partial\bm{A}}{\partial\zeta}\right]_{\Sigma}=\left[\frac{\partial\bar{\bm{g}}^{\alpha}}{\partial\zeta}\right]_{\Sigma}\left[\frac{\partial\bm{u}}{\partial x^\alpha}\right]_{\Sigma}+{\bm{\nabla}}_{\pi}\left[\frac{\partial\bm{u}}{\partial\zeta}\right]_{\Sigma}+\bm{n}\left[\frac{\partial^2\bm{u}}{\partial\zeta^2}\right]_{\Sigma}.
	\end{eqnarray}
On the right-hand side of~\eqref{ss2}, the wall-normal derivative of ${\bar{\bm{g}}^{\alpha}}$ in the first term is evaluated as
\begin{eqnarray}\label{ss3}
\frac{\partial\bar{\bm{g}}^{\alpha}}{\partial\zeta}=\left(\frac{\partial\bar{\bm{g}}^{\alpha}}{\partial\zeta}\bm{\cdot}\bar{\bm{g}}_{\beta}\right)\bar{\bm{g}}^{\beta}+\left(\frac{\partial\bar{\bm{g}}^{\alpha}}{\partial\zeta}\bm{\cdot}\bm{n}\right)\bm{n}.
\end{eqnarray}
Applying the Leibniz rule with respect to the normal derivative $\partial_{\zeta}$ yields
\begin{subequations}\label{ss4ab}
\begin{eqnarray}\label{ss4a}
	\frac{\partial\bar{\bm{g}}^{\alpha}}{\partial\zeta}\bm{\cdot}\bar{\bm{g}}_{\beta}=\frac{\partial}{\partial\zeta}\left(\bar{\bm{g}}^{\alpha}\bm{\cdot}\bar{\bm{g}}_{\beta}\right)-\bar{\bm{g}}^{\alpha}\bm{\cdot}\frac{\partial\bar{\bm{g}}_{\beta}}{\partial\zeta}=\frac{\partial\delta_{\beta}^{\alpha}}{\partial\zeta}-\bar{\bm{g}}^{\alpha}\bm{\cdot}\frac{\partial\bar{\bm{g}}_{\beta}}{\partial\zeta}=-\bar{\bm{g}}^{\alpha}\bm{\cdot}\frac{\partial\bar{\bm{g}}_{\beta}}{\partial\zeta},
\end{eqnarray}
\begin{eqnarray}\label{ss4b}
\frac{\partial\bar{\bm{g}}^{\alpha}}{\partial\zeta}\bm{\cdot}\bm{n}=\frac{\partial}{\partial\zeta}\left(\bar{\bm{g}}^{\alpha}\bm{\cdot}\bm{n}\right)-\bar{\bm{g}}^{\alpha}\bm{\cdot}\frac{\partial\bm{n}}{\partial\zeta}=0.
\end{eqnarray}
\end{subequations}
On the base surface $\bm{\Sigma}$, by virtue of~\eqref{eq26a},~\eqref{ss4a} and~\eqref{ss4b},~\eqref{ss3} simplifies to
\begin{eqnarray}\label{ss5}
\left[\frac{\partial\bar{\bm{g}}^{\alpha}}{\partial\zeta}\right]_{\Sigma}=\left(-{\bm{g}}^{\alpha}\bm{\cdot}\left[\frac{\partial\bar{\bm{g}}_{\beta}}{\partial\zeta}\right]_{\Sigma}\right){\bm{g}}^{\beta}=b_{\beta}^{\alpha}\bm{g}^{\beta}=\bm{K}\bm{\cdot}\bm{g}^{\alpha}.
\end{eqnarray}
Therefore, we obtain
\begin{eqnarray}\label{ss6}
	\left[\frac{\partial\bar{\bm{g}}^{\alpha}}{\partial\zeta}\right]_{\Sigma}\left[\frac{\partial\bm{u}}{\partial x^\alpha}\right]_{\Sigma}=\bm{K}\bm{\cdot}\bm{g}^{\alpha}\frac{\partial\bm{u}}{\partial x^\alpha}=\bm{K}\bm{\cdot}\bm{\nabla}_{\pi}\bm{u}.
\end{eqnarray}
Note that $\bm{\nabla}_{\pi}\bm{u}$ in~\eqref{ss6} can be replaced by~\eqref{qq00}, the subsequent application of~\eqref{third} and~\eqref{xi_decomp2} yields~\eqref{hhh1}. It is noted that 
\begin{eqnarray}
\bm{K}\bm{\cdot}\bm{\xi}_{R}
&=&-2(\bm{K}^2\bm{\cdot}\bm{u}_{\pi}+\bm{K}\bm{\cdot}\bm{\nabla}_{\pi}u_{n})\nonumber\\
&=&-2((K\bm{K}-K_{G}\bm{G})\bm{\cdot}\bm{u}_{\pi}+\bm{K}\bm{\cdot}\bm{\nabla}_{\pi}u_{n})\nonumber\\
&=&-2(K\bm{K}\bm{\cdot}\bm{u}_{\pi}-K_{G}\bm{u}_{\pi}+\bm{K}\bm{\cdot}\bm{\nabla}_{\pi}u_{n}).
\end{eqnarray}

By applying~\eqref{qq8}, the second term on the right-hand side of~\eqref{ss2} evaluates to
\begin{eqnarray}
{\bm{\nabla}}_{\pi}\left[\frac{\partial\bm{u}}{\partial\zeta}\right]_{\Sigma}&=&\bm{\nabla}_{\pi}(\bm{n}\bm{\cdot}\bm{A})\nonumber\\
&=&\bm{\nabla}_{\pi}\left[\vartheta_{n}\bm{n}+\left(\frac{1}{2}\bm{\xi}_{R}+\bm{\xi}_{S}\right)\right]\nonumber\\
&=&\left(\bm{\nabla}_{\pi}\vartheta_{n}\right)\bm{n}+\vartheta_{n}\left(\bm{\nabla}_{\pi}\bm{n}\right)+\bm{\nabla}_{\pi}\left(\frac{1}{2}\bm{\xi}_{R}+\bm{\xi}_{S}\right)\nonumber\\
&=&\left(\bm{\nabla}_{\pi}\vartheta_{n}\right)\bm{n}-\vartheta_{n}\bm{K}+\bm{\nabla}_{\pi}\left(\frac{1}{2}\bm{\xi}_{R}+\bm{\xi}_{S}\right).
\end{eqnarray}
This completes the proof.
\end{proof}

\begin{theorem}\label{thm3}
On the base surface $\bm{\Sigma}$, the wall-normal derivative of the tangential and normal vorticity components $(\bm{\omega}_{\pi},{\omega}_{n})$ can be expressed as
\begin{subequations}
		\begin{eqnarray}\label{pp1}
		\left[\frac{\partial\bm{\omega}_{\pi}}{\partial n}\right]_{\Sigma}&=&-\bm{n}\times\bm{\nabla}_{\pi}\vartheta_{n}+\bm{K}\bm{\cdot}\bm{S}-K\bm{S}+\bm{n}\times\left[\frac{\partial^2\bm{u}_{\pi}}{\partial n^2}\right]_{\Sigma}\nonumber\\
		&=&-\bm{n}\times\bm{\nabla}_{\pi}\vartheta_{n}-\hat{\bm{K}}\bm{\cdot}\bm{S}+\bm{n}\times\left[\frac{\partial^2\bm{u}_{\pi}}{\partial n^2}\right]_{\Sigma},
	\end{eqnarray}
	\begin{eqnarray}\label{pp2}	\left[\frac{\partial{\omega}_{n}}{\partial n}\right]_{\Sigma}&=&\bm{\mathcal{E}}\bm{:}(\bm{K}\bm{\cdot}\bm{\nabla}_{C}\bm{u}_{\pi})-\bm{\nabla}_{\pi}\bm{\cdot}\left(\frac{1}{2}\bm{R}+\bm{S}\right)\nonumber\\
	&=&-\bm{\nabla}_{\pi}\bm{\cdot}\bm{R}-\bm{\nabla}_{\pi}\bm{\cdot}\bm{S}+K{\omega}_{n}\nonumber\\
	&=&-\bm{\nabla}_{\pi}\bm{\cdot}\bm{\omega}.
	\end{eqnarray}
\end{subequations}
\end{theorem}
\begin{proof}[Proof of Theorem~{\upshape\ref{thm3}}] Successively applying the antisymmetrization operator $(\mathscr{A})$ and the Hodge star operator $(\star)$ to~\eqref{hh1}, and subsequently multiplying the derived expression by a factor of 2, yields
	\begin{eqnarray}\label{pp3}
	\left[\frac{\partial\bm{\omega}}{\partial n}\right]_{\Sigma}=2\star\mathscr{A}\left(\bm{K}\bm{\cdot}\bm{\nabla}_{\pi}\bm{u}\right)+\bm{\nabla}_{\pi}\times\left[\frac{\partial\bm{u}}{\partial n}\right]_{\Sigma}+\bm{n}\times\left[\frac{\partial^2\bm{u}}{\partial n^2}\right]_{\Sigma},
	\end{eqnarray}
where $2\star\mathscr{A}\left(\bm{K}\bm{\cdot}\bm{\nabla}_{\pi}\bm{u}\right)=(\bm{K}\bm{\cdot}\bm{\nabla}_{\pi})\times\bm{u}$.
By employing~\eqref{id1} and~\eqref{hhh1}, the first term on the right-hand side of~\eqref{pp3} evaluates to
\begin{eqnarray}\label{pp5}
	2\star\mathscr{A}\left(\bm{K}\bm{\cdot}\bm{\nabla}_{\pi}\bm{u}\right)
	&=&2\star\mathscr{A}(\bm{K}\bm{\cdot}\bm{\nabla}_{C}\bm{u}_{\pi})-\frac{1}{2}2\star\mathscr{A}\left((\bm{K}\bm{\cdot}\bm{\xi}_{R})\bm{n}\right)-u_{n}2\star\mathscr{A}(\bm{K}^{2})\nonumber\\
	&=&2\star\mathscr{A}(\bm{K}\bm{\cdot}\bm{\nabla}_{C}\bm{u}_{\pi})+\frac{1}{2}\bm{n}\times\left(\bm{K}\bm{\cdot}\bm{\xi}_{R}\right)\nonumber\\
	&=&2\star\mathscr{A}(\bm{K}\bm{\cdot}\bm{\nabla}_{C}\bm{u}_{\pi})+\frac{1}{2}\left(K\bm{R}-\bm{K}\bm{\cdot}\bm{R}\right)\nonumber\\
	&=&\bm{\mathcal{E}}\bm{:}(\bm{K}\bm{\cdot}\bm{\nabla}_{C}\bm{u}_{\pi})\bm{n}+\frac{1}{2}\hat{\bm{K}}\bm{\cdot}\bm{R}.
\end{eqnarray}
Similarly, using~\eqref{hh2}, the second term on the right-hand side of~\eqref{pp3} is calculated as
\begin{eqnarray}\label{hh2a}
	\bm{\nabla}_{\pi}\times\left[\frac{\partial\bm{u}}{\partial n}\right]_{\Sigma}&=&	2\star\mathscr{A}(\left(\bm{\nabla}_{\pi}\vartheta_{n}\right)\bm{n})-2\star\vartheta_{n}\mathscr{A}(\bm{K})+2\star\mathscr{A}\bm{\nabla}_{\pi}\left(\frac{1}{2}\bm{\xi}_{R}+\bm{\xi}_{S}\right)\nonumber\\
	&=&-\bm{n}\times\bm{\nabla}_{\pi}\vartheta_{n}+\bm{\nabla}_{\pi}\times\left(\frac{1}{2}\bm{\xi}_{R}+\bm{\xi}_{S}\right).
\end{eqnarray}
The last term in~\eqref{hh2a} is expanded as
\begin{eqnarray}\label{hh3}
& &\bm{\nabla}_{\pi}\times\left(\frac{1}{2}\bm{\xi}_{R}+\bm{\xi}_{S}\right)=\bm{\nabla}_{\pi}\times\left[\left(\frac{1}{2}\bm{R}+\bm{S}\right)\times\bm{n}\right]\nonumber\\
	&=&\bm{n}\bm{\cdot}\bm{\nabla}_{\pi}\left(\frac{1}{2}\bm{R}+\bm{S}\right)-\left(\frac{1}{2}\bm{R}+\bm{S}\right)\bm{\cdot}\bm{\nabla}_{\pi}\bm{n}\nonumber\\
	& &+\left(\bm{\nabla}_{\pi}\bm{\cdot}\bm{n}\right)\left(\frac{1}{2}\bm{R}+\bm{S}\right)-\left[\bm{\nabla}_{\pi}\bm{\cdot}\left(\frac{1}{2}\bm{R}+\bm{S}\right)\right]\bm{n}\nonumber\\
	&=&\bm{K}\bm{\cdot}\left(\frac{1}{2}\bm{R}+\bm{S}\right)-K\left(\frac{1}{2}\bm{R}+\bm{S}\right)-\left[\bm{\nabla}_{\pi}\bm{\cdot}\left(\frac{1}{2}\bm{R}+\bm{S}\right)\right]\bm{n}\nonumber\\
	&=&-\hat{\bm{K}}\bm{\cdot}\left(\frac{1}{2}\bm{R}+\bm{S}\right)-\left[\bm{\nabla}_{\pi}\bm{\cdot}\left(\frac{1}{2}\bm{R}+\bm{S}\right)\right]\bm{n}.
\end{eqnarray}
Substituting~\eqref{hh3} into~\eqref{hh2a} generates
\begin{eqnarray}\label{hh4}	\bm{\nabla}_{\pi}\times\left[\frac{\partial\bm{u}}{\partial n}\right]_{\Sigma}&=&-\bm{n}\times\bm{\nabla}_{\pi}\vartheta_{n}-\left[\bm{\nabla}_{\pi}\bm{\cdot}\left(\frac{1}{2}\bm{R}+\bm{S}\right)\right]\bm{n}\nonumber\\
	& &-\hat{\bm{K}}\bm{\cdot}\left(\frac{1}{2}\bm{R}+\bm{S}\right).
\end{eqnarray}
Based on~\eqref{pp5} and~\eqref{hh4},~\eqref{pp3} becomes
\begin{eqnarray}\label{pp6}
	\left[\frac{\partial\bm{\omega}}{\partial n}\right]_{\Sigma}&=&\bm{\mathcal{E}}\bm{:}(\bm{K}\bm{\cdot}\bm{\nabla}_{C}\bm{u}_{\pi})\bm{n}+\bm{n}\times\left[\frac{\partial^2\bm{u}}{\partial n^2}\right]_{\Sigma}\nonumber\\
	& &-\bm{n}\times\bm{\nabla}_{\pi}\vartheta_{n}-\left[\bm{\nabla}_{\pi}\bm{\cdot}\left(\frac{1}{2}\bm{R}+\bm{S}\right)\right]\bm{n}-\hat{\bm{K}}\bm{\cdot}\bm{S}.
\end{eqnarray}
Taking the tangential component of~\eqref{pp6} recovers~\eqref{pp1}.

The wall-normal component of~\eqref{pp6} is expressed as
\begin{eqnarray}\label{pp7}
	\left[\frac{\partial\omega_n}{\partial n}\right]_{\Sigma}=\bm{\mathcal{E}}\bm{:}(\bm{K}\bm{\cdot}\bm{\nabla}_{C}\bm{u}_{\pi})-\bm{\nabla}_{\pi}\bm{\cdot}\left(\frac{1}{2}\bm{R}+\bm{S}\right).
\end{eqnarray}
The two terms on the right-hand side of~\eqref{pp7} evaluate to
\begin{subequations}
	\begin{equation}\label{pp8}
		\bm{\mathcal{E}}\bm{:}(\bm{K}\bm{\cdot}\bm{\nabla}_{C}\bm{u}_{\pi})=\varepsilon^{3\alpha\gamma}b^\beta_{\alpha}\nabla_{\beta}u_{\gamma},
	\end{equation}
	\begin{equation}\label{pp9}
		\begin{aligned}
	-\bm{\nabla}_{\pi}\bm{\cdot}\left(\frac{1}{2}\bm{R}\right)
	&=-\bm{\nabla}\bm{\cdot}\bm{W}_{\rm eff}\\
	&=\bm{\nabla}_{\pi}\bm{\cdot}\left[\bm{n}\times\left(\bm{K}\bm{\cdot}\bm{u}_{\pi}\right)\right]+\bm{\nabla}_{\pi}\bm{\cdot}\left(\bm{n}\times\bm{\nabla}_{\pi}u_{n}\right).
		\end{aligned}
	\end{equation}
\end{subequations}
On the right-hand side of~\eqref{pp9}, the first term is evaluated as
\begin{eqnarray}\label{pp10}
	\bm{\nabla}_{\pi}\bm{\cdot}\left[\bm{n}\times\left(\bm{K}\bm{\cdot}\bm{u}_{\pi}\right)\right]&=&\left(\bm{K}\bm{\cdot}\bm{u}_{\pi}\right)\bm{\cdot}(\bm{\nabla}_{\pi}\times\bm{n})-\bm{n}\bm{\cdot}\bm{\nabla}_{\pi}\times\left(\bm{K}\bm{\cdot}\bm{u}_{\pi}\right)\nonumber\\
	&=&-\bm{n}\bm{\cdot}\bm{\nabla}_{\pi}\times\left(\bm{K}\bm{\cdot}\bm{u}_{\pi}\right)\nonumber\\
	&=&-\bm{n}\bm{\cdot}\bm{g}^{\gamma}\times\left[\nabla_{\gamma}(b_{\alpha}^{\beta}u_{\beta})\bm{g}^{\alpha}+u_{\beta}b_{\alpha}^{\beta}b_{\gamma}^{\alpha}\bm{n}\right]\nonumber\\
	&=&-\bm{n}\bm{\cdot}\bm{g}^{\gamma}\times\bm{g}^{\alpha}\nabla_{\gamma}(b_{\alpha}^{\beta}u_{\beta})\nonumber\\
	&=&\varepsilon^{3\alpha\gamma}\nabla_{\gamma}(b_{\alpha}^{\beta}u_{\beta})\nonumber\\
	&=&\varepsilon^{3\alpha\gamma}b_{\alpha}^{\beta}\nabla_{\gamma}u_{\beta}+\varepsilon^{3\alpha\gamma}u_{\beta}\nabla_{\gamma}b_{\alpha}^{\beta}\nonumber\\
	&=&\varepsilon^{3\alpha\gamma}b_{\alpha}^{\beta}\nabla_{\gamma}u_{\beta},
\end{eqnarray}
where we have invoked the Codazzi-Mainardi equations $\nabla_{\gamma}b_{\alpha}^{\beta}=\nabla_{\alpha}b_{\gamma}^{\beta}$~\citep{chern2000lectures,docarmo2016diffgeo} and the property in~\eqref{curl_n}. The second term is evaluated as
\begin{eqnarray}\label{pp11}
	\bm{\nabla}_{\pi}\bm{\cdot}\left(\bm{n}\times\bm{\nabla}_{\pi}u_{n}\right)&=&\bm{\nabla}_{\pi}u_{n}\bm{\cdot}(\bm{\nabla}_{\pi}\times\bm{n})-\bm{n}\bm{\cdot}\bm{\nabla}_{\pi}\times(\bm{\nabla}_{\pi}u_{n})\nonumber\\
	&=&-\bm{n}\bm{\cdot}\bm{\nabla}_{\pi}\times(\bm{\nabla}_{\pi}u_{n})\nonumber\\
	&=&-\bm{n}\times\bm{g}^{\alpha}\bm{\cdot}\partial_{\alpha}(\bm{g}^{\beta}\partial_{\beta}u_{n})\nonumber\\
	&=&-\varepsilon^{3\alpha\gamma}\bm{g}_{\gamma}\bm{\cdot}(\nabla_{\alpha}(\partial_{\beta}u_{n})\bm{g}^{\beta}+b^{\beta}_{\alpha}\partial_{\beta}u_{n}\bm{n})\nonumber\\
	&=&-\varepsilon^{3\alpha\beta}\nabla_{\alpha}(\partial_{\beta}u_n)\nonumber\\
	&=&-\varepsilon^{3\alpha\beta}(\partial_{\alpha}\partial_{\beta}u_{n}-\Gamma_{\alpha\beta}^{\gamma}\partial_{\gamma}u_n)\nonumber\\
	&=&-\varepsilon^{3\alpha\beta}\partial_{\alpha}\partial_{\beta}u_{n}+\varepsilon^{3\alpha\beta}\Gamma_{\alpha\beta}^{\gamma}\partial_{\gamma}u_n\nonumber\\
	&=&0,
	\end{eqnarray}
	where use has been made of the symmetry properties ($\partial_{\alpha}\partial_{\beta}=\partial_{\beta}\partial_{\alpha}$ and $\Gamma_{\alpha\beta}^{\gamma}=\Gamma_{\beta\alpha}^{\gamma}$), as well as~\eqref{curl_n}. By combining~\eqref{pp10} and~\eqref{pp11},~\eqref{pp9} is evaluated to
	\begin{eqnarray}\label{pp12}
	-\bm{\nabla}_{\pi}\bm{\cdot}\left(\frac{1}{2}\bm{R}\right)=	\varepsilon^{3\alpha\gamma}b_{\alpha}^{\beta}\nabla_{\gamma}u_{\beta}.
	\end{eqnarray}
	The wall-normal vorticity component is 
	\begin{eqnarray}\label{pp13}
		\omega_{n}=\bm{n}\bm{\cdot}(\bm{\nabla}_{\pi}\times\bm{u}_{\pi})=\varepsilon^{3\gamma\lambda}\nabla_{\gamma}u_{\lambda}.
	\end{eqnarray}
	Multiplying~\eqref{pp13} with the permutation symbol $\varepsilon^{3\alpha\beta}$, we get
	\begin{eqnarray}\label{pp14}
		\varepsilon_{3\alpha\beta}\omega_{n}=\varepsilon_{3\alpha\beta}\varepsilon^{3\gamma\lambda}\nabla_{\gamma}u_{\lambda}=(\delta_{\alpha}^{\gamma}\delta_{\beta}^{\lambda}-\delta_{\beta}^{\gamma}\delta_{\alpha}^{\lambda})\nabla_{\gamma}u_{\lambda}=\nabla_{\alpha}u_{\beta}-\nabla_{\beta}u_{\alpha}.
	\end{eqnarray}
	Using~\eqref{pp12} and~\eqref{pp14}, the term in~\eqref{pp8} can be calculated as
	\begin{eqnarray}\label{pp15}
	\bm{\mathcal{E}}\bm{:}(\bm{K}\bm{\cdot}\bm{\nabla}_{C}\bm{u}_{\pi})&=&\varepsilon^{3\alpha\gamma}b^\beta_{\alpha}(\nabla_{\beta}u_{\gamma}-\nabla_{\gamma}u_{\beta})+\varepsilon^{3\alpha\gamma}b_{\alpha}^{\beta}\nabla_{\gamma}u_{\beta}\nonumber\\
	&=&\varepsilon^{3\alpha\gamma}b^\beta_{\alpha}\varepsilon_{3\beta\gamma}\omega_{n}+\varepsilon^{3\alpha\gamma}b_{\alpha}^{\beta}\nabla_{\gamma}u_{\beta}\nonumber\\
	&=&(\delta^{\alpha}_{\beta}\delta^{\gamma}_{\gamma}-\delta^{\gamma}_{\beta}\delta^{\alpha}_{\gamma})b^{\beta}_{\alpha}\omega_{n}+\varepsilon^{3\alpha\gamma}b_{\alpha}^{\beta}\nabla_{\gamma}u_{\beta}\nonumber\\
	&=&b_{\alpha}^{\alpha}\omega_{n}+\varepsilon^{3\alpha\gamma}b_{\alpha}^{\beta}\nabla_{\gamma}u_{\beta}\nonumber\\
	&=&K\omega_{n}-\bm{\nabla}_{\pi}\bm{\cdot}\left(\frac{1}{2}\bm{R}\right)\nonumber\\
	&=&-\bm{\nabla}_{\pi}\bm{\cdot}\left(\bm{\omega}_{n}+\frac{1}{2}\bm{R}\right),
	\end{eqnarray}
	where we have considered the fact that
	\begin{eqnarray}\label{pp16}
		-\bm{\nabla}_{\pi}\bm{\cdot}\bm{\omega}_{n}=-\bm{\nabla}_{\pi}\omega_{n}\bm{\cdot}\bm{n}-\omega_{n}\bm{\nabla}_{\pi}\bm{\cdot}\bm{n}=K\omega_{n}.
	\end{eqnarray}
	Substituting~\eqref{pp15} into~\eqref{pp7} and using~\eqref{decomp_vorticity1} recovers~\eqref{pp2}. This completes the proof.
\end{proof}

\begin{theorem}\label{thm_new1}
On the base surface $\bm{\Sigma}$, it holds that
\begin{subequations}
\begin{equation}\label{BVF_identity1}
\left[\frac{\partial\bm{\omega}}{\partial n}\right]_{\Sigma}=-\bm{n}\times\left(\bm{\nabla}\times\bm{\omega}\right)+(\bm{n}\times\bm{\nabla})\times\bm{\omega},
\end{equation}
\begin{equation}\label{BVF_identity2}
	\left[\frac{\partial\bm{\omega}_{\pi}}{\partial n}\right]_{\Sigma}=-\bm{n}\times\left(\bm{\nabla}\times\bm{\omega}\right)+\bm{\nabla}_{\pi}\bm{\omega}\bm{\cdot}\bm{n},
\end{equation}
\begin{equation}\label{BVF_identity3}
	\left[\frac{\partial\bm{\omega}_{n}}{\partial n}\right]_{\Sigma}=-(\bm{\nabla}_{\pi}\bm{\cdot}\bm{\omega})\bm{n}=\left(\bm{n}\times\bm{\nabla}\right)\times\bm{\omega}-\bm{\nabla}_{\pi}\bm{\omega}\bm{\cdot}\bm{n},
\end{equation}
\end{subequations}
along with the following identities:
\begin{subequations}
\begin{equation}\label{BVF_identity4}
	\bm{n}\times(\bm{\nabla}_{\pi}\times\bm{\omega})=\bm{\nabla}_{\pi}\bm{\omega}\bm{\cdot}\bm{n}=\bm{\nabla}_{\pi}\omega_{n}+\bm{K}\bm{\cdot}\bm{\omega}_{\pi}.
\end{equation}
\begin{equation}\label{431b}
	\bm{n}\times(\bm{\nabla}_{\pi}\times\bm{\omega}_{n})=\bm{\nabla}_{\pi}\omega_{n},
\end{equation}
\begin{equation}\label{431c}
	\bm{n}\times(\bm{\nabla}_{\pi}\times\bm{\omega}_{\pi})=\bm{K}\bm{\cdot}\bm{\omega}_{\pi},
\end{equation}
\end{subequations}
In addition, we have
\begin{subequations}\label{xx321}
	\begin{equation}\label{431f}
		(\bm{n}\times\bm{\nabla})\times\bm{\omega}=\bm{\nabla}\bm{\omega}\bm{\cdot}\bm{n}=-\bm{n}\bm{\cdot}\bm{B}_{\bm\omega},~~\bm{B}_{\bm\omega}\equiv-\bm{\omega}\bm{\nabla}=-\bm{\nabla}\bm{\omega}^{\rm T}.
	\end{equation}
	\begin{equation}\label{431d}
		(\bm{n}\times\bm{\nabla})\times\bm{\omega}_{n}=(\bm{n}\times\bm{\nabla}_{\pi})\times\bm{\omega}_{n}=\bm{\nabla}_{\pi}\omega_{n}+K\bm{\omega}_{n},
	\end{equation}
	\begin{equation}\label{431e}
		(\bm{n}\times\bm{\nabla})\times\bm{\omega}_{\pi}=(\bm{n}\times\bm{\nabla}_{\pi})\times\bm{\omega}_{\pi}=\bm{K}\bm{\cdot}\bm{\omega}_{\pi}-(\bm{\nabla}_{\pi}\bm{\cdot}\bm{\omega}_{\pi})\bm{n}.
	\end{equation}
\end{subequations}
\end{theorem}
\begin{proof}[Proof of Theorem~{\upshape\ref{thm_new1}}]
	Note that
\begin{eqnarray}\label{pp20}
	\bm{n}\times(\bm{\nabla}_{\pi}\times\bm{\omega})&=&\bm{n}\times(\bm{g}^{\alpha}\times\frac{\partial\bm{\omega}}{\partial x^\alpha})\nonumber\\
	&=&(\bm{n}\bm{\cdot}\frac{\partial\bm{\omega}}{\partial x^\alpha})\bm{g}^{\alpha}-(\bm{n}\bm{\cdot}\bm{g}^{\alpha})\frac{\partial\bm{\omega}}{\partial x^\alpha}\nonumber\\
	&=&\left(\bm{g}^{\alpha}\frac{\partial\bm{\omega}}{\partial x^\alpha}\right)\bm{\cdot}\bm{n}=\bm{\nabla}_{\pi}\bm{\omega}\bm{\cdot}\bm{n}\nonumber\\
	&=&\bm{\nabla}_{\pi}(\bm{\omega}\bm{\cdot}\bm{n})-\bm{\omega}\bm{\cdot}\bm{\nabla}_{\pi}\bm{n}\nonumber\\
	&=&\bm{\nabla}_{\pi}\omega_{n}+\bm{K}\bm{\cdot}\bm{\omega}_{\pi}.
\end{eqnarray}
Therefore,~\eqref{BVF_identity4} is proved. Equation~\eqref{BVF_identity2} holds because
\begin{eqnarray}
	\bm{n}\times\left(\bm{\nabla}\times\bm{\omega}\right)&=&	\bm{n}\times\left(\bm{\nabla}_{\pi}\times\bm{\omega}+\bm{n}\times\frac{\partial\bm{\omega}}{\partial n}\right)\nonumber\\
	&=&\bm{n}\times(\bm{\nabla}_{\pi}\times\bm{\omega})+\frac{\partial\omega_{n}}{\partial n}-\frac{\partial\bm{\omega}}{\partial n}\nonumber\\
	&=&\bm{\nabla}_{\pi}\bm{\omega}\bm{\cdot}\bm{n}-\frac{\partial\bm{\omega}_{\pi}}{\partial n}.
\end{eqnarray}
 Equation~\eqref{BVF_identity3} holds because
 \begin{eqnarray}\label{eq437a}
 (\bm{n}\times\bm{\nabla})\times\bm{\omega}&=&(\bm{n}\times\bm{g}^{\alpha})\times\frac{\partial\bm{\omega}}{\partial x^\alpha}\nonumber\\
 &=&\left(\bm{n}\bm{\cdot}\frac{\partial\bm{\omega}}{\partial x^\alpha}\right)\bm{g}^{\alpha}-\left(\bm{g}^{\alpha}\bm{\cdot}\frac{\partial\bm{\omega}}{\partial x^\alpha}\right)\bm{n}\nonumber\\
 &=&\bm{\nabla}_{\pi}\bm{\omega}\bm{\cdot}\bm{n}-(\bm{\nabla}_{\pi}\bm{\cdot}\bm{\omega})\bm{n}\nonumber\\
 &=&\bm{\nabla}_{\pi}\bm{\omega}\bm{\cdot}\bm{n}+\frac{\partial\bm{\omega}_n}{\partial n}.
 \end{eqnarray}
 The sum of~\eqref{BVF_identity2} and~\eqref{BVF_identity3} yields~\eqref{BVF_identity1}. On one hand,~\eqref{431f} holds because
 \begin{eqnarray}
  (\bm{n}\times\bm{\nabla})\times\bm{\omega}&=&\bm{\nabla}_{\pi}\bm{\omega}\bm{\cdot}\bm{n}+\bm{n}\frac{\partial\bm{\omega}}{\partial n}\bm{\cdot}\bm{n}\nonumber\\
  &=&\left(\bm{\nabla}_{\pi}\bm{\omega}+\bm{n}\frac{\partial\bm{\omega}}{\partial n}\right)\bm{\cdot}\bm{n}\nonumber\\
  &=&\bm{\nabla}\bm{\omega}\bm{\cdot}\bm{n}.
  \end{eqnarray}
 On the other hand,~\eqref{eq437a} can be expanded as
 \begin{eqnarray}\label{eq438a}
 (\bm{n}\times\bm{\nabla})\times\bm{\omega}=\bm{\nabla}_{\pi}\omega_{n}+\bm{K}\bm{\cdot}\bm{\omega}_{\pi}+(-\bm{\nabla}_{\pi}\bm{\cdot}\bm{\omega}_{\pi}+K\omega_{n})\bm{n}.
 \end{eqnarray}
 Letting $\bm{\omega}_{\pi}=\bm{0}$ and $\bm{\omega}_{n}=\bm{0}$ gives~\eqref{431d} and~\eqref{431e}, respectively.
 This completes the proof.
\end{proof}
\begin{remark}
For any smooth tensor field $\bm{\Phi}\in\mathscr{T}^{p}(\mathbb{R}^3)~(p\in\mathbb{N}^+)$ on $\bm{\Sigma}$, the following identity holds~\citep{Xie2013688}
\begin{equation}\label{Xie_id1}
(\bm{n}\times\bm{\nabla}_{\pi})\bm{\cdot}(\bm{n}\times\bm{\Phi})=\bm{\nabla}_{\pi}\bm{\cdot}\bm{\Phi}+K\bm{n}\bm{\cdot}\bm{\Phi}.
\end{equation}
If $\bm{n}\bm{\cdot}\bm{\Phi}=\bm{0}$, then~\eqref{Xie_id1} reduces to~\citet[Eq. (A5)]{Wu1995}.
Taking $\bm{\Phi}=\bm{\omega}_{\pi}\otimes\bm{n}\in\mathscr{T}^{2}(\mathbb{R}^3)$ in~\eqref{Xie_id1} and using~\eqref{431e}, we obtain
\begin{eqnarray}
(\bm{n}\times\bm{\nabla}_{\pi})\bm{\cdot}(\bm{n}\times\bm{\omega}_{\pi}\bm{n})&=&\bm{\nabla}_{\pi}\bm{\cdot}(\bm{\omega}_{\pi}\bm{n})+K\bm{n}\bm{\cdot}\bm{\omega}_{\pi}\bm{n}\nonumber\\
&=&\bm{\nabla}_{\pi}\bm{\cdot}(\bm{\omega}_{\pi}\bm{n})\nonumber\\
&=&(\bm{\nabla}_{\pi}\bm{\cdot}\bm{\omega}_{\pi})\bm{n}-\bm{K}\bm{\cdot}\bm{\omega}_{\pi}\nonumber\\
&=&-(\bm{n}\times\bm{\nabla}_{\pi})\times\bm{\omega}_{\pi}.
\end{eqnarray}
Using $\bm{\xi}=\bm{\omega}_{\pi}\times\bm{n}$, this results in
\begin{equation}\label{phase1}
\underline{(\bm{n}\times\bm{\nabla}_{\pi})\bm{\cdot}(\bm{\xi}\bm{n})=(\bm{n}\times\bm{\nabla}_{\pi})\times\bm{\omega}_{\pi}.}
\end{equation}
Taking $\bm{\Phi}=\omega_{n}\bm{G}\in\mathscr{T}^{2}(T\Sigma)\subset\mathscr{T}^{2}(\mathbb{R}^3)$ in~\eqref{Xie_id1}, and utilizing relations~\eqref{431d} and~\eqref{div_G}, we obtain
\begin{eqnarray}
(\bm{n}\times\bm{\nabla}_{\pi})\bm{\cdot}(\bm{n}\times\omega_{n}\bm{G})&=&\bm{\nabla}_{\pi}\bm{\cdot}(\omega_{n}\bm{G})+K\bm{n}\bm{\cdot}(\omega_{n}\bm{G})\nonumber\\
&=&\bm{\nabla}_{\pi}\omega_{n}\bm{\cdot}\bm{G}+\omega_{n}\bm{\nabla}\bm{\cdot}\bm{G}\nonumber\\
&=&\bm{\nabla}_{\pi}\omega_{n}+K\bm{\omega}_{n}\nonumber\\
&=&(\bm{n}\times\bm{\nabla}_{\pi})\times\bm{\omega}_{n},
\end{eqnarray}
This results in
\begin{equation}\label{phase2}
\underline{(\bm{n}\times\bm{\nabla}_{\pi})\bm{\cdot}(\bm{n}\times\omega_{n}\bm{G})=(\bm{n}\times\bm{\nabla}_{\pi})\times\bm{\omega}_{n}.}
\end{equation}
In conjunction with the above derivation, the sum of~\eqref{phase1} and~\eqref{phase2} gives
\begin{eqnarray}
-\bm{n}\bm{\cdot}\bm{B}_{\bm\omega}&=&(\bm{n}\times\bm{\nabla}_{\pi})\times\bm{\omega}\nonumber\\
&=&\bm{\nabla}_{\pi}\bm{\omega}\bm{\cdot}\bm{n}-(\bm{\nabla}_{\pi}\bm{\cdot}\bm{\omega})\bm{n}\nonumber\\
&=&(\bm{n}\times\bm{\nabla}_{\pi})\bm{\cdot}(\bm{n}\times(\omega_{n}\bm{G}-\bm{\omega}_{\pi}\bm{n}))\nonumber\\
&=&\bm{\nabla}_{\pi}\bm{\cdot}(\omega_{n}\bm{G}-\bm{\omega}_{\pi}\bm{n})\nonumber\\
&=&\bm{\nabla}_{\pi}\omega_{n}+\bm{K}\bm{\cdot}\bm{\omega}_{\pi}-(\bm{\nabla}_{\pi}\bm{\cdot}\bm{\omega})\bm{n}.
\end{eqnarray}
\end{remark}	
\begin{remark}
Interestingly, when discussing the material time derivative of the directed material surface element $\delta\bm{\Sigma}$, similar mathematical structures also arise for $\bm{n}\bm{\cdot}\bm{B}$ on $\bm{\Sigma}$ (see~\eqref{SS1} and~\citet{Wu1995} for details):
\begin{eqnarray}\label{ggg1}
	-\bm{n}\bm{\cdot}\bm{B}&=&(\bm{n}\times\bm{\nabla}_{\pi})\times\bm{u}\nonumber\\
	&=&\bm{\nabla}_{\pi}\bm{u}\bm{\cdot}\bm{n}-(\bm{\nabla}_{\pi}\bm{\cdot}\bm{u})\bm{n}\nonumber\\
	&=&(\bm{n}\times\bm{\nabla}_{\pi})\bm{\cdot}(\bm{n}\times(u_{n}\bm{G}-\bm{u}_{\pi}\bm{n}))\nonumber\\
	&=&\bm{\nabla}_{\pi}\bm{\cdot}(u_{n}\bm{G}-\bm{u}_{\pi}\bm{n})\nonumber\\
	&=&\bm{\nabla}_{\pi}{u}_{n}+\bm{K}\bm{\cdot}\bm{u}_{\pi}-(\bm{\nabla}_{\pi}\bm{\cdot}\bm{u})\bm{n}.
\end{eqnarray}
In analogy to~\eqref{ggg1} and noting~\eqref{dishington}, we formally define $\bm{B}_{\bm\omega}$ as presented in~\eqref{431f}:
\begin{equation}
\bm{B}_{\bm\omega}\equiv(\bm{\nabla}\bm{\cdot}\bm{\omega})\bm{I}-\bm{\omega}\bm{\nabla}=-\bm{\omega}\bm{\nabla}=-\bm{\nabla\omega}^{\rm T},
\end{equation}
where we have used the divergence-free condition $\bm{\nabla}\bm{\cdot}\bm{\omega}=0$.
Since $\bm{\nabla}\bm{\cdot}(\bm{\omega}\bm{\nabla})=\bm{\nabla}(\bm{\nabla}\bm{\cdot}\bm{\omega})=\bm{0}$, the tensor $\bm{B}_{\bm\omega}$ exerts no impact on the material derivative of vorticity for constant viscosity fluids. It will be demonstrated that $\bm{B}_{\bm\omega}$ is responsible for the existence of two alternative interpretations of vorticity dynamics as well as the 3D viscous contribution to Lighthill-Panton-Wu BVF, as elaborated in Section~\ref{def_expre_BVFs}.
\end{remark}
\begin{remark}
Consider a material volume $\mathscr{V}$ enclosed by the boundary surface $\partial\mathscr{V}$, with the outward-directed unit normal $\hat{\bm{n}}=-\bm{n}$.
The generalized Gauss theorem implies that
\begin{eqnarray}
-\oint_{\partial\mathscr{V}}\bm{n}\bm{\cdot}\bm{B}_{\bm\omega}dS=\int_{\mathscr{V}}\bm{\nabla}\bm{\cdot}\bm{B}_{\bm\omega}dV=\bm{0},
\end{eqnarray}
Alternatively, applying the generalized Stokes theorem yields
\begin{equation}
-\oint_{\partial\mathscr{V}}\bm{n}\bm{\cdot}\bm{B}_{\bm\omega}dS=-\oint_{\partial\mathscr{V}}(\hat{\bm{n}}\times\bm{\nabla}_{\pi})\bm{\times}\bm{\omega}dS=-\oint_{\partial(\partial\mathscr{V})=\emptyset}d\bm{l}\bm{\times}\bm{\omega}=\bm{0}.
\end{equation}
In other words, $\bm{B}_{\bm\omega}$ produces no net vorticity flux across the boundary.
\end{remark}

\begin{theorem}\label{thm_new2}
On the base surface $\bm{\Sigma}$, it holds that
\begin{subequations}
	\begin{equation}\label{pp21}
		\left[\frac{\partial\bm{R}}{\partial n}\right]_{\Sigma}=-\bm{n}\times\left(\bm{\nabla}\times\bm{R}\right)+\bm{K}\bm{\cdot}\bm{R},
	\end{equation}
	\begin{equation}\label{pp22}
		\left[\frac{\partial\bm{S}}{\partial n}\right]_{\Sigma}=-\bm{n}\times\left(\bm{\nabla}\times\bm{S}\right)+\bm{K}\bm{\cdot}\bm{S}.
	\end{equation}
\end{subequations}
\end{theorem}
\begin{proof}[Proof of Theorem~{\upshape\ref{thm_new2}}] Because of $R_{n}=\bm{R}\bm{\cdot}\bm{n}=0$, we obtain
	\begin{eqnarray}\label{pp25}
\bm{n}\times(\bm{\nabla}_{\pi}\times\bm{R})
=\bm{\nabla}_{\pi}R_{n}+\bm{K}\bm{\cdot}\bm{R}
=\bm{K}\bm{\cdot}\bm{R}.
	\end{eqnarray}
Therefore,~\eqref{pp21} holds due to
\begin{eqnarray}\label{pp26}
\bm{n}\times(\bm{\nabla}\times\bm{R})
&=&	\bm{n}\times(\bm{\nabla}_{\pi}\times\bm{R})-\frac{\partial\bm{R}}{\partial n}\nonumber\\
&=&\bm{K}\bm{\cdot}\bm{R}-\frac{\partial\bm{R}}{\partial n}.
\end{eqnarray}
Equation~\eqref{pp22} can be proved in a similar fashion. This completes the proof.
\end{proof}
\begin{remark}
In a small vicinity of the base surface $\bm{\Sigma}$, we introduce 
\begin{subequations}
	\begin{equation}
		\bm{B}_{\bm\omega}=\bm{B}_{\bm R}+\bm{B}_{\bm S},
	\end{equation}
	\begin{equation}
		\bm{B}_{\bm R}\equiv(\bm{\nabla}\bm{\cdot}\bm{R})\bm{I}-\bm{R}\bm{\nabla},~\bm{B}_{\bm S}\equiv(\bm{\nabla}\bm{\cdot}\bm{S})\bm{I}-\bm{S}\bm{\nabla}.
	\end{equation}
\end{subequations}
Then, it follows that
	\begin{subequations}
		\begin{equation}\label{pp23}
		-\bm{n}\bm{\cdot}\bm{B}_{\bm R}=(\bm{n}\times\bm{\nabla}_{\pi})\times\bm{R}=\bm{K}\bm{\cdot}\bm{R}-(\bm{\nabla}_{\pi}\bm{\cdot}\bm{R})\bm{n},
		\end{equation}
		\begin{equation}\label{pp24}
		-\bm{n}\bm{\cdot}\bm{B}_{\bm S}=(\bm{n}\times\bm{\nabla}_{\pi})\times\bm{S}=\bm{K}\bm{\cdot}\bm{S}-(\bm{\nabla}_{\pi}\bm{\cdot}\bm{S})\bm{n}.
		\end{equation}
	\end{subequations}
\end{remark}

\section{Definitions and expressions of boundary vorticity fluxes}\label{def_expre_BVFs}
\subsection{Vorticity current tensor and vorticity transport equation}
For compressible viscus flows, the Navier-Stokes (NS) equations are expressed as~\citep{navier1827memoire,stokes1845theories}
\begin{eqnarray}\label{NSeq}
	\frac{\partial\bm{u}}{\partial t}+\bm{u}\bm{\cdot}\bm{\nabla}\bm{u}=-\frac{1}{\rho}\bm{\nabla}p+\frac{1}{\rho}\bm{\nabla}\bm{\cdot}\bm{T}+\bm{f},
\end{eqnarray}
where $\rho$ is the fluid density, $p$ is the pressure, and $\bm{f}$ is the body force per unit mass. The left-hand side equals to the fluid acceleration $\bm{a}\equiv D\bm{u}/Dt$. For Newtonian fluids, the viscous stress tensor is given by
\begin{equation}\label{viscous_stress}
	\bm{T}=2\mu\bm{D}+\left(\mu_{b}-\frac{2}{3}\mu\right){\vartheta}\bm{I},
\end{equation}
where $\mu$ is the dynamic viscosity, and $\mu_b$ the bulk viscosity. The longitudinal viscosity is defined as $\mu_{\vartheta}\equiv\mu_{b}+4\mu/3$. The corresponding kinematic viscosities are denoted by $\nu_{\vartheta}$, $\nu$ and $\nu_{b}$ after division by the density $\rho$. For simplicity, while retaining the essential physics, the linear diffusion approximation introduced by \citet{Lighthill1956viscosity} is adopted in the theoretical derivation, so that all viscosities are assumed constant. Alternatively, the terms caused by varying viscosity can also be incorporated in the definition of $\bm{f}$. Since $\bm{\nabla}\bm{\cdot}\bm{B}=\bm{0}$ holds for the surface deformation tensor $\bm{B}$, the divergence of $\bm{T}$ is evaluated as~\citep{Wu2006vorticity}
\begin{eqnarray}\label{div_T}	\bm{\nabla}\bm{\cdot}\bm{T}&=&\bm{\nabla}\left(\mu_{\vartheta}\vartheta\right)-\bm{\nabla}\times\left(\mu\bm{\omega}\right)-2\bm{\nabla}\mu\bm{\cdot}\bm{B}\nonumber\\
&=&\mu_{\vartheta}\bm{\nabla}{\vartheta}-\mu\bm{\nabla}\times\bm{\omega}.
\end{eqnarray}
In addition, the Lamb identity gives the following expression for the convective acceleration:
\begin{equation}\label{Lamb_identity}
\bm{u}\bm{\cdot}\bm{\nabla u}=\bm{\omega}\times\bm{u}+\bm{\nabla}\left(\frac{1}{2}u^2\right).
\end{equation}
Substituting~\eqref{div_T} and~\eqref{Lamb_identity} into~\eqref{NSeq} yields
\begin{equation}\label{NS2}
	\bm{a}=\frac{\partial\bm{u}}{\partial t}+\bm{\omega}\times\bm{u}+\bm{\nabla}\left(\frac{1}{2}u^2\right)=-\bm{\nabla}\hat{P}-\nu\bm{\nabla}\times\bm{\omega}+\bm{f},
\end{equation}
where the composite pressure $P\equiv{p}-\mu_{\vartheta}\vartheta$ includes the dilatation modification, and $\bm{\omega}\times\bm{u}$ is the Lamb vector. For the present purpose, we neglect the density variation in the pressure gradient term and assume $-\rho^{-1}\bm{\nabla}P=\bm{\nabla}\hat{P}$ with $\hat{P}\equiv P/\rho$. Consequently, the baroclinic torque term $\rho^{-2}\bm{\nabla}\rho\times\bm{\nabla}p$ is absent from the vorticity transport equation. Alternatively, the terms due to the density and viscosity variations can be included in the force term.

The curl of the Lamb vector can be written as
\begin{eqnarray}
\bm{\nabla}\times(\bm{\omega}\times\bm{u})&=&\bm{u}\bm{\cdot}\bm{\nabla}\bm{\omega}-\bm{\omega}\bm{\cdot}\bm{\nabla}\bm{u}+\left(\bm{\nabla}\bm{\cdot}\bm{u}\right)\bm{\omega}-(\bm{\nabla}\bm{\cdot}\bm{\omega})\bm{u}\nonumber\\
&=&\left[\bm{u}\bm{\cdot}\bm{\nabla}\bm{\omega}+\left(\bm{\nabla}\bm{\cdot}\bm{u}\right)\bm{\omega}\right]-\left[\bm{\omega}\bm{\cdot}\bm{\nabla}\bm{u}+(\bm{\nabla}\bm{\cdot}\bm{\omega})\bm{u}\right]\nonumber\\
&=&\bm{\nabla}\bm{\cdot}\left(\bm{u\omega}-\bm{\omega u}\right)\nonumber\\
&=&\bm{\nabla}\bm{\cdot}\bm{J}_{\rm inv},
\end{eqnarray}
where the inviscid vorticity current tensor is introduced as
\begin{eqnarray}
\bm{J}_{\rm inv}\equiv\bm{u\omega}-\bm{\omega}\bm{u}.
\end{eqnarray}
However, since $\bm{\nabla}\bm{\cdot}\bm{\omega}=0$, these exist two definitions of the viscous vorticity current tensor~\citep{Terrington2023LH}:
\begin{eqnarray}
-\nu\bm{\nabla}\times(\bm{\nabla}\times\bm{\omega})=\nu\nabla^2\bm{\omega}=-\bm{\nabla}\bm{\cdot}\bm{J}_{\rm vis}^{(1)}=-\bm{\nabla}\bm{\cdot}\bm{J}_{\rm vis}^{(2)},
\end{eqnarray}
Here, the definition adopted by~\citet{Lighthill1963},~\citet{Panton1984}, and~\citet{Wu1993} is
\begin{eqnarray}\label{Jvis_LPW}
\bm{J}_{\rm vis}^{(1)}=-\nu\bm{\nabla\omega},
\end{eqnarray}
while the alternative definition reads~\citep{huggins1970exact,huggins1971dynamical,huggins1994vortex}
\begin{eqnarray}\label{Jvis_LH}
	\bm{J}_{\rm vis}^{(2)}\equiv-\nu\left(\bm{\nabla\omega}-\bm{\omega\nabla}\right).
\end{eqnarray}
Based on the historical evolution, $\bm{J}_{\rm vis}^{(1)}$ is referred to as the Lighthill-Panton-Wu vorticity current tensor, whereas $\bm{J}_{\rm vis}^{(2)}$ is designated as the Huggins vorticity current tensor.
Using~\eqref{431f}, the two viscous vorticity current tensors must satisfy
\begin{eqnarray}
\bm{J}_{\rm vis}^{(2)}=\bm{J}_{\rm vis}^{(1)}-\nu\bm{B}_{\bm\omega},~~\text{where}~~\bm{\nabla}\bm{\cdot}\bm{B}_{\bm\omega}=\bm{0}.
\end{eqnarray}
This indicates that the tensor $\bm{B}_{\bm\omega}$ reveals the non-uniqueness of $\bm{J}_{\rm vis}$. Replacing $\bm{J}_{\rm vis}$ with $\bm{J}_{\rm vis}+C\bm{B}_{\bm\omega}~(C~\text{is a real constant})$ does not alter the result.
Consequently, the vorticity transport equation can be concisely written as
\begin{eqnarray}\label{vorticity_eq}
	\frac{\partial\bm{\omega}}{\partial t}=-\bm{\nabla}\bm{\cdot}\bm{J}_{\rm inv}-\bm{\nabla}\bm{\cdot}\bm{J}_{\rm vis},
\end{eqnarray}
where $\bm{J}_{\rm vis}$ can takes the form of $\bm{J}_{\rm vis}^{(1)}$ in~\eqref{Jvis_LPW} or $\bm{J}_{\rm vis}^{(2)}$ in~\eqref{Jvis_LH}.

\subsection{Two alternative definitions of boundary vorticity flux}
By integrating~\eqref{vorticity_eq} over a fixed control volume $\mathscr{V}$ bounded by its surface $\partial\mathscr{V}$, and subsequently applying the Gauss-Ostrogradsky divergence theorem, one can derive the following integral conservation law for vorticity~\citep{Terrington2023LH}:
\begin{eqnarray}\label{integral}
	\frac{d}{dt}\int_{\mathscr V}\bm{\omega}dV=-\oint_{\partial \mathscr{V}}\hat{\bm{n}}\bm{\cdot}\bm{J}_{\rm inv}dS-\oint_{\partial \mathscr{V}}\hat{\bm{n}}\bm{\cdot}\bm{J}_{\rm vis}dS,
\end{eqnarray}
where $\hat{\bm{n}}=-\bm{n}$ denotes the outward-pointing unit normal vector on $\partial{\mathscr{V}}$. On the right-hand side of~\eqref{integral}, the first and second terms correspond, respectively, to the total inviscid and viscous fluxes crossing the control-volume surface. The first term, due to vorticity convection and stretching, gives
\begin{equation}
-\oint_{\partial \mathscr{V}}\hat{\bm{n}}\bm{\cdot}\bm{J}_{\rm inv}dS=\oint_{\partial \mathscr{V}}\left[(\bm{n}\bm{\cdot}\bm{u})\bm{\omega}-(\bm{n}\bm{\cdot}\bm{\omega})\bm{u}\right]dS.
\end{equation}
The integrand of the second term defines the boundary vorticity flux (BVF), given explicitly by
\begin{equation}\label{BVF_def}
	\bm{\sigma}_{\bm{\omega}}\equiv\hat{\bm{n}}\bm{\cdot}\bm{J}_{\rm vis}=-\bm{n}\bm{\cdot}\bm{J}_{\rm vis}.
\end{equation}
\begin{itemize}
	\item From~\eqref{Jvis_LPW} and~\eqref{BVF_def}, we obtain the Lighthill-Panton-Wu BVF~\citep{Lighthill1963,Panton1984,Wu1986,Wu1993}:
	\begin{eqnarray}\label{BVF1}
		\bm{\sigma}_{\bm\omega}^{(1)}\equiv-\bm{n}\bm{\cdot}\bm{J}_{\rm vis}^{(1)}=\nu\bm{n}\bm{\cdot}\bm{\nabla\omega}=\nu\frac{\partial\bm{\omega}}{\partial n}~~\text{on}~~\partial\mathscr{V}.
	\end{eqnarray}
	\item From~\eqref{Jvis_LH} and~\eqref{BVF_def}, we obtain the Lyman-Huggins BVF (or simply, the Lyman flux)~\citep{Lyman1990vorticity}:
	\begin{equation}\label{BVF2}
		\bm{\sigma}_{\bm\omega}^{(2)}\equiv-\bm{n}\bm{\cdot}\bm{J}_{\rm vis}^{(2)}=-\nu\bm{n}\times\left(\bm{\nabla}\times\bm{\omega}\right)~~\text{on}~~\partial\mathscr{V}.
	\end{equation}
\end{itemize}
When $\partial\mathscr{V}$ contains a solid boundary segment $\bm{\Sigma}$, the BVF measures the local vorticity creation rate per unit area and thereby describes the vorticity source strength, since no fluid vorticity exists on the opposite side of the boundary. Although the two BVFs generally yield different local vorticity creation rates along a boundary due to $\bm{n}\bm{\cdot}\bm{B}_{\bm\omega}$, their respective integrals over any closed surface $\partial\mathscr{V}$ are equal to the negative of the volume integral of the vorticity diffusion term:
\begin{eqnarray}
	\oint_{\partial \mathscr{V}}\bm{\sigma}_{\bm\omega}^{(1)}dS=\oint_{\partial \mathscr{V}}\bm{\sigma}_{\bm\omega}^{(2)}dS=-\int_{\mathscr V}\nu\nabla^2\bm{\omega}dV.
\end{eqnarray}
\subsection{Full expressions of two boundary vorticity fluxes}
The full expression of the BVF $\bm{\sigma}_{\bm\omega}^{(1)}$ was derived by~\citet{Wu1993} and~\citet{WuWu1996}, extending the work of~\citet{Lighthill1963} for 2D flow past a stationary wall to an arbitrarily moving and deforming boundary $\bm{\Sigma}$. For a clearer exposition of the theoretical structure, we decompose $\bm{\sigma}_{\bm\omega}^{(1)}$ by applying~\eqref{BVF_identity1} on $\bm{\Sigma}$, yielding
\begin{subequations}
	\begin{equation}\label{Lyman_flux1}
		\bm{\sigma}_{\bm\omega}^{(1)}\equiv\nu\left[\frac{\partial\bm{\omega}}{\partial n}\right]=\bm{\sigma}_{\bm\omega}^{(2)}+\bm{\sigma}_{\rm vis},
	\end{equation}
	\begin{equation}\label{sigma_vis}
		\bm{\sigma}_{\bm\omega}^{(2)}\equiv-\nu\bm{n}\times\left(\bm{\nabla}\times\bm{\omega}\right),~\bm{\sigma}_{\rm vis}\equiv\nu(\bm{n}\times\bm{\nabla})\times\bm{\omega}.
	\end{equation}
\end{subequations}
Clearly, the Lyman-Huggins BVF $\bm{\sigma}_{\bm\omega}^{(2)}$~\citep{Lyman1990vorticity} constitutes a part of the Lighthill-Panton-Wu BVF $\bm{\sigma}_{\bm\omega}^{(1)}$. By applying~\eqref{NS2} to the boundary $\bm{\Sigma}$ and imposing the no-slip boundary condition, the explicit expression for $\bm{\sigma}_{\bm\omega}^{(2)}$ is obtained as
\begin{equation}\label{Wub}
	\bm{\sigma}_{\bm\omega}^{(2)}=\bm{n}\times(\bm{a}-\bm{f})+\bm{n}\times\bm{\nabla}_{\pi}\hat{P}.
\end{equation}
Equation~\eqref{Wub} reveals that vorticity is generated on a boundary through tangential acceleration, tangential external force, and surface pressure gradient, with the aid of viscosity and the no-slip boundary condition.
By contrast, $\bm{\sigma}_{\rm vis}$ is directly obtainable from the kinematic relation~\eqref{xx321}, taking the form
\begin{eqnarray}\label{Wuc}
	\bm{\sigma}_{\rm vis}&=&-\nu\bm{n}\bm{\cdot}\bm{B}_{\bm\omega}=\nu\bm{\nabla}\bm{\omega}\bm{\cdot}\bm{n}\nonumber\\
	&=&\nu\bm{\nabla}_{\pi}\bm{\omega}\bm{\cdot}\bm{n}-\nu(\bm{\nabla}_{\pi}\bm{\cdot}\bm{\omega})\bm{n}\nonumber\\
	&=&\nu\bm{\nabla}_{\pi}\omega_{n}+\nu\bm{K}\bm{\cdot}\bm{\omega}_{\pi}-\nu(\bm{\nabla}_{\pi}\bm{\cdot}\bm{\omega})\bm{n}.
\end{eqnarray}
Equation~\eqref{Wuc} represents the 3D viscous contributions to $\bm{\sigma}_{\bm\omega}^{(1)}$, which arise from three distinct sources: the surface tangential gradient of the wall-normal vorticity component, the coupling between the surface curvature tensor and the surface tangential vorticity component, as well as the surface vorticity divergence. This indicates that the boundary vorticity field can induce significant tangential BVF in regions of large surface curvature, such as along the side edge of a wing.

From~\eqref{BVF_identity2} and~\eqref{BVF_identity3},  an orthogonal decomposition of $\bm{\sigma}_{\bm\omega}^{(1)}$ is obtained as
\begin{subequations}\label{520abc}
	\begin{equation}\label{ggg}
	\bm{\sigma}_{\bm\omega}^{(1)}=\bm{\sigma}_{\pi}^{(1)}+\bm{\sigma}_{n}^{(1)},
	\end{equation}
	\begin{eqnarray}\label{BVF_identity2p}
\bm{\sigma}_{\pi}^{(1)}&\equiv&\nu\left[\frac{\partial\bm{\omega}_{\pi}}{\partial n}\right]_{\Sigma}=\bm{\sigma}_{\bm\omega}^{(2)}+(\bm{\sigma}_{\rm vis})_{\pi}\nonumber\\
&=&\bm{n}\times(\bm{a}-\bm{f})+\bm{n}\times\bm{\nabla}_{\pi}\hat{P}+\nu\bm{\nabla}_{\pi}{\omega}_{n}+\nu\bm{K}\bm{\cdot}\bm{\omega}_{\pi},
	\end{eqnarray}
	\begin{equation}\label{BVF_identity3p}
	\bm{\sigma}_{n}^{(1)}\equiv\nu\left[\frac{\partial\bm{\omega}_{n}}{\partial n}\right]_{\Sigma}=(\bm{\sigma}_{\rm vis})_{n}=-\nu(\bm{\nabla}_{\pi}\bm{\cdot}\bm{\omega})\bm{n}.
	\end{equation}
\end{subequations}
We observe that the tangential BVF $\bm{\sigma}_{\pi}^{(1)}$ encompasses the entirety of  $\bm{\sigma}_{\bm\omega}^{(2)}$ along with the tangential component of the viscous contribution, $(\bm{\sigma}_{\rm vis})_{\pi}=\nu\bm{\nabla}_{\pi}\bm{\omega}\bm{\cdot}\bm{n}$. The wall-normal BVF, defined as $\bm{\sigma}_{n}^{(1)}=(\bm{\sigma}_{\rm vis})_{n}=(\bm{\sigma}_{\rm vis}\bm{\cdot}\bm{n})\bm{n}$, represents the tilting of boundary vorticity lines toward the wall-normal
direction, and is determined solely by the solenoidal condition of vorticity. In the vicinity of a focus on a stationary wall $\bm{\Sigma}$, such as that associated with a tornado vortex, $\bm{\sigma}_{n}^{(1)}$ typically manifests with
converging boundary vorticity lines $(\bm{\omega}\text{-lines})$ and swirling skin-friction lines $(\bm{\tau}\text{-lines})$~\citep{Wu2018review,LiuTS2019PAS}.
\subsection{Simplified expressions of boundary vorticity fluxes on a stationary wall}\label{simple_BVF}
On a stationary curved wall $\bm{\Sigma}$, the velocity adherence condition $\bm{u}=\bm{0}$ implies $\bm{\omega}_{n}=\bm{0}$, $\bm{\omega}=\bm{\omega}_{\pi}$, and $\bm{a}=\bm{0}$; moreover, $\vartheta_{\pi}=0$ and $\vartheta=\vartheta_{n}$ hold for compressible flow. Consequently, we obtain
\begin{itemize}
	\item The Lighthill-Panton-Wu BVF
	\begin{equation}\label{u0_LP_BVF}
		\bm{\sigma}_{\bm\omega}^{(1)}=-\bm{n}\times\bm{f}+\bm{n}\times\bm{\nabla}_{\pi}\hat{P}+\nu\bm{K}\bm{\cdot}\bm{\omega}_{\pi}-\nu(\bm{\nabla}_{\pi}\bm{\cdot}\bm{\omega})\bm{n}.
	\end{equation}
	\item The Lyman-Huggins BVF
	\begin{equation}\label{u0_LH_BVF}
		\bm{\sigma}_{\bm\omega}^{(2)}=-\bm{n}\times\bm{f}+\bm{n}\times\bm{\nabla}_{\pi}\hat{P}.
	\end{equation}
\end{itemize}

\section{Definitions and expressions of boundary enstrophy fluxes}\label{kkk1}
\subsection{Enstrophy transport equation}
By taking the dot product of~\eqref{vorticity_eq} with $\bm{\omega}$, the enstrophy transport equation is derived as
\begin{equation}\label{eq61}
\frac{\partial\Omega}{\partial t}+\bm{u}\bm{\cdot}\bm{\nabla}\Omega=-\bm{\omega}\bm{\cdot}\bm{B}\bm{\cdot}\bm{\omega}-\bm{\nabla}\bm{\cdot}(\bm{J}_{\rm vis}\bm{\cdot}\bm{\omega})+\bm{J}_{\rm vis}\bm{:}\bm{\nabla\omega},
\end{equation}
where the left-hand side represents the material derivative of enstrophy. On the right-hand side, the first term equals to the sum of the vortex stretching term and the dilatation-enstrophy coupling term, namely,
\begin{equation}\label{eq62}
	-\bm{\omega}\bm{\cdot}\bm{B}\bm{\cdot}\bm{\omega}=\bm{\omega}\bm{\cdot}\bm{D}\bm{\cdot}\bm{\omega}-2\vartheta\Omega.
\end{equation}
The second term describes the diffusion of enstrophy.
The last term describes the viscous dissipation of enstrophy.
For incompressible flow,~\eqref{eq61} is consistent with the formulation of~\cite{Terrington2023LH}.

\subsection{Two alternative definitions of boundary enstrophy flux}
The integration of the viscous diffusion term over a fixed control volume $\mathscr{V}$ gives~\citep{Terrington2023LH}
\begin{eqnarray}\label{surface_integral}
\int_{\mathscr{V}}\bm{\nabla}\bm{\cdot}(\bm{J}_{\rm vis}\bm{\cdot}\bm{\omega})dV=\oint_{\partial \mathscr{V}}\hat{\bm{n}}\bm{\cdot}\bm{J}_{\rm vis}\bm{\cdot}\bm{\omega}dS,
\end{eqnarray}
where the integrand of the surface integral defines the boundary enstrophy flux (BEF) as
\begin{equation}\label{general_BEF}
F_{\Omega}\equiv\hat{\bm{n}}\bm{\cdot}\bm{J}_{\rm vis}\bm{\cdot}\bm{\omega}=\bm{\sigma}_{\bm\omega}\bm{\cdot}\bm{\omega},
\end{equation}
which describes the enstrophy creation rate on a boundary. Corresponding to the two BVFs in~\eqref{BVF1} and~\eqref{BVF2}, there exist two types of BEF as follows.
\begin{itemize}
	\item Using~\eqref{BVF1} and~\eqref{general_BEF}, we obtain
	\begin{eqnarray}\label{LPW_BEF}
		F_{\Omega}^{(1)}=\bm{\sigma}_{\bm\omega}^{(1)}\bm{\cdot}\bm{\omega}=\nu\frac{\partial\bm{\omega}}{\partial n}\bm{\cdot}\bm{\omega}=\nu\frac{\partial{\Omega}}{\partial n}~~\text{on}~~\partial\mathscr{V}.
	\end{eqnarray}
This definition bears analogy to classical Fourier's law of heat conduction, and has been used by~\citet{Wu1995},~\citet{LiuTS2016}, and~\citet{Chen2021features}. We refer to $F_{\Omega}^{(1)}$ as the Lighthill-Panton-Wu BEF.
\item 
Using~\eqref{BVF2} and~\eqref{general_BEF}, we obtain 
\begin{eqnarray}\label{LH_BEF}
	F_{\Omega}^{(2)}=\bm{\sigma}_{\bm\omega}^{(2)}\bm{\cdot}\bm{\omega}=-\nu\bm{\xi}\bm{\cdot}\left(\bm{\nabla}\times\bm{\omega}\right)~~\text{on}~~\partial\mathscr{V}.
\end{eqnarray}
We refer to $F_{\Omega}^{(2)}$ as the Lyman-Huggins BEF or the Lyman enstrophy flux.
\end{itemize}

\subsection{Full expressions of two boundary enstrophy fluxes}
The full expressions of the two BEFs are given as follows.
\begin{itemize}
	\item Using~\eqref{520abc}, the Lighthill-Panton-Wu BEF in~\eqref{LPW_BEF} is obtained as
	\begin{subequations}\label{6p7abcd}
		\begin{equation}\label{6p7a}
			F_{\Omega}^{(1)}=F_{\pi}^{(1)}+F_{n}^{(1)},
		\end{equation}
		\begin{equation}\label{6p7b}
			F_{\pi}^{(1)}\equiv\nu\left[\frac{\partial\Omega_{\pi}}{\partial n}\right]_{\Sigma}=\bm{\omega}_{\pi}\bm{\cdot}\bm{\sigma}_{\pi}^{(1)},~F_{n}^{(1)}\equiv\nu\left[\frac{\partial\Omega_{n}}{\partial n}\right]_{\Sigma}=\omega_{n}\sigma_{n}^{(1)},
		\end{equation}
		\begin{equation}\label{6p7c}
			F_{\pi}^{(1)}=\bm{\xi}\bm{\cdot}(\bm{a}-\bm{f})+\bm{\xi}\bm{\cdot}\bm{\nabla}_{\pi}\hat{P}+\nu\bm{\omega}_{\pi}\bm{\cdot}\bm{\nabla}_{\pi}\omega_{n}+\nu\bm{\omega}_{\pi}\bm{\cdot}\bm{K}\bm{\cdot}\bm{\omega}_{\pi},
		\end{equation}
		\begin{equation}\label{6p7d}
			F_{n}^{(1)}=-\nu\omega_{n}\left(\bm{\nabla}_{\pi}\bm{\cdot}\bm{\omega}\right)=-\nu\omega_{n}\left(\bm{\nabla}_{\pi}\bm{\cdot}\bm{\omega}_{\pi}\right)+2\nu{K}\Omega_{n}.
		\end{equation}
	\end{subequations}
	\item Using~\eqref{Wub}, the Lyman-Huggins BEF in~\eqref{LH_BEF} is obtained as
	\begin{equation}\label{6p8}
		F_{\Omega}^{(2)}=\bm{\xi}\bm{\cdot}(\bm{a}-\bm{f})+\bm{\xi}\bm{\cdot}\bm{\nabla}_{\pi}\hat{P}.
	\end{equation}
\end{itemize}

\subsection{Simplified expressions of boundary enstrophy fluxes on a stationary wall}\label{simple_BEF}
On a stationary curved wall (following~\S\ref{simple_BVF}),~\eqref{6p7abcd} and~\eqref{6p8} reduce to
\begin{itemize}
	\item The Lighthill-Panton-Wu BEF
	\begin{equation}
		F_{\Omega}^{(1)}=-\bm{\xi}\bm{\cdot}\bm{f}+\bm{\xi}\bm{\cdot}\bm{\nabla}_{\pi}\hat{P}+\nu\bm{\omega}\bm{\cdot}\bm{K}\bm{\cdot}\bm{\omega},
	\end{equation}
	\item The Lyman-Huggins BEF
	\begin{equation}
		F_{\Omega}^{(2)}=-\bm{\xi}\bm{\cdot}\bm{f}+\bm{\xi}\bm{\cdot}\bm{\nabla}_{\pi}\hat{P}.
	\end{equation}
\end{itemize}
Note that the skin friction is $\bm{\tau}=\mu\bm{\xi}=\mu\bm{\omega}\times\bm{n}$ with $\bm{\omega}=\bm{\omega}_{\pi}$. At a surface point where $\bm{\omega}\neq\bm{0}$, $(\bm{\tau},\bm{\omega})$ forms a special
orthonormal frame on $\bm{\Sigma}$, which has been called the $\bm{\tau}$-frame in a vorticity dynamics theory of 3D flow separation~\citep{Wu2000POF}.

\section{Intrinsic decomposition of boundary vorticity and enstrophy fluxes under the Lighthill-Panton-Wu definition}\label{BVF_LPW}
\subsection{Intrinsic decomposition of the Lighthill-Panton-Wu boundary vorticity flux}
\begin{theorem}\label{thm7}
On the base surface $\bm{\Sigma}$, the Lighthill-Panton-Wu BVF $\bm{\sigma}_{\bm\omega}^{(1)}$ in~\eqref{520abc} can be intrinsically decomposed as
\begin{subequations}
\begin{equation}\label{LH1}
	\bm{\sigma}_{\bm\omega}^{(1)}=\bm{\sigma}_{\pi}^{(1)}+\bm{\sigma}_{n}^{(1)}=\bm{\sigma}_{R}^{(1)}+\bm{\sigma}_{S}^{(1)}+\bm{\sigma}_{n}^{(1)},
\end{equation}
\begin{equation}\label{LH1a}
\bm{\sigma}_{R}^{(1)}\equiv\nu\left[\frac{\partial\bm{R}}{\partial n}\right]_{\Sigma},~\bm{\sigma}_{S}^{(1)}\equiv\nu\left[\frac{\partial\bm{S}}{\partial n}\right]_{\Sigma},~\bm{\sigma}_{n}^{(1)}\equiv\nu\left[\frac{\partial\bm{\omega}_{n}}{\partial n}\right]_{\Sigma},
\end{equation}
\end{subequations}
where the boundary rigid-rotation flux $\bm{\sigma}_{R}^{(1)}$, the boundary spin flux $\bm{\sigma}_{S}^{(1)}$, and the wall-normal BVF $\bm{\sigma}_{n}^{(1)}$ are respectively expressed as
\begin{subequations}\label{vvvv1}
\begin{equation}\label{LH2}
	\bm{\sigma}_{R}^{(1)}=\dashuline{-2\nu\bm{n}\times\bm{\nabla}_{\pi}\vartheta_{n}}~\underline{-2\nu{K}\bm{S}+2\nu\bm{K}\bm{\cdot}\bm{S}},
\end{equation}
\begin{eqnarray}\label{LH3}
	\bm{\sigma}_{S}^{(1)}&=&\dashuline{2\nu\bm{n}\times\bm{\nabla}_{\pi}\vartheta_{n}}\underline{\underline{+\nu\bm{\nabla}_{\pi}\omega_{n}}}\nonumber\\
	& &~\dotuline{+\bm{n}\times(\bm{a}-\bm{f})+\bm{n}\times\bm{\nabla}_{\pi}\hat{P}}\nonumber\\
	& &~\underline{+\nu\bm{K}\bm{\cdot}\bm{R}+2\nu{K}\bm{S}-\nu\bm{K}\bm{\cdot}\bm{S}},
\end{eqnarray}
	\begin{equation}\label{LH4}	\sigma_{n}^{(1)}=-\nu\bm{\nabla}_{\pi}\bm{\cdot}\bm{R}-\nu\bm{\nabla}_{\pi}\bm{\cdot}\bm{S}+\nu K{\omega}_{n}=-\nu\bm{\nabla}_{\pi}\bm{\cdot}\bm{\omega}.
\end{equation}
\end{subequations}
In terms of the conjugate curvature tensor pair $(\bm{K},\hat{\bm{K}})$ introduced in~\eqref{713_ab}, \eqref{vvvv1} admit the following compact representations:
\begin{subequations}
	\begin{equation}\label{LH2equivalent}
		\bm{\sigma}_{R}^{(1)}=\dashuline{-2\nu\bm{n}\times\bm{\nabla}_{\pi}\vartheta_{n}}~\underline{-2\nu\hat{\bm{K}}\bm{\cdot}\bm{S}},
	\end{equation}
	\begin{eqnarray}\label{LH3equivalent}
		\bm{\sigma}_{S}^{(1)}&=&\dashuline{2\nu\bm{n}\times\bm{\nabla}_{\pi}\vartheta_{n}}\underline{\underline{+\nu\bm{\nabla}_{\pi}\omega_{n}}}\nonumber\\
		& &~\dotuline{+\bm{n}\times(\bm{a}-\bm{f})+\bm{n}\times\bm{\nabla}_{\pi}\hat{P}}\nonumber\\
		& &~\underline{+\nu\bm{K}\bm{\cdot}\bm{R}+\nu\bm{K}\bm{\cdot}\bm{S}+2\nu\hat{\bm{K}}\bm{\cdot}\bm{S}}.
	\end{eqnarray}
\end{subequations}
\end{theorem}
Equations~\eqref{LH2equivalent} and~\eqref{LH3equivalent} reveal that $\bm{\sigma}_{R}^{(1)}$ and $\bm{\sigma}_{S}^{(1)}$ 
both receive contributions from the surface gradient of surface-normal dilatation ($\pm2\nu\bm{n}\times\bm{\nabla}_{\pi}\vartheta_{n}$), and the $\hat{\bm{K}}-\bm{S}$ interaction ($\pm 2\nu\hat{\bm{K}}\bm{\cdot}\bm{S}$). The boundary spin flux $\bm{\sigma}_{S}^{(1)}$ additionally comprises five distinct sources: the tangential gradient of surface-normal vorticity ($\nu\bm{\nabla}_{\pi}\omega_{n}$), the tangential acceleration or external force ($\bm{n}\times(\bm{a}-\bm{f})$), the surface pressure gradient ($\bm{n}\times\bm{\nabla}_{\pi}\hat{P}$), the ${\bm{K}}-\bm{R}$ and ${\bm{K}}-\bm{S}$ coupling terms ($\nu\bm{K}\bm{\cdot}\bm{R}$ and $\nu\bm{K}\bm{\cdot}\bm{S}$, respectively).
\begin{proof}[Proof of Theorem~{\upshape\ref{thm7}}] By virtue of \eqref{hh1}, the boundary rigid-rotation flux is evaluated as
\begin{eqnarray}\label{mm1}
\bm{\sigma}_{R}^{(1)}&\equiv&\nu\left[\frac{\partial\bm{R}}{\partial n}\right]_{\Sigma}=-2\nu\bm{n}\times\left(\left[\frac{\partial\bm{A}}{\partial n}\right]_{\Sigma}\bm{\cdot}\bm{n}\right)\nonumber\\
&=&-2\nu\bm{n}\times(\bm{K}\bm{\cdot}\bm{\nabla}_{\pi}\bm{u}\bm{\cdot}\bm{n})-2\nu\bm{n}\times\left(\bm{\nabla}_{\pi}\left[\frac{\partial\bm{u}}{\partial n}\right]_{\Sigma}\bm{\cdot}\bm{n}\right).
\end{eqnarray}
On the right-hand side of \eqref{mm1}, the first term, by virtue of \eqref{hhh1}, is evaluated as
\begin{eqnarray}\label{mm2}
-2\nu\bm{n}\times(\bm{K}\bm{\cdot}\bm{\nabla}_{\pi}\bm{u}\bm{\cdot}\bm{n})
	&=&-2\nu\bm{n}\times\left[\bm{K}\bm{\cdot}\left(\bm{\nabla}_{C}\bm{u}_{\pi}-\frac{1}{2}\bm{\xi}_{R}\bm{n}-u_{n}\bm{K}\right)\bm{\cdot}\bm{n}\right]\nonumber\\
	&=&\nu\bm{n}\times\left(\bm{K}\bm{\cdot}\bm{\xi}_{R}\right)\nonumber\\
	&=&\hat{\bm{K}}\bm{\cdot}\bm{R}.
\end{eqnarray}
By means of~\eqref{hh2}, the second term is calculated as
\begin{eqnarray}\label{mm3}
	&&-2\nu\bm{n}\times\left(\bm{\nabla}_{\pi}\left[\frac{\partial\bm{u}}{\partial n}\right]_{\Sigma}\bm{\cdot}\bm{n}\right)\nonumber\\
	&=&-2\nu\bm{n}\times\left[\bm{\nabla}_{\pi}\vartheta_{n}+\bm{\nabla}_{\pi}\left(\frac{1}{2}\bm{\xi}_{R}+\bm{\xi}_{S}\right)\bm{\cdot}\bm{n}\right]\nonumber\\
	&=&-2\nu\bm{n}\times\bm{\nabla}_{\pi}\vartheta_{n}-\nu\bm{n}\times\left(\bm{K}\bm{\cdot}\bm{\xi}_{R}\right)-2\nu\bm{n}\times\left(\bm{K}\bm{\cdot}\bm{\xi}_{S}\right)\nonumber\\
	&=&-2\nu\bm{n}\times\bm{\nabla}_{\pi}\vartheta_{n}-\nu\hat{\bm{K}}\bm{\cdot}\bm{R}-2\nu\hat{\bm{K}}\bm{\cdot}\bm{S}.
\end{eqnarray}
where the first identity in~\eqref{713_ab} enables the expansion of the last term as
\begin{eqnarray}\label{mm4}
-2\nu\hat{\bm{K}}\bm{\cdot}\bm{S}=-2\nu{K}\bm{S}+2\nu\bm{K}\bm{\cdot}\bm{S}.
\end{eqnarray}
Substitution of~\eqref{mm2} and~\eqref{mm3} into~\eqref{mm1} yields~\eqref{LH2} and~\eqref{LH2equivalent}.

By virtue of~\eqref{decomp_vorticity2}, the boundary spin flux is evaluated as
\begin{eqnarray}\label{mm5}
	\bm{\sigma}_{S}^{(1)}&\equiv&\nu\left[\frac{\partial\bm{S}}{\partial n}\right]_{\Sigma}=\nu\bm{n}\times\left(\left[\frac{\partial\bm{A}}{\partial n}\right]_{\Sigma}\bm{\cdot}\bm{n}\right)+\nu\bm{n}\times\left(\bm{n}\bm{\cdot}\left[\frac{\partial\bm{A}}{\partial n}\right]_{\Sigma}\right).
\end{eqnarray}
On the right-hand side of~\eqref{mm5}, the first term is derived from~\eqref{LH2} as 
\begin{eqnarray}\label{mm6}
\nu\bm{n}\times\left(\left[\frac{\partial\bm{A}}{\partial n}\right]_{\Sigma}\bm{\cdot}\bm{n}\right)
=\nu\bm{n}\times\bm{\nabla}_{\pi}\vartheta_{n}+\nu\hat{\bm{K}}\bm{\cdot}\bm{S}.
\end{eqnarray}
By using~\eqref{hh1},~\eqref{pp1} and~\eqref{BVF_identity2p}, the second term evaluates to
	\begin{eqnarray}\label{mm7}
		&&\nu\bm{n}\times\left(\bm{n}\bm{\cdot}\left[\frac{\partial\bm{A}}{\partial n}\right]_{\Sigma}\right)=\nu\bm{n}\times\left[\frac{\partial^2\bm{u}_{\pi}}{\partial n^2}\right]_{\Sigma}\nonumber\\
		&=&\nu\left[\frac{\partial\bm{\omega}_{\pi}}{\partial n}\right]_{\Sigma}+\nu\bm{n}\times\bm{\nabla}_{\pi}\vartheta_{n}+\nu\hat{\bm{K}}\bm{\cdot}\bm{S}\nonumber\\
		&=&\bm{n}\times\left(\bm{a}-\bm{f}\right)+\bm{n}\times\bm{\nabla}_{\pi}\hat{P}+\nu\bm{\nabla}_{\pi}\omega_{n}+\nu\bm{K}\bm{\cdot}\bm{\omega}_{\pi}+\nu\bm{n}\times\bm{\nabla}_{\pi}\vartheta_{n}+\nu\hat{\bm{K}}\bm{\cdot}\bm{S}\nonumber\\	&=&\bm{n}\times\left(\bm{a}-\bm{f}\right)+\bm{n}\times\bm{\nabla}_{\pi}\hat{P}+\nu\bm{\nabla}_{\pi}\omega_{n}+\nu\bm{K}\bm{\cdot}\bm{R}+\nu\bm{n}\times\bm{\nabla}_{\pi}\vartheta_{n}+\nu{K}\bm{S}.
	\end{eqnarray}
Substituting~\eqref{mm6} and~\eqref{mm7} into~\eqref{mm5} gives~\eqref{LH3} and~\eqref{LH3equivalent}. The opposing terms $-2\nu\hat{\bm{K}}\bm{\cdot}\bm{S}$ and $2\nu\hat{\bm{K}}\bm{\cdot}\bm{S}$ cancel each other, so that the total BVF is preserved. Equation~\eqref{LH4} follows from~\eqref{pp2} and~\eqref{BVF_identity3p}. This completes the proof. 
\end{proof}
\begin{remark}
The effect of Gauss curvature can be seen in the quantity $\bm{K}\bm{\cdot}\bm{R}$. Specifically, by~\eqref{id1_equivalent}, the rigid-rotation mode is expressed as
\begin{eqnarray}
	\bm{R}=2\hat{\bm{K}}\bm{\cdot}(\bm{u}_{\pi}\times\bm{n})-2\bm{n}\times\bm{\nabla}_{\pi}u_{n}.
\end{eqnarray}
Hence, by the second identity in~\eqref{713_ab}, we obtain
\begin{eqnarray}
	\bm{K}\bm{\cdot}\bm{R}=2K_{G}(\bm{u}_{\pi}\times\bm{n})+2\bm{K}\bm{\cdot}(\bm{\nabla}_{\pi}u_n\times\bm{n}).
\end{eqnarray}
\end{remark}

\subsection{Intrinsic decomposition of the Lighthill-Panton-Wu boundary enstrophy flux}
\begin{theorem}\label{thm8}
On the base surface $\bm{\Sigma}$, the Lighthill-Panton-Wu BEF ${F}_{\Omega}^{(1)}$ in~\eqref{6p7abcd} can be intrinsically decomposed as
\begin{subequations}\label{defv1}
	\begin{equation}\label{def1}
	F_{\Omega}^{(1)}\equiv\nu\left[\frac{\partial\Omega}{\partial n}\right]_{\Sigma}=F_{\pi}^{(1)}+F_{n}^{(1)},~F_{\pi}^{(1)}=F_{RR}^{(1)}+F_{RS}^{(1)}+F_{SS}^{(1)},
	\end{equation}
	\begin{equation}\label{def2}
	F_{RR}^{(1)}\equiv\nu\left[\frac{\partial\Omega_{R}}{\partial n}\right]_{\Sigma}=\bm{R}\bm{\cdot}\bm{\sigma}_{R}^{(1)},~	F_{SS}^{(1)}\equiv\nu\left[\frac{\partial\Omega_{S}}{\partial n}\right]_{\Sigma}=\bm{S}\bm{\cdot}\bm{\sigma}_{S}^{(1)},
	\end{equation}
	\begin{eqnarray}\label{def3}
	F_{RS}^{(1)}=\nu\left[\frac{\partial\Omega_{RS}}{\partial n}\right]_{\Sigma}=\bm{R}\bm{\cdot}\bm{\sigma}_{S}^{(1)}+\bm{S}\bm{\cdot}\bm{\sigma}_{R}^{(1)},~F_{n}^{(1)}\equiv\nu\left[\frac{\partial\Omega_n}{\partial n}\right]_{\Sigma}=\omega_{n}\sigma_{n}^{(1)},
	\end{eqnarray}
\end{subequations}
where the constituents of the BEF are explicitly expressed as
\begin{subequations}\label{eq7p13}
		\begin{equation}\label{FRR_1}
		F_{RR}^{(1)}=\dashuline{-2\nu\bm{\xi}_{R}\bm{\cdot}\bm{\nabla}_{\pi}\vartheta_{n}}~\underline{-2\nu{K}\bm{R}\bm{\cdot}\bm{S}+2\nu\bm{R}\bm{\cdot}\bm{K}\bm{\cdot}\bm{S}},
	\end{equation}
	\begin{eqnarray}\label{FSS_1}
		F_{SS}^{(1)}&=&\dashuline{2\nu\bm{\xi}_{S}\bm{\cdot}\bm{\nabla}_{\pi}\vartheta_{n}}~\underline{\underline{+\nu\bm{S}\bm{\cdot}\bm{\nabla}_{\pi}\omega_{n}}}\nonumber\\
		& &\dotuline{+\bm{\xi}_{S}\bm{\cdot}(\bm{a}-\bm{f})+\bm{\xi}_{S}\bm{\cdot}\bm{\nabla}_{\pi}\hat{P}}\nonumber\\
		& &\underline{+\nu\bm{R}\bm{\cdot}\bm{K}\bm{\cdot}\bm{S}+2\nu{K}{S}^2-\nu\bm{S}\bm{\cdot}\bm{K}\bm{\cdot}\bm{S}},
	\end{eqnarray}
	\begin{eqnarray}\label{FRS_1}
		F_{RS}^{(1)}&=&\dashuline{-2\nu\bm{\xi}_{S}\bm{\cdot}\bm{\nabla}_{\pi}\vartheta_{n}+2\nu\bm{\xi}_{R}\bm{\cdot}\bm{\nabla}_{\pi}\vartheta_{n}}~\underline{\underline{+\nu\bm{R}\bm{\cdot}\bm{\nabla}_{\pi}\omega_{n}}}\nonumber\\
		& &\dotuline{+\bm{\xi}_{R}\bm{\cdot}(\bm{a}-\bm{f})+\bm{\xi}_{R}\bm{\cdot}\bm{\nabla}_{\pi}\hat{P}}\nonumber\\
		& & \underline{-2\nu{K}{S}^2+2\nu\bm{S}\bm{\cdot}\bm{K}\bm{\cdot}\bm{S}+\nu\bm{R}\bm{\cdot}\bm{K}\bm{\cdot}\bm{R}}\nonumber\\
		& &\underline{+2\nu{K}\bm{R}\bm{\cdot}\bm{S}-\nu\bm{R}\bm{\cdot}\bm{K}\bm{\cdot}\bm{S},}
	\end{eqnarray}
	\begin{equation}
	F_{n}^{(1)}=-\nu\omega_{n}\left(\bm{\nabla}_{\pi}\bm{\cdot}\bm{\omega}\right)=-\nu\omega_{n}\left(\bm{\nabla}_{\pi}\bm{\cdot}\bm{\omega}_{\pi}\right)+2\nu{K}\Omega_{n}.
\end{equation}
\end{subequations}
	In terms of the conjugate curvature tensor pair $(\bm{K},\hat{\bm{K}})$ introduced in~\eqref{713_ab},~\eqref{eq7p13} admit the following compact formulations:
\begin{subequations}\label{vvvv7p15}
	\begin{equation}\label{FRR_1x}
		F_{RR}^{(1)}=\dashuline{-2\nu\bm{\xi}_{R}\bm{\cdot}\bm{\nabla}_{\pi}\vartheta_{n}}~\underline{-2\nu\bm{R}\bm{\cdot}\hat{\bm{K}}\bm{\cdot}\bm{S}},
	\end{equation}
	\begin{eqnarray}\label{FSS_1x}
		F_{SS}^{(1)}
		&=&\dashuline{2\nu\bm{\xi}_{S}\bm{\cdot}\bm{\nabla}_{\pi}\vartheta_{n}}~\underline{\underline{+\nu\bm{S}\bm{\cdot}\bm{\nabla}_{\pi}\omega_{n}}}\nonumber\\
		& &\dotuline{+\bm{\xi}_{S}\bm{\cdot}(\bm{a}-\bm{f})+\bm{\xi}_{S}\bm{\cdot}\bm{\nabla}_{\pi}\hat{P}}\nonumber\\
		& &\underline{+\nu\bm{R}\bm{\cdot}\bm{K}\bm{\cdot}\bm{S}+2\nu\bm{S}\bm{\cdot}\hat{\bm{K}}\bm{\cdot}\bm{S}+\nu\bm{S}\bm{\cdot}\bm{K}\bm{\cdot}\bm{S}},
	\end{eqnarray}
	\begin{eqnarray}\label{FRS_1x}
		F_{RS}^{(1)}
		&=&\dashuline{-2\nu\bm{\xi}_{S}\bm{\cdot}\bm{\nabla}_{\pi}\vartheta_{n}+2\nu\bm{\xi}_{R}\bm{\cdot}\bm{\nabla}_{\pi}\vartheta_{n}}~\underline{\underline{+\nu\bm{R}\bm{\cdot}\bm{\nabla}_{\pi}\omega_{n}}}\nonumber\\
		& &\dotuline{+\bm{\xi}_{R}\bm{\cdot}(\bm{a}-\bm{f})+\bm{\xi}_{R}\bm{\cdot}\bm{\nabla}_{\pi}\hat{P}}\nonumber\\
		& &\underline{-2\nu\bm{S}\bm{\cdot}\hat{\bm{K}}\bm{\cdot}\bm{S}+\nu\bm{R}\bm{\cdot}\bm{K}\bm{\cdot}\bm{R}+2\nu\bm{R}\bm{\cdot}\hat{\bm{K}}\bm{\cdot}\bm{S}+\nu\bm{R}\bm{\cdot}\bm{K}\bm{\cdot}\bm{S}}.
	\end{eqnarray}
\end{subequations}
\end{theorem}
Equation \eqref{vvvv7p15} reveals that all source terms characterizing the interactions between vorticity modes and surface geometry can be expressed in a unified bilinear (or quadratic) form. Furthermore, both $F_{SS}^{(1)}$ and $F_{RS}^{(1)}$ receive additional contributions from the external force, boundary acceleration, and surface pressure gradient.
\begin{proof}[Proof of Theorem~{\upshape\ref{thm8}}] Substituting~\eqref{vvvv1} into~\eqref{defv1} yields~\eqref{eq7p13}. This completes the proof.
\end{proof}
\begin{remark}
	In the expression for $F_{RS}^{(1)}$, the respective contributions of $\bm{\sigma}_{R}^{(1)}\bm{\cdot}\bm{S}$ and $\bm{\sigma}_{S}^{(1)}\bm{\cdot}\bm{R}$ are given by 
	\begin{subequations}
		\begin{equation}
			\bm{\sigma}_{R}^{(1)}\bm{\cdot}\bm{S}=-2\nu\bm{\xi}_{S}\bm{\cdot}\bm{\nabla}_{\pi}\vartheta_{n}-2\nu\bm{S}\bm{\cdot}\hat{\bm{K}}\bm{\cdot}\bm{S},
		\end{equation}
		\begin{eqnarray}
			\bm{\sigma}_{S}^{(1)}\bm{\cdot}\bm{R}&=&2\nu\bm{\xi}_{R}\bm{\cdot}\bm{\nabla}_{\pi}\vartheta_{n}+\nu\bm{R}\bm{\cdot}\bm{\nabla}_{\pi}\omega_{n}\nonumber\\
			 & &+\bm{\xi}_{R}\bm{\cdot}(\bm{a}-\bm{f})+\bm{\xi}_{R}\bm{\cdot}\bm{\nabla}_{\pi}\hat{P}\nonumber\\
			& &+\nu\bm{R}\bm{\cdot}\bm{K}\bm{\cdot}\bm{R}+\nu\bm{R}\bm{\cdot}\bm{K}\bm{\cdot}\bm{S}+2\nu\bm{R}\bm{\cdot}\hat{\bm{K}}\bm{\cdot}\bm{S}.
		\end{eqnarray}
	\end{subequations}
\end{remark}
\begin{remark} 
From~\eqref{ff1},~\eqref{ff2} and~\eqref{ff3}, the quadratic-type forms associated with $(\bm{R},\bm{S})$ and $(\bm{\xi}_{R},\bm{\xi}_{S})$ necessarily satisfy
\begin{subequations}
	\begin{equation}
		\bm{\xi}_{S}\bm{\cdot}\bm{K}\bm{\cdot}\bm{\xi}_{S}=\bm{S}\bm{\cdot}\hat{\bm{K}}\bm{\cdot}\bm{S},~~\bm{\xi}_{S}\bm{\cdot}\hat{\bm{K}}\bm{\cdot}\bm{\xi}_{S}=\bm{S}\bm{\cdot}{\bm{K}}\bm{\cdot}\bm{S},
	\end{equation}
		\begin{equation}
		\bm{\xi}_{R}\bm{\cdot}\bm{K}\bm{\cdot}\bm{\xi}_{S}=\bm{R}\bm{\cdot}\hat{\bm{K}}\bm{\cdot}\bm{S},~~\bm{\xi}_{R}\bm{\cdot}\hat{\bm{K}}\bm{\cdot}\bm{\xi}_{S}=\bm{R}\bm{\cdot}{\bm{K}}\bm{\cdot}\bm{S},
	\end{equation}
	\begin{equation}
		\bm{\xi}_{S}\bm{\cdot}\bm{K}\bm{\cdot}\bm{\xi}_{S}+\bm{S}\bm{\cdot}\bm{K}\bm{\cdot}\bm{S}=K{S}^2=K{\xi}_{S}^2=2K\Omega_{S}.
	\end{equation}
\end{subequations}
\end{remark}

\subsection{Simplified expressions of boundary vorticity and enstrophy fluxes on a stationary wall}\label{yy0000}
On a stationary curved wall (following~\S\ref{simple_BVF}), we have $\bm{R}=\bm{0}$ and $\bm{\omega}=\bm{\omega}_{\pi}=\bm{S}$. This implies $\bm{\xi}_{R}=\bm{0}$ and $\bm{\xi}=\bm{\xi}_{S}$, together with $F_{n}^{(1)}=F_{RR}^{(1)}=0$. Consequently, we obtain
\begin{itemize}
	\item Intrinsic decomposition of Lighthill-Panton-Wu BVF
	\begin{subequations}
	    \begin{equation}
	    	\begin{aligned}
	    		\bm{\sigma}_{\bm\omega}^{(1)}
	    		&=\bm{\sigma}_{R}^{(1)}+\bm{\sigma}_{S}^{(1)}+\bm{\sigma}_{n}^{(1)}\\
	    		&=\dotuline{-\bm{n}\times\bm{f}+\bm{n}\times\bm{\nabla}_{\pi}\hat{P}}~\underline{+\nu\bm{K}\bm{\cdot}\bm{\omega}}-\nu(\bm{\nabla}_{\pi}\bm{\cdot}\bm{\omega})\bm{n},
	    	\end{aligned}
	    \end{equation}
		\begin{equation}\label{LH2equivalentc}
			\bm{\sigma}_{R}^{(1)}=\dashuline{-2\nu\bm{n}\times\bm{\nabla}_{\pi}\vartheta}~\underline{-2\nu\hat{\bm{K}}\bm{\cdot}\bm{\omega}},
		\end{equation}
		\begin{eqnarray}\label{LH3equivalencc}
			\bm{\sigma}_{S}^{(1)}&=&\dashuline{2\nu\bm{n}\times\bm{\nabla}_{\pi}\vartheta}\dotuline{-\bm{n}\times\bm{f}+\bm{n}\times\bm{\nabla}_{\pi}\hat{P}}\nonumber\\
			& &\underline{+\nu\bm{K}\bm{\cdot}\bm{\omega}+2\nu\hat{\bm{K}}\bm{\cdot}\bm{\omega}}.
		\end{eqnarray}
		\begin{equation}
			\bm{\sigma}_{n}^{(1)}=-\nu(\bm{\nabla}_{\pi}\bm{\cdot}\bm{\omega})\bm{n}.
		\end{equation}
	\end{subequations}
	It is noted that $\hat{\bm{K}}\bm{\cdot}\bm{\omega}=K\bm{\omega}-\bm{K}\bm{\cdot}\bm{\omega}=\bm{n}\times\left(\bm{K}\bm{\cdot}\bm{\xi}\right)$.
	\item Intrinsic decomposition of Lighthill-Panton-Wu BEF
	\begin{subequations}
		\begin{eqnarray}
			F_{\Omega}^{(1)}&=&F_{\pi}^{(1)}=F_{SS}^{(1)}+F_{RS}^{(1)}\nonumber\\
			&=&\dotuline{-\bm{\xi}\bm{\cdot}\bm{f}+\bm{\xi}\bm{\cdot}\bm{\nabla}_{\pi}\hat{P}}~\underline{+\nu\bm{\omega}\bm{\cdot}\bm{K}\bm{\cdot}\bm{\omega}},
		\end{eqnarray}
		\begin{eqnarray}
			F_{SS}^{(1)}&=&\dashuline{2\nu\bm{\xi}\bm{\cdot}\bm{\nabla}_{\pi}\vartheta}~\dotuline{-\bm{\xi}\bm{\cdot}\bm{f}+\bm{\xi}\bm{\cdot}\bm{\nabla}_{\pi}\hat{P}}\nonumber\\
			& &\underline{+\nu\bm{\omega}\bm{\cdot}\bm{K}\bm{\cdot}\bm{\omega}+2\nu\bm{\xi}\bm{\cdot}\bm{K}\bm{\cdot}\bm{\xi}},
		\end{eqnarray}
		\begin{equation}
			F_{RS}^{(1)}=\dashuline{-2\nu\bm{\xi}\bm{\cdot}\bm{\nabla}_{\pi}\vartheta}~\underline{-2\nu\bm{\xi}\bm{\cdot}\bm{K}\bm{\cdot}\bm{\xi}}.
		\end{equation}
	\end{subequations}
	It is noted that $\bm{\xi}\bm{\cdot}\bm{K}\bm{\cdot}\bm{\xi}=\bm{\omega}\bm{\cdot}\hat{\bm{K}}\bm{\cdot}\bm{\omega}=K\omega^2-\bm{\omega}\bm{\cdot}{\bm{K}}\bm{\cdot}\bm{\omega}$.
\end{itemize}

\section{Intrinsic decomposition of boundary vorticity and enstrophy fluxes under the Lyman-Huggins definition}\label{BVF_LH}
\subsection{Intrinsic decomposition of the Lyman-Huggins boundary vorticity flux}
\begin{theorem}\label{theorem9}
	On the base surface $\bm{\Sigma}$, the Lyman-Huggins BVF $\bm{\sigma}_{\bm{\omega}}^{(2)}$ in~\eqref{Wub} can be intrinsically decomposed as
	\begin{subequations}
			\begin{equation}\label{uu1}
			\bm{\sigma}_{\bm{\omega}}^{(2)}=\bm{\sigma}_{R}^{(2)}+\bm{\sigma}_{S}^{(2)}+\bm{\mathfrak{S}},
		\end{equation}
\begin{equation}\label{uu2}
	\begin{gathered}
		\bm{\sigma}_{R}^{(2)}=-\nu\bm{n}\times\left(\bm{\nabla}\times\bm{R}\right), \quad
		\bm{\sigma}_{S}^{(2)}=-\nu\bm{n}\times\left(\bm{\nabla}\times\bm{S}\right), \\
		\bm{\mathfrak{S}}=-\nu\bm{n}\times\left(\bm{\nabla}\times\bm{\omega}_{n}\right),
	\end{gathered}
\end{equation}
\end{subequations}
where the constituents of the BVF are explicitly expressed as
\begin{subequations}\label{xxh}
	\begin{equation}\label{LH2prime}
	\bm{\sigma}_{R}^{(2)}=\dashuline{-2\nu\bm{n}\times\bm{\nabla}_{\pi}\vartheta_{n}}~\underline{-2\nu\hat{\bm{K}}\bm{\cdot}\bm{S}-\nu\bm{K}\bm{\cdot}\bm{R}},
\end{equation}
\begin{eqnarray}\label{LH3prime}
	\bm{\sigma}_{S}^{(2)}&=&\dashuline{2\nu\bm{n}\times\bm{\nabla}_{\pi}\vartheta_{n}}\underline{\underline{+\nu\bm{\nabla}_{\pi}\omega_{n}}}\nonumber\\
	& &\dotuline{+\bm{n}\times(\bm{a}-\bm{f})+\bm{n}\times\bm{\nabla}_{\pi}\hat{P}}\nonumber\\
	& &~\underline{+2\nu\hat{\bm{K}}\bm{\cdot}\bm{S}+\nu\bm{K}\bm{\cdot}\bm{R}},
\end{eqnarray}
\begin{equation}\label{LH4prime}	\bm{\mathfrak{S}}=-\nu\bm{n}\times(\bm{\nabla}_{\pi}\times\bm{\omega}_{n})=~\underline{\underline{-\nu\bm{\nabla}_{\pi}\omega_{n}}}.
\end{equation}	
\end{subequations}
\end{theorem}
\begin{proof}[Proof of Theorem~{\upshape\ref{theorem9}}] By combining~\eqref{pp21} and~\eqref{pp22} with the definitions in~\eqref{LH1a} and~\eqref{uu2}, we derive the relationship between $(\bm{\sigma}_{R}^{(1)},\bm{\sigma}_{S}^{(1)})$ and $(\bm{\sigma}_{R}^{(2)},\bm{\sigma}_{S}^{(2)})$:
\begin{subequations}\label{relation1}
	\begin{equation}
		\bm{\sigma}_{R}^{(2)}=\bm{\sigma}_{R}^{(1)}-\nu\bm{K}\bm{\cdot}\bm{R},
	\end{equation}
	\begin{equation}
		\bm{\sigma}_{S}^{(2)}=\bm{\sigma}_{S}^{(1)}-\nu\bm{K}\bm{\cdot}\bm{S}.
	\end{equation}
\end{subequations}
Substituting~\eqref{LH2} and~\eqref{LH3} into~\eqref{relation1} yields~\eqref{LH2prime} and~\eqref{LH3prime}, while \eqref{LH4prime} is an immediate consequence of \eqref{431b}. The proof is thus complete.
\end{proof}
Unlike $\bm{\sigma}_{R}^{(1)}$ and $\bm{\sigma}_{S}^{(1)}$, both $\bm{\sigma}_{R}^{(2)}$ and $\bm{\sigma}_{S}^{(2)}$ share contributions due to $\hat{\bm{K}}-\bm{S}$ and $\bm{K}-\bm{R}$ couplings, the sole difference being the sign reversal. The contributions arising from the tangential gradient of surface-normal vorticity cancel out identically between $\bm{\sigma}_{S}^{(2)}$ and $\mathfrak{S}$.

\subsection{Intrinsic decomposition of the Lyman-Huggins boundary enstrophy flux}
\begin{theorem}\label{theorem13}
On the base surface $\bm{\Sigma}$, the Lyman-Huggins BEF $F_{\Omega}^{(2)}$ in~\eqref{6p8} can be intrinsically decomposed as
\begin{subequations}\label{eq84abcd}
	\begin{equation}\label{eq84a}
		F_{\Omega}^{(2)}=F_{RR}^{(2)}+F_{RS}^{(2)}+F_{SS}^{(2)}+F_{R\mathfrak{S}}^{(2)}+F_{S\mathfrak{S}}^{(2)},
	\end{equation}
	\begin{equation}\label{eq84b}
		F_{RR}^{(2)}=\bm{R}\bm{\cdot}\bm{\sigma}_{R}^{(2)}=-\nu\bm{\xi}_{R}\bm{\cdot}(\bm{\nabla}\times\bm{R}),~F_{SS}^{(2)}=\bm{S}\bm{\cdot}\bm{\sigma}_{S}^{(2)}=-\nu\bm{\xi}_{S}\bm{\cdot}(\bm{\nabla}\times\bm{S}),
	\end{equation}
	\begin{equation}\label{eq84c}
		F_{RS}^{(2)}=\bm{R}\bm{\cdot}\bm{\sigma}_{S}^{(2)}+\bm{S}\bm{\cdot}\bm{\sigma}_{R}^{(2)}=-\nu\bm{\xi}_{R}\bm{\cdot}(\bm{\nabla}\times\bm{S})-\nu\bm{\xi}_{S}\bm{\cdot}(\bm{\nabla}\times\bm{R}),
	\end{equation}
	\begin{equation}\label{eq84d}
		F_{R\mathfrak{S}}^{(2)}=\bm{R}\bm{\cdot}\bm{\mathfrak{S}},~F_{S\mathfrak{S}}^{(2)}=\bm{S}\bm{\cdot}\bm{\mathfrak{S}},
	\end{equation}
\end{subequations}
where the constituents of the BEF are explicitly expressed as
\begin{subequations}\label{eq85abcd}
\begin{eqnarray}\label{eq85a}
	F_{RR}^{(2)}=\dashuline{-2\nu\bm{\xi}_{R}\bm{\cdot}\bm{\nabla}_{\pi}\vartheta_{n}}~\underline{-2\nu\bm{R}\bm{\cdot}\hat{\bm{K}}\bm{\cdot}\bm{S}-\nu\bm{R}\bm{\cdot}\bm{K}\bm{\cdot}\bm{R}},
\end{eqnarray}
\begin{eqnarray}\label{eq85b}
F_{SS}^{(2)}&=&\dashuline{2\nu\bm{\xi}_{S}\bm{\cdot}\bm{\nabla}_{\pi}\vartheta_{n}}~\underline{\underline{+\nu\bm{S}\bm{\cdot}\bm{\nabla}_{\pi}\omega_{n}}}\nonumber\\
& &\dotuline{+\bm{\xi}_{S}\bm{\cdot}(\bm{a}-\bm{f})+\bm{\xi}_{S}\bm{\cdot}\bm{\nabla}_{\pi}\hat{P}}\nonumber\\
& &\underline{+\nu\bm{R}\bm{\cdot}\bm{K}\bm{\cdot}\bm{S}+2\nu\bm{S}\bm{\cdot}\hat{\bm{K}}\bm{\cdot}\bm{S}},
\end{eqnarray}
\begin{eqnarray}\label{eq85c}
F_{RS}^{(2)}&=&\dashuline{-2\nu\bm{\xi}_{S}\bm{\cdot}\bm{\nabla}_{\pi}\vartheta_{n}+2\nu\bm{\xi}_{R}\bm{\cdot}\bm{\nabla}\vartheta_{n}}~\underline{\underline{+\nu\bm{R}\bm{\cdot}\bm{\nabla}_{\pi}\omega_{n}}}\nonumber\\
& &\dotuline{+\bm{\xi}_{R}\bm{\cdot}(\bm{a}-\bm{f})+\bm{\xi}_{R}\bm{\cdot}\bm{\nabla}_{\pi}\hat{P}}\nonumber\\
& &\underline{-2\nu\bm{S}\bm{\cdot}\hat{\bm{K}}\bm{\cdot}\bm{S}+\nu\bm{R}\bm{\cdot}\bm{K}\bm{\cdot}\bm{R}+2\nu\bm{R}\bm{\cdot}\hat{\bm{K}}\bm{\cdot}\bm{S}-\nu\bm{R}\bm{\cdot}\bm{K}\bm{\cdot}\bm{S}},
\end{eqnarray}
\begin{equation}\label{eq85d}
F_{R\mathfrak{S}}^{(2)}=\underline{\underline{-\nu\bm{R}\bm{\cdot}\bm{\nabla}_{\pi}\omega_{n}}},
\end{equation}
\begin{equation}\label{eq85e}
F_{S\mathfrak{S}}^{(2)}=\underline{\underline{-\nu\bm{S}\bm{\cdot}\bm{\nabla}_{\pi}\omega_{n}}}.
\end{equation}
\end{subequations}
\end{theorem}
\begin{proof}[Proof of Theorem~{\upshape\ref{theorem13}}]Substituting~\eqref{xxh} into~\eqref{eq84abcd} yields~\eqref{eq85abcd}. This completes the proof.
\end{proof}
\begin{remark}
In~\eqref{eq85a},~\eqref{eq85b} and~\eqref{eq85c}, the terms representing vorticity-curvature interaction can be expanded as
\begin{subequations}
\begin{equation}
	-2\nu\bm{R}\bm{\cdot}\hat{\bm{K}}\bm{\cdot}\bm{S}=-2\nu{K}\bm{R}\bm{\cdot}\bm{S}+2\nu\bm{R}\bm{\cdot}\bm{K}\bm{\cdot}\bm{S},
\end{equation}
\begin{equation}
	2\nu\bm{S}\bm{\cdot}\hat{\bm{K}}\bm{\cdot}\bm{S}=2\nu\bm{\xi}_{S}\bm{\cdot}\bm{K}\bm{\cdot}\bm{\xi}_{S}=2\nu KS^2-2\nu\bm{S}\bm{\cdot}\bm{K}\bm{\cdot}\bm{S},
\end{equation}
\begin{equation}
	2\nu\bm{R}\bm{\cdot}\hat{\bm{K}}\bm{\cdot}\bm{S}-\nu\bm{R}\bm{\cdot}\bm{K}\bm{\cdot}\bm{S}=2\nu{K}\bm{R}\bm{\cdot}\bm{S}-3\nu\bm{R}\bm{\cdot}\bm{K}\bm{\cdot}\bm{S}.
\end{equation}	
\end{subequations}
\end{remark}
\begin{remark}
Comparison between~\eqref{vvvv7p15} and~\eqref{eq85abcd} yields
\begin{subequations}
\begin{equation}
	F_{RR}^{(2)}=F_{RR}^{(1)}-\nu\bm{R}\bm{\cdot}\bm{K}\bm{\cdot}\bm{R},
\end{equation}
\begin{equation}
	F_{SS}^{(2)}=F_{SS}^{(1)}-\nu\bm{S}\bm{\cdot}\bm{K}\bm{\cdot}\bm{S},
\end{equation}
\begin{equation}	
	F_{RS}^{(2)}=F_{RS}^{(1)}-2\nu\bm{R}\bm{\cdot}\bm{K}\bm{\cdot}\bm{S}.
\end{equation}
Therefore, we obtain
\begin{equation}
	F_{RR}^{(2)}+F_{SS}^{(2)}+F_{RS}^{(2)}=F_{RR}^{(1)}+F_{SS}^{(1)}+F_{RS}^{(1)}-\nu\bm{\omega}_{\pi}\bm{\cdot}\bm{K}\bm{\cdot}\bm{\omega}_{\pi}.
\end{equation}	
\end{subequations}
\end{remark}

\subsection{Simplified expressions of boundary vorticity and enstrophy fluxes on a stationary wall}
On a stationary curved wall $\bm{\Sigma}$ (following~\S\ref{simple_BVF} and~\S\ref{yy0000}), we obtain
\begin{itemize}
	\item Intrinsic decomposition of Lyman-Huggins BVF
	\begin{subequations}
	\begin{equation}
		\bm{\sigma}_{\bm{\omega}}^{(2)}=\bm{\sigma}_{R}^{(2)}+\bm{\sigma}_{S}^{(2)}=\dotuline{-\bm{n}\times\bm{f}+\bm{n}\times\bm{\nabla}_{\pi}\hat{P}},
	\end{equation}
	\begin{equation}
		\bm{\sigma}_{R}^{(2)}=\dashuline{-2\nu\bm{n}\times\bm{\nabla}_{\pi}\vartheta}~\underline{-2\nu\hat{\bm{K}}\bm{\cdot}\bm{\omega}},
	\end{equation}
	\begin{equation}
			\bm{\sigma}_{S}^{(2)}=\dashuline{2\nu\bm{n}\times\bm{\nabla}_{\pi}\vartheta}~\dotuline{-\bm{n}\times\bm{f}+\bm{n}\times\bm{\nabla}_{\pi}\hat{P}}~\underline{+2\nu\hat{\bm{K}}\bm{\cdot}\bm{\omega}}.
	\end{equation}
	\end{subequations}
	It is noted that $\hat{\bm{K}}\bm{\cdot}\bm{\omega}=K\bm{\omega}-\bm{K}\bm{\cdot}\bm{\omega}=\bm{n}\times\left(\bm{K}\bm{\cdot}\bm{\xi}\right)$.
	\item Intrinsic decomposition of Lyman-Huggins BEF
	\begin{subequations}
	\begin{equation}
	F_{\Omega}^{(2)}=F_{SS}^{(2)}+F_{RS}^{(2)}=\dotuline{-\bm{\xi}\bm{\cdot}\bm{f}+\bm{\xi}\bm{\cdot}\bm{\nabla}_{\pi}\hat{P}},
	\end{equation}
	\begin{equation}
		F_{SS}^{(2)}=\dashuline{2\nu\bm{\xi}\bm{\cdot}\bm{\nabla}_{\pi}\vartheta}~\dotuline{-\bm{\xi}\bm{\cdot}\bm{f}+\bm{\xi}\bm{\cdot}\bm{\nabla}_{\pi}\hat{P}}~\underline{+2\nu\bm{\xi}\bm{\cdot}\bm{K}\bm{\cdot}\bm{\xi}},
	\end{equation}
\begin{equation}
	F_{RS}^{(2)}=\dashuline{-2\nu\bm{\xi}\bm{\cdot}\bm{\nabla}_{\pi}\vartheta}~\underline{-2\nu\bm{\xi}\bm{\cdot}\bm{K}\bm{\cdot}\bm{\xi}}.
\end{equation}
	\end{subequations}
	It is noted that $\bm{\xi}\bm{\cdot}\bm{K}\bm{\cdot}\bm{\xi}=\bm{\omega}\bm{\cdot}\hat{\bm{K}}\bm{\cdot}\bm{\omega}=K\omega^2-\bm{\omega}\bm{\cdot}{\bm{K}}\bm{\cdot}\bm{\omega}$.
\end{itemize}

\section{Several supplementary theorems on kinematics and dynamics}\label{Several supplementary theorems}
\begin{theorem}\label{thm4}
	On the base surface $\bm{\Sigma}$, the wall-normal derivative of the dilatation $\vartheta$ is expressed as
	\begin{subequations}
			\begin{eqnarray}\label{eq9p1}
			\left[\frac{\partial\vartheta}{\partial n}\right]_{\Sigma}=		\bm{K}\bm{:}\bm{\nabla}_{\pi}\bm{u}+\bm{\nabla}_{\pi}\bm{\cdot}\left(\frac{1}{2}\bm{\xi}_{R}+\bm{\xi}_{S}\right)-K\vartheta_{n}+\left[\frac{\partial^2 u_n}{\partial n^2}\right]_{\Sigma},
		\end{eqnarray}
		where the first two terms on the right-hand side are further expressed as
		\begin{eqnarray}
			\bm{K}\bm{:}\bm{\nabla}_{\pi}\bm{u}=\bm{K}\bm{:}\bm{\nabla}_{C}\bm{u}_{\pi}-u_n(K^2-2K_G),
		\end{eqnarray}
		\begin{eqnarray}
			\bm{\nabla}_{\pi}\bm{\cdot}\left(\frac{1}{2}\bm{\xi}_{R}+\bm{\xi}_{S}\right)=(\bm{n}\bm{\times}\bm{\nabla}_{\pi})\bm{\cdot}\left(\frac{1}{2}\bm{R}+\bm{S}\right).
		\end{eqnarray}
	\end{subequations}
\end{theorem}
\begin{proof}[Proof of Theorem~{\upshape\ref{thm4}}]
	Taking the trace of~\eqref{hh1} and~\eqref{hh2} gives
	\begin{eqnarray}\label{hh1x}
		\left[\frac{\partial\vartheta}{\partial n}\right]_{\Sigma}=\bm{K}\bm{:}\bm{\nabla}_{\pi}\bm{u}+\bm{\nabla}_{\pi}\bm{\cdot}\left[\frac{\partial\bm{u}}{\partial n}\right]_{\Sigma}+\left[\frac{\partial^2 u_n}{\partial n^2}\right]_{\Sigma},
	\end{eqnarray}
	\begin{eqnarray}\label{hh2x}
		\bm{\nabla}_{\pi}\bm{\cdot}\left[\frac{\partial\bm{u}}{\partial n}\right]_{\Sigma}=	-\vartheta_{n}K+\bm{\nabla}_{\pi}\bm{\cdot}\left(\frac{1}{2}\bm{\xi}_{R}+\bm{\xi}_{S}\right).
	\end{eqnarray}
	Using~\eqref{third} and~\eqref{qq00}, the first term on the right-hand side of~\eqref{hh1x} is evaluated as
	\begin{equation}
		\begin{aligned}
			\bm{K}\bm{:}\bm{\nabla}_{\pi}\bm{u}&=\bm{K}\bm{:}\bm{\nabla}_{C}\bm{u}_{\pi}-u_n{\rm tr}(\bm{K}^2)\\
			&=\bm{K}\bm{:}\bm{\nabla}_{C}\bm{u}_{\pi}-u_{n}{\rm tr}(K\bm{K}-K_{G}\bm{G})\\
			&=\bm{K}\bm{:}\bm{\nabla}_{C}\bm{u}_{\pi}-u_{n}(K^2-2K_G),
		\end{aligned}
	\end{equation}
	where $K^2-2K_{G}=\lambda_{1}^{2}+\lambda_{2}^{2}$.
	Using~\eqref{curl_n}, the second term on the right-hand side of~\eqref{hh2x} is
	\begin{eqnarray}
		\bm{\nabla}_{\pi}\bm{\cdot}\left(\frac{1}{2}\bm{\xi}_{R}+\bm{\xi}_{S}\right)&=&\bm{\nabla}_{\pi}\bm{\cdot}\left[\left(\frac{1}{2}\bm{R}+\bm{S}\right)\times\bm{n}\right]\nonumber\\
		&=&\bm{n}\bm{\cdot}\bm{\nabla}_{\pi}\times\left(\frac{1}{2}\bm{R}+\bm{S}\right)-\left(\frac{1}{2}\bm{R}+\bm{S}\right)\bm{\cdot}(\bm{\nabla}_{\pi}\times\bm{n})\nonumber\\
		&=&(\bm{n}\times\bm{\nabla}_{\pi})\bm{\cdot}\left(\frac{1}{2}\bm{R}+\bm{S}\right).
	\end{eqnarray}
	This completes the proof.
\end{proof}

\begin{theorem}\label{thm5}
	Expressed via the surface Laplacian $\nabla_{\pi}^{2}\bm{u}$ on the base surface $\bm{\Sigma}$, the wall-normal derivative of the tangential vorticity $\bm{\omega}_{\pi}$ can be expressed as
	\begin{subequations}
			\begin{equation}\label{appx1}
			\begin{aligned}
				\left[\frac{\partial\bm{\omega}_{\pi}}{\partial n}\right]_{\Sigma}=&-\bm{n}\times\bm{\nabla}_{\pi}\vartheta+\bm{n}\times(\nabla_{\pi}^{2}\bm{u})_{\pi}-K\left(\frac{1}{2}\bm{R}+\bm{S}\right)\\
				&+\bm{\nabla}_{\pi}\omega_{n}+\bm{K}\bm{\cdot}\bm{\omega}_{\pi}+\bm{n}\times\left[\frac{\partial^2\bm{u}_{\pi}}{\partial n^2}\right]_{\Sigma},
			\end{aligned}
		\end{equation}
		and it holds that
		\begin{equation}\label{appx2}
			\begin{aligned}
				\bm{n}\times(\nabla_{\pi}^{2}\bm{u})_{\pi}&=\frac{1}{2}K\bm{R}-\bm{K}\bm{\cdot}\bm{R}-\bm{\nabla}_{\pi}\omega_{n}+\bm{n}\times\bm{\nabla}_{\pi}\vartheta_{\pi}\\
				&=\frac{1}{2}\hat{\bm{K}}\bm{\cdot}\bm{R}-\frac{1}{2}\bm{K}\bm{\cdot}\bm{R}-\bm{\nabla}_{\pi}\omega_{n}+\bm{n}\times\bm{\nabla}_{\pi}\vartheta_{\pi}.
			\end{aligned}
		\end{equation}
	\end{subequations}
\end{theorem}
\begin{proof}[Proof of Theorem~{\upshape\ref{thm5}}]
	The curl of the vorticity admits the following decomposition:
	\begin{eqnarray}\label{eq9p10}
		\bm{\nabla}\times\bm{\omega}=\bm{\nabla}\times\left(\bm{\nabla}\times\bm{u}\right)=\bm{\nabla}{\vartheta}-\nabla^{2}\bm{u}.
	\end{eqnarray}
	Evaluated on the base surface $\bm{\Sigma}$, it is implied from~\eqref{BVF_identity2} and~\eqref{BVF_identity4} that
	\begin{eqnarray}
		\bm{n}\times\left(\bm{\nabla}\times\bm{\omega}\right)=\bm{\nabla}_{\pi}\omega_{n}+\bm{K}\bm{\cdot}\bm{\omega}_{\pi}-\left[\frac{\partial\bm{\omega}_{\pi}}{\partial n}\right]_{\Sigma}.
	\end{eqnarray}
	By using~\eqref{sdo4} and~\eqref{qq8}, the Laplacian of velocity on $\bm{\Sigma}$ is calculated as
	\begin{eqnarray}\label{eq9p12}
		\nabla^{2}\bm{u}&=&\nabla_{\pi}^2\bm{u}-K\left[\frac{\partial\bm{u}}{\partial n}\right]_{\Sigma}+\left[\frac{\partial^2\bm{u}}{\partial n^2}\right]_{\Sigma}\nonumber\\
		&=&\nabla_{\pi}^2\bm{u}-K\left[\vartheta_{n}\bm{n}+\left(\frac{1}{2}\bm{\xi}_{R}+\bm{\xi}_{S}\right)\right]+\left[\frac{\partial^2\bm{u}}{\partial n^2}\right]_{\Sigma}.
	\end{eqnarray}
Consequently, using \eqref{eq9p10} and \eqref{eq9p12}, we obtain
	\begin{eqnarray}
		\bm{n}\times(\bm{\nabla}\vartheta-\nabla^{2}\bm{u})=\bm{n}\times\bm{\nabla}_{\pi}\vartheta-\bm{n}\times(\nabla_{\pi}^{2}\bm{u})_{\pi}+K\left(\frac{1}{2}\bm{R}+\bm{S}\right)-\bm{n}\times\left[\frac{\partial^2\bm{u}_{\pi}}{\partial n^2}\right]_{\Sigma}.
	\end{eqnarray}
Equation~\eqref{appx2} is verified by direct comparison of~\eqref{pp1} and~\eqref{appx1}. This completes the proof.
\end{proof}

\begin{theorem}\label{app_theorem}
	On the based surface $\bm{\Sigma}$, the boundary pressure flux can be expressed as
	\begin{eqnarray}\label{bpf}	\left[\frac{\partial\hat{P}}{\partial n}\right]_{\Sigma}&=&-\left(a_n-f_n\right)-\nu\bm{\nabla}_{\pi}\bm{\cdot}\bm{\xi}\nonumber\\
		&=&-\left(a_n-f_n\right)-\nu(\bm{n}\times\bm{\nabla}_{\pi})\bm{\cdot}\bm{\omega}\nonumber\\
		&=&-\left(a_n-f_n\right)-\nu(\bm{n}\times\bm{\nabla})\bm{\cdot}\bm{\omega}\nonumber\\
		&=&-\left(a_n-f_n\right)-\nu\bm{n}\bm{\cdot}(\bm{\nabla}\times\bm{\omega}).
	\end{eqnarray}
	Here, $a_{n}\equiv\bm{a}\bm{\cdot}\bm{n}$ and $f_{n}\equiv\bm{f}\bm{\cdot}\bm{n}$ are the surface-normal components of acceleration and external force at $\bm{\Sigma}$.
\end{theorem}
\begin{proof}[Proof of Theorem~{\upshape\ref{app_theorem}}]
From~\eqref{aa2} and~\eqref{xi_decomp2}, it holds that
\begin{equation}\label{eq9p14}
-\frac{1}{2}\bm{\xi}_{R}=-\bm{W}_{\rm eff}\times\bm{n}=\bm{\nabla}_{\pi}\bm{u}\bm{\cdot}\bm{n},
\end{equation}
which implies that
\begin{eqnarray}\label{eq9p15}
-\frac{1}{2}\bm{\nabla}_{\pi}\bm{\cdot}\bm{\xi}_{R}&=&\bm{\nabla}_{\pi}\bm{\cdot}\left(\bm{\nabla}_{\pi}\bm{u}\bm{\cdot}\bm{n}\right)\nonumber\\
&=&(\nabla_{\pi}^2\bm{u})_{n}-\bm{K}\bm{:}\bm{\nabla}_{\pi}\bm{u}.
\end{eqnarray}
On the base surface $\bm{\Sigma}$, by using~\eqref{eq9p10}, evaluating the wall-normal component of~\eqref{NS2} gives
\begin{eqnarray}\label{eq9p16}
\left[\frac{\partial\hat{P}}{\partial n}\right]_{\Sigma}=-({a}_{n}-f_{n})-\nu\left[\frac{\partial\vartheta}{\partial n}\right]_{\Sigma}+\nu\nabla^{2}\bm{u}\bm{\cdot}\bm{n}.
\end{eqnarray}
The wall-normal component of~\eqref{eq9p12} gives
\begin{equation}\label{eq9p17}
	\nu\nabla^{2}\bm{u}\bm{\cdot}\bm{n}=\nu(\nabla_{\pi}^{2}\bm{u})_{n}-\nu {K}\vartheta_{n}+\nu\left[\frac{\partial^2{u}_{n}}{\partial n^2}\right]_{\Sigma}.
\end{equation}
Substituting \eqref{eq9p17} into \eqref{eq9p16}, and using \eqref{eq9p1} and~\eqref{eq9p15}, we have
\begin{eqnarray}\label{eq9p18}
\left[\frac{\partial\hat{P}}{\partial n}\right]_{\Sigma}&=&-({a}_{n}-f_{n})+\nu(\nabla_{\pi}^{2}\bm{u})_{n}-\nu {K}\vartheta_{n}-\nu\left(\left[\frac{\partial\vartheta}{\partial n}\right]_{\Sigma}-\left[\frac{\partial^2{u}_{n}}{\partial n^2}\right]_{\Sigma}\right)\nonumber\\
&=&-({a}_{n}-f_{n})-\frac{1}{2}\nu\bm{\nabla}_{\pi}\bm{\cdot}\bm{\xi}_{R}+\nu\bm{K}\bm{:}\bm{\nabla}_{\pi}\bm{u}-\nu K \vartheta_{n}\nonumber\\
& &-\nu\bm{K}\bm{:}\bm{\nabla}_{\pi}\bm{u}-\nu\bm{\nabla}_{\pi}\bm{\cdot}\left(\frac{1}{2}\bm{\xi}_{R}+\bm{\xi}_{S}\right)+\nu K \vartheta_{n}\nonumber\\
&=&-(a_{n}-f_{n})-\nu\bm{\nabla}_{\pi}\bm{\cdot}\bm{\xi}.
\end{eqnarray}
 Equation \eqref{bpf} can also be directly derived from \eqref{NS2} on $\bm{\Sigma}$. The same result is obtained via different approaches to proof. Also, it is obvious from~\eqref{Xie_id1} that 
 \begin{equation}
 \bm{\nabla}_{\pi}\bm{\cdot}\bm{\xi}=(\bm{n}\times\bm{\nabla}_{\pi})\bm{\cdot}\bm{\omega}=(\bm{n}\times\bm{\nabla})\bm{\cdot}\bm{\omega}=\bm{n}\bm{\cdot}(\bm{\nabla}\times\bm{\omega}).
 \end{equation}
 This completes the proof.
\end{proof}
\begin{remark}
On a stationary boundary $\bm{\Sigma}$,~\eqref{bpf} reduces to
\begin{eqnarray}
	\left[\frac{\partial{P}}{\partial n}\right]_{\Sigma}=\rho{f}_{n}-\bm{\nabla}_{\pi}\bm{\cdot}\bm{\tau},
\end{eqnarray}
where the right-hand side is determined by the wall-normal component of the body force $(\rho{f}_{n})$, and the skin-friction divergence $(\bm{\nabla}_{\pi}\bm{\cdot}\bm{\tau})$. Assuming $f_{n}=0$, we consider a point $\mathcal{P}\in\bm{\Sigma}(\mathscr{D}_{\bm x})$. Let $S\subset\bm{\Sigma}(\mathscr{D}_{\bm x})$ be a small open neighborhood of $\mathcal{P}$, bounded by the closed circuit $L$, with the unit tangent vector $\bm{t}$ (oriented counterclockwise when viewed from the tip of $\bm{n}$). An application of the mean-value theorem together with the generalized Stokes formula gives
\begin{eqnarray}
	\left[\frac{\partial{P}}{\partial n}\right]_{\mathcal{P}}&=&-\lim\limits_{{\rm Area}(S)\rightarrow 0}\frac{1}{{\rm Area}(S)}\int_{S}\bm{\nabla}_{\pi}\bm{\cdot}\bm{\tau}dS\nonumber\\
	&=&-\lim\limits_{{\rm Area}(S)\rightarrow 0}\frac{1}{{\rm Area}(S)}\oint_{L}(\bm{t}\times\bm{n})\bm{\cdot}\bm{\tau}dl.
\end{eqnarray}
Thus, the sign of $\partial_{n}P$ at $\mathcal{P}$ is directly tied to the flux of $\bm{\tau}$ across $L$. The sources and sinks of the $\bm{\tau}$-field  accordingly indicate the nature of the normal pressure gradient.
\end{remark}

\section{Concluding remarks}\label{conclusion}
\begin{itemize}
\item By adopting a differential-geometric approach, we propose a rigorous theory for the intrinsic decompositions of boundary vorticity flux (BVF) and boundary enstrophy flux (BEF) in a compressible viscous flow interacting with an arbitrarily moving and deforming boundary. Moving beyond the single-vorticity paradigm in fluid mechanics, the theory transcends the classical framework of boundary vorticity dynamics originally proposed by~\citet{Lighthill1963}, providing general mathematical measures that quantify the strengths of boundary sources responsible for the rigid-rotation and spin modes inherent in vorticity.
\item The proposed framework is applicable to both the Lighthill-Panton-Wu and Lyman-Huggins interpretations of boundary vorticity and enstrophy dynamics. In contrast to existing vorticity dynamics theories, the theoretical structure in the present exposition allows for a clear discrimination and connection between kinematics (Theorems~\ref{SLC1}~--~\ref{thm_new2},~\ref{thm4},~\ref{thm5}) and dynamics (Theorems~\ref{thm7}~--~\ref{theorem13},~\ref{app_theorem}) on a deformable boundary.
\item The theory rigorously identifies a complete set of boundary sources for the rigid-rotation and spin modes, as well as for distinct enstrophy constituents. The derived expressions explicitly reveal the interaction among external force fields, surface geometry, surface kinematics, as well as both longitudinal and transverse physical processes on a deformable boundary.
\item Despite the high generality of arbitrary boundary motion and deformation, the results demonstrate that introducing a conjugate curvature tensor pair effectively absorbs these complexities into a unified algebraic structure, thereby yielding compact expressions for all boundary sources, manifesting as quadratic-form-type interactions between vorticity modes and surface geometry. It is emphasized that the effects of Gaussian and mean curvatures are intrinsically embedded within these source terms.
\item The proposed theory is expected to provide a valuable tool for (i) global surface flow diagnostics when integrated with experimental measurements and numerical simulations, (ii) the rational design of boundary flow control strategies, and (iii) advancing the physical understanding of formation mechanisms underlying near-wall coherent structures and flow-induced noise, in fluid mechanics.
\end{itemize}
\section*{Declarations}
\begin{itemize}
	\item \textbf{Funding} This work was funded by the National Natural Science Foundation of China (Grant No. 12402262).
	\item \textbf{Conflict of interest} The authors have no conflict of interest to declare that are
	relevant to the content of this article.
	\item \textbf{Data availability} No data was generated for the purposes of this research.
\end{itemize}

\backmatter


\begin{appendices}




\end{appendices}


\bibliography{sn-bibliography}

\end{document}